\documentclass[review]{elsarticle}
\usepackage[a4paper]{geometry}

\usepackage{lineno,hyperref}

\journal{Journal of Computational Physics}

\usepackage{placeins}
\usepackage{subcaption}
\usepackage{tikz}
\usetikzlibrary{snakes,arrows,shapes}
\usepackage[normalem]{ulem}

\usepackage{amsmath}
\usepackage{amssymb}
\usepackage{latexsym}

\definecolor{dotedge}{RGB}{168, 0, 54}
\definecolor{dotfill}{RGB}{254, 254, 206}
\definecolor{mpl_blue}{RGB}{31, 119, 180}
\definecolor{mpl_orange}{RGB}{255, 127, 14}
\definecolor{mpl_green}{RGB}{44, 160, 44}

\usepackage{amsthm}
\theoremstyle{definition}

\theoremstyle{plain}
\newtheorem{theorem}{Theorem}

\newtheorem{lemma}{Lemma}

\newcommand{\ud}{\mathrm{d}}
\newcommand{\pdiff}[2]{\frac{\partial{#1}}{\partial{#2}}}

\newcommand{\tdiff}[2]{\frac{\ud{#1}}{\ud{#2}}}
\newcommand{\tdif}[1]{\frac{\ud}{\ud{#1}}}

\renewcommand{\emph}[1]{\underline{#1}}

\begin{document}

\begin{frontmatter}

\title{Combined State and Parameter Estimation in Level-Set Methods}

\author[Cambridge]{Hans Yu}
\author[Cambridge]{Matthew P.\ Juniper}
\author[Cambridge,Munich]{Luca Magri\corref{mycorrespondingauthor}}
\cortext[mycorrespondingauthor]{Corresponding author}
\ead{lm547@cam.ac.uk}

\address[Cambridge]{Department of Engineering, University of Cambridge, Trumpington Street, Cambridge CB2 1PZ, UK}
\address[Munich]{Institute for Advanced Study, Technical University of Munich, Lichtenbergstrasse 2a, 85748 Garching, Germany (visiting fellowship)}

\begin{abstract}
%\hy{HY's abstract.}
%A data-driven approach to level-set methods based on data assimilation is introduced.
%Data assimilation is treated as a problem in Bayesian inference, which relies on a probabilistic formulation.
%This is in contrast to optimization problems, which rely on variational formulations.
%Efficient implementations of the Kalman filter and smoother are derived from the standpoint of Bayesian inference.
%The applications are state estimation, parameter estimation and uncertainty quantification of physical models.
%The G-equation data-assimilation methodology is fully verified for various canonical level-set test cases.
%Verification refers to both the numerical performance of the G-equation solver and the statistical performance of the ensemble Kalman filter and smoother.
%Practical issues such as covariance collapse, numerical vs.\ sampling errors and the square-root algorithm are discussed.

Reduced-order models based on level-set methods are widely used tools to qualitatively capture and track the nonlinear dynamics of an interface.
%In this paper, we enhance such a level-set model with a statistical learning technique in order to make the model quantitatively predictive. The statistical learning method is Bayesian, so the uncertainty of the outputs is naturally included in this framework.
The aim of this paper is to develop a physics-informed, data-driven, statistically rigorous learning algorithm for state and parameter estimation with level-set methods.
A Bayesian approach based on data assimilation is introduced.
Data assimilation is enabled by the ensemble Kalman filter and smoother, which are used in their probabilistic formulations.
The level-set data assimilation framework is verified in one-dimensional and two-dimensional test cases, where state estimation, parameter estimation and uncertainty quantification are performed. The statistical performance of the proposed ensemble Kalman filter and smoother is quantified by twin experiments.  In the twin experiments, the combined state and parameter estimation fully recovers the reference solution, which validates the proposed algorithm.
The level-set data assimilation framework is then applied to the prediction of the nonlinear dynamics of a forced premixed flame, which exhibits the formation of sharp cusps and intricate topological changes, such as pinch-off events.
The proposed physics-informed statistical learning algorithm opens up new possibilities for making reduced-order models of interfaces quantitatively predictive, any time that reference data is available.
\end{abstract}

\end{frontmatter}

%------------------------------------------------------------------------------

\section{Introduction}
\label{sec:intro}

A number of problems in computational physics involve the motion of interfaces, e.g.\ semiconductor manufacturing, multi-phase flows, crystal growth, groundwater flow, computer vision, grid generation and seismology~\cite{Sethian2001}.
In general, there are two approaches to calculating the motion of an interface~\cite{Gibou2018}:
In front-tracking methods, the interface is parameterized and discretized so that one follows the motion of the whole interface by tracking a sufficient number of points on the interface.
In front-capturing methods, the interface is embedded into a function defined over the whole domain.
One example of front-capturing methods are the so-called level-set methods, where the interface is embedded into a strictly monotonic function~\cite{Osher1988}.
Both front-tracking and front-capturing methods have well-understood advantages, and hybrid methods exist to mitigate their respective disadvantages~\cite{Sethian2003}.
Nevertheless, it is worth mentioning that level-set methods provide a natural formulation for calculating the motion of an interface~\cite{Sethian2001}:
Level-set methods deal well with non-smooth features such as corners and cusps as well as topological merging and break-up.
Furthermore, level-set methods are easily extended from two to three and higher dimensions.
For more information on level-set methods, the reader is referred to many excellent expositions in the literature~\cite{Sethian1999a, Osher2001, Sethian2001, Sethian2003, Gibou2018}.

Despite the elegance of explaining physical phenomena by the motion of interfaces and calculating the motion of the interfaces by level-set methods, one has to remain aware that the assumption that manifolds are infinitely thin is often an asymptotic assumption.
This is particularly true in fluid mechanics, which is governed by conservation laws~\cite{Kollmann2010}.
%immiscible fluids, shock waves, premixed flames
%\item In premixed combustion, level-set methods are used in G-equation models \cite{Peters2000, Pitsch2005, Moureau2009}.
One example is the kinematics of premixed flames~\cite{Peters2000}:
Depending on the combustion regime, it may be assumed that a thin reactive-diffusive layer separates the burnt and unburnt gases.
While the laminar flame speed, at which the premixed flame propagates from the burnt into the unburnt region, is a well-defined thermo-chemical property of the fuel-air mixture in a one-dimensional flow, it more generally depends on the balance between heat conduction and mass diffusion, whose effects vary with flame stretch and curvature.
%Physical challenges in thermoacoustics:
%\begin{itemize}
%\item flame-burner interaction \cite{Cuquel2013, Mejia2015, Kraus2018}
%\item perturbation convection speed \hy{Lieuwen, Ghoniem}
%\end{itemize}
Moreover, the turbulent flame speed scales differently depending on the interactions between turbulence and combustion length scales.
%\item In particular, the G-equation model is able to explain the linear and non-linear dynamics of thermoacoustic oscillations in ducted premixed flames \cite{Dowling1999, Kashinath2014, Waugh2014, Orchini2016, Semlitsch2017}.
Despite the potential quantitative shortcomings in describing the kinematics of a premixed flame as the motion of an interface, it has been shown that this model successfully explains the linear and nonlinear dynamics observed in thermoacoustic instabilities of ducted premixed flames, which are relevant to the design of combustion chambers in jet and rocket engines~\cite{Fleifil1996, Dowling1999, Kashinath2014, Waugh2014}.
The discrepancies which arise when comparing to more faithful simulations or experiments are usually attributed to the unpredictable nature of turbulent flow or to uncertainties in the model and its parameters.
% Compare to flame transfer/describing functions.
Thus, it is relevant to assess the ability of a qualitative, physics-informed model to make quantitative, time-accurate predictions.
The aim of this study is to develop a data-driven, statistically rigorous framework for state and parameter estimation in models using level-set methods, which has not been done before, and to apply it to an oscillating flame.

Inference over interfaces based on level-set methods has been performed, e.g.\ in the context of shape optimization, either by directly taking the functional derivative of the objective functional with respect to the level-set function~\cite{Zhao1996, Osher2001a}, or, equivalently, by embedding the level-set method into shape calculus~\cite{Sethian2000a, Allaire2004, Chantalat2009}.
As an alternative to these variational approaches, the framework for state and parameter estimation in this study is based on data assimilation~\cite{Evensen2009}.
Data assimilation finds the statistically optimal combination of model predictions and observations.
It combines concepts from control theory, probability theory and dynamic programming~\cite{Jazwinski2007, Gelb1974, Stengel1994}.
The data assimilation technique used in this study is the ensemble Kalman filter \cite{Evensen1994, Burgers1998}.
In the ensemble Kalman filter, a Monte-Carlo approach is used to represent the necessary statistics at every timestep~\cite{Doucet2001}, which makes it a computationally efficient technique in terms of storage requirements.
Compared to other data assimilation techniques based on the Kalman filter, e.g.\ the extended Kalman filter~\cite{Miller1994, Rozier2007}, the ensemble Kalman filter is found to be particularly robust with respect to larger nonlinearities~\cite{Evensen2009}.
This is relevant for level-set methods, due to strongly nonlinear events such as cusp formation, topological merging and break-up.
A practical advantage of the ensemble Kalman filter is its non-intrusive implementation with little effort required for its parallelization.
The ensemble Kalman filter has been successfully applied to a number of problems in fluid mechanics:
turbulent near-wall flow~\cite{Colburn2011};
transonic flows around airfoils and wings~\cite{Kato2015};
viscous flow around a cylinder~\cite{Mons2016};
model uncertainties in Reynolds-averaged Navier-Stokes (RANS) equations~\cite{Xiao2016};
vortex models of separated flow~\cite{Darakananda2018a};
and extinction and reignition dynamics in turbulent non-premixed combustion~\cite{Labahn2018}.

In this study, the ensemble Kalman filter is combined with a narrow-band level-set method and a fast marching method to form a computationally efficient level-set data assimilation framework~\cite{Peng1999, Sethian1996}.
In analogy to the distinction between front-tracking and front-capturing methods in describing the motion of an interface, data can also be assimilated according to various paradigms, e.g.\ as demonstrated by~\cite{Moreno2007}, \cite{Rochoux2013} and \cite{Gao2017}.
For the various paradigms, it is not a priori clear which will yield superior results.
For this reason, the theory and the derivations behind our level-set data assimilation framework are worked out in detail.
The paper is structured as follows:
In Section~\ref{sec:da}, the ensemble Kalman filter and the ensemble Kalman smoother are derived within the context of Bayesian inference.
The estimation problems in data assimilation and parameter estimation are then formulated in terms of probability distributions.
In Section~\ref{sec:ls}, various formulations for the laws of motion are discussed, including the Hamilton-Jacobi equation.
It is shown that the solutions to the Hamilton-Jacobi equation, the so-called generating functions, form a natural state space for data assimilation.
The level-set data assimilation framework based on the Hamilton-Jacobi equation is verified in one-dimensional and two-dimensional examples.
In Section~\ref{sec:flame}, statistical inference using the level-set data assimilation framework is demonstrated on the nonlinear dynamics of a ducted premixed flame.
The results and insights of this study are summarized in the conclusions~(Section~\ref{sec:conclusion}).

\FloatBarrier

%------------------------------------------------------------------------------

\section{Data assimilation and parameter estimation}
\label{sec:da}

%\hy{@MPJ @LM
%This section did not turn out as concise as I hoped for.
%The message, in my opinion, is very clear:
%There are so many ways to do statistical inference ('learning').
%Variational vs probabilistic; sequential vs batch-wise; filtering vs smoothing; assume linearity; assume normality; state estimation vs parameter estimation vs combined state and parameter estimation; etc.
%I believe that, without the Bayesian formalism, it is nigh impossible to have a discussion about whether people are solving fundamentally different inference problems (e.g. sequential filtering vs batch-wise filtering), or whether it is 'just' their implementations/approximations which are different (i.e., it does not really matter whether you use Gaussian processes or neural networks as long as you are emulating the same joint/marginal probability distributions).
%This is the reason why I put so much effort into this section.
%}
%\lmbf{I did enjoy this section. You did very good work. We will trim it down / reword it later on; don't worry for the time being.}

%This section gives a brief introduction to statistical inference, particularly Bayesian inference.
Data assimilation and parameter estimation are here treated as problems in statistical inference.
Statistical inference quantifies the degree of belief (or confidence) in a physical model, the parameters that it receives and the states that it predicts.
Statistical inference follows a probabilistic formulation, which provides precise definitions for the different tasks addressed in inference, which include filtering, smoothing and prediction.
Probabilistic formulations for statistical inference include frequentist inference and Bayesian inference.
In frequentist inference, an error functional is defined to measure the statistical distance between a candidate solution and the available data.
%Most of the times, the error functional is conveniently a sum of squares, and subject to Tikhonov regularization to make the problem less ill-posed.
%\lmbf{Recall what this is and why it works.}
%Although both inference frameworks may give equivalent results under certain conditions \lmbf{citation}, all derivations here are based on Bayes' rule (\ref{app:bayes}).
%\lmbf{The theorem has to be reported in the text here because it is central.}
Inference becomes an optimization problem to minimize this error functional.
In Bayesian inference, existing knowledge is quantified in the form of a probability distribution over candidate solutions.
%Formulation of the prior requires to make assumptions explicit, which adds to the transparency of the procedure.
When data becomes available, the probability distribution is updated, effectively combining the existing knowledge with the data.
For normal distributions under linear dynamics, both formulations of statistical inference give equivalent results~\cite{Evensen2009}.
Under these circumstances, frequentist inference may be considered more accessible because it allows the application of familiar tools from convex optimization~\cite{Boyd2004}.
Nevertheless, Cox's axioms demonstrate that probability theory, as used in Bayesian inference, gives a natural formulation for inference in the general case~\cite{Jaynes2003}.
All definitions and derivations in this section are given in terms of Bayesian inference (\ref{app:bayes}).
For more details, the reader may refer to~\cite{Evensen2009} or~\cite{Sarkka2013}.
%\hy{Mention assumptions and uncertainties.}

%------------------------------------------------------------------------------

\subsection{Probabilistic state space model}
\label{sec:da:state_space}

The state of a system at a timestep $k$ is uniquely defined by the state vector~$x_k$.
%The state vector $x_k$ comprises all physical quantities necessary to uniquely specify a point in state space.
%It is assumed that the state vectors $x_k$ form a Markov chain~(Fig.~\ref{fig:da:state_space:markov_chain}).
%I.e., a state~$x_k$ only depends on its previous state~$x_{k-1}$.
The evolution of the state is governed by a physical model~$f$ and its parameters~$\theta$.
%The dynamical model is derived from the governing equations of the physical model discretized in time.
%Note that the state vectors are hidden variables, and as such not directly observable.
Thus, our belief in the state~$x_k$ depends exclusively on our belief in (i) the previous state~$x_{k-1}$ as well as our choice of (ii) the physical model~$f$ and (iii) its parameters~$\theta$.
At the same time, the state vector $x_k$ is compared to noisy observations $y_k$ through a measurement operator $M$.
%It is assumed that the physical model continuous in space and time may be reformulated in terms of a probabilistic state space model \cite{Sarkka2013}:
Formulated in terms of a probabilistic state space model, we obtain
\begin{equation}
x_k
= x_{k-1} + \int_{t_{k-1}}^{t_k}{f(x(t), \theta)\,\ud{t}}
= G(x_{k-1}, \theta)
\sim p(x_k \mid x_{k-1}, \theta, f) \quad ,
\label{eq:da:state_space1}
\end{equation}
\begin{equation}
y_k = M(x_k) \sim p(y_k \mid x_k) \quad .
\label{eq:da:state_space2}
\end{equation}
In brief, the transition from one state to the next is governed by the operator~$G$.
The operator~$G$ is derived from the physical model~$f$, and depends on the parameters~$\theta$.
The states~$x_k$ and the observations~$y_k$ are considered realizations inside their respective probabilistic state spaces.
The degrees of belief in each are subject to the conditional probability distributions~$p(x_k \mid x_{k-1}, \theta, f)$ and~$p(y_k \mid x_k)$ respectively.
The degree of belief in the state $x_k$ is conditional on the previous state~$x_{k-1}$ and the choices in the physical model~$f$ and its parameters~$\theta$.
The degree of belief in the observation~$y_k$ is conditional on the true state~$x_k$.
%While not strictly necessary, this discussion focuses on finite-dimensional state spaces.
%Finite-dimensional state spaces arise for example when a physical model continuous in space is discretized.

The rules of the probabilistic state space model may be formalized as follows~\cite{Sarkka2013}.
Firstly, we assume that the physical model may be described by a Markov chain~\cite{Doucet2001}.
This means that the belief in a state depends only on the belief in the previous state:
\begin{enumerate}
\item The future is independent of the past given the present:
\begin{equation}
p(x_k \mid x_{0:k-1}, y_{1:k-1}, \theta, f) = p(x_k \mid x_{k-1}, \theta, f) \quad ,
\label{eq:bayes:future}
\end{equation}
where~$x_{0:k-1}$ denotes the states at all timesteps from~0 to~$k-1$, and~$y_{1:k-1}$ denotes the observations at all timesteps from~1 to~$k-1$.
\item The past is independent of the future given the present:
\begin{equation}
p(x_k \mid x_{k+1:N}, y_{k+1:N}, \theta, f) = p(x_k \mid x_{k+1}, \theta, f) \quad ,
\label{eq:bayes:past}
\end{equation}
where~$x_{k+1:N}$ and~$y_{k+1:N}$ respectively denote the states and observations at all timesteps from~$k+1$ to the final timestep~$N$.
\end{enumerate}
Secondly, observations are assumed to be conditionally independent in time.
The probability of an observation depends only on the current state:
\begin{equation}
p(y_k \mid x_{0:N}, y_{1:k-1}, y_{k+1:N}, \theta, f) = p(y_k \mid x_k) \quad .
\end{equation}
In Fig.~\ref{fig:da:state_space:markov_chain}, the relationship between states and observations is shown as well as the roles of models, parameters and measurement.

\begin{figure}
\centering
\begin{tikzpicture}[>=latex,line join=bevel,]
\begin{scope}
  \pgfsetstrokecolor{black}
  \definecolor{strokecol}{rgb}{1.0,1.0,1.0};
  \pgfsetstrokecolor{strokecol}
  \definecolor{fillcol}{rgb}{1.0,1.0,1.0};
  \pgfsetfillcolor{fillcol}
  \filldraw (0.0bp,0.0bp) -- (0.0bp,89.0bp) -- (397.0bp,89.0bp) -- (397.0bp,0.0bp) -- cycle;
\end{scope}
\begin{scope}
  \pgfsetstrokecolor{black}
  \definecolor{strokecol}{rgb}{1.0,1.0,1.0};
  \pgfsetstrokecolor{strokecol}
  \definecolor{fillcol}{rgb}{1.0,1.0,1.0};
  \pgfsetfillcolor{fillcol}
  \filldraw (0.0bp,0.0bp) -- (0.0bp,89.0bp) -- (397.0bp,89.0bp) -- (397.0bp,0.0bp) -- cycle;
\end{scope}
\begin{scope}
  \pgfsetstrokecolor{black}
  \definecolor{strokecol}{rgb}{1.0,1.0,1.0};
  \pgfsetstrokecolor{strokecol}
  \definecolor{fillcol}{rgb}{1.0,1.0,1.0};
  \pgfsetfillcolor{fillcol}
  \filldraw (0.0bp,0.0bp) -- (0.0bp,89.0bp) -- (397.0bp,89.0bp) -- (397.0bp,0.0bp) -- cycle;
\end{scope}
  \node (bar) at (393.5bp,79.5bp) [draw=dotedge,fill=white,draw=none] {$$};
  \node (measurement0) at (62.5bp,44.5bp) [draw=black,fill=white,rectangle] {$M$};
  \node (measurement1) at (198.5bp,45.5bp) [draw=black,fill=white,rectangle] {$M$};
  \node (measurement2) at (334.5bp,44.5bp) [draw=black,fill=white,rectangle] {$M$};
  \node (prediction2) at (334.5bp,79.5bp) [draw=dotedge,fill=dotfill,ellipse] {$x_{k+1}$};
  \node (prediction0) at (62.5bp,79.5bp) [draw=dotedge,fill=dotfill,ellipse] {$x_{k-1}$};
  \node (prediction1) at (198.5bp,79.5bp) [draw=dotedge,fill=dotfill,ellipse] {$\phantom{x}x_k\phantom{x}$};
  \node (model2) at (267.0bp,79.5bp) [draw=black,fill=white,rectangle] {$G, \theta$};
  \node (observations2) at (334.5bp,9.5bp) [draw=dotedge,fill=dotfill,ellipse] {$y_{k+1}$};
  \node (observations0) at (62.5bp,9.5bp) [draw=dotedge,fill=dotfill,ellipse] {$y_{k-1}$};
  \node (observations1) at (198.5bp,10.5bp) [draw=dotedge,fill=dotfill,ellipse] {$\phantom{y}y_k\phantom{y}$};
  \node (foo) at (3.5bp,79.5bp) [draw=dotedge,fill=white,draw=none] {$$};
  \node (model1) at (130.0bp,79.5bp) [draw=black,fill=white,rectangle] {$G, \theta$};
  \draw [->] (prediction0) ..controls (90.089bp,79.5bp) and (99.477bp,79.5bp)  .. (model1);
  \draw [->,dashed] (foo) ..controls (12.019bp,79.5bp) and (22.314bp,79.5bp)  .. (prediction0);
  \draw [->] (measurement0) ..controls (62.5bp,34.682bp) and (62.5bp,32.017bp)  .. (observations0);
  \draw [->] (prediction1) ..controls (227.1bp,79.5bp) and (236.39bp,79.5bp)  .. (model2);
  \draw [->] (measurement1) ..controls (198.5bp,35.682bp) and (198.5bp,33.017bp)  .. (observations1);
  \draw [->] (model1) ..controls (149.26bp,79.5bp) and (158.57bp,79.5bp)  .. (prediction1);
  \draw [->,dashed] (prediction2) ..controls (362.48bp,79.5bp) and (372.08bp,79.5bp)  .. (bar);
  \draw [->] (measurement2) ..controls (334.5bp,34.682bp) and (334.5bp,32.017bp)  .. (observations2);
  \draw [->] (prediction1) ..controls (198.5bp,67.802bp) and (198.5bp,65.326bp)  .. (measurement1);
  \draw [->] (prediction2) ..controls (334.5bp,67.148bp) and (334.5bp,64.383bp)  .. (measurement2);
  \draw [->] (prediction0) ..controls (62.5bp,67.148bp) and (62.5bp,64.383bp)  .. (measurement0);
  \draw [->] (model2) ..controls (286.14bp,79.5bp) and (295.58bp,79.5bp)  .. (prediction2);
\end{tikzpicture}
\caption{
Probabilistic state space model as a Markov chain.
The operator~$G$ and its parameters~$\theta$ govern the transition from one state to the next.
The state vector~$x_k$ is related to the observations~$y_k$ through a measurement operator~$M$.
}
\label{fig:da:state_space:markov_chain}
\end{figure}
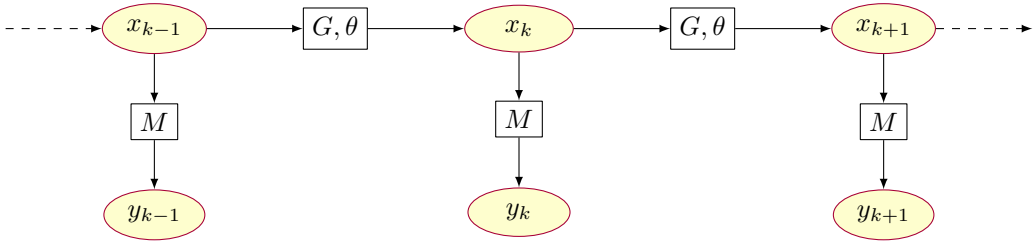

The goal of data assimilation is to find the joint probability distribution~$p(x_{0:N}, y_{1:N}, \theta, f)$.
The probabilistic state space spans all states~$x_k$ from timestep~$0$ to~$N$ and all observations~$y_k$ from timestep~$1$ to~$N$, as well as the physical model~$f$ and its parameters~$\theta$.
This joint probability distribution gives a complete statistical description.
In principle, all probability distributions of interest may be derived from this joint probability distribution~(Table~\ref{tab:da:state_space:data_assimilation}).
%In particular, the following conditional probability distributions are of interest:
%\begin{enumerate}
%\item State estimation ($p(x_{0:N} \mid y_{1:N}, \theta, f)$).
%Given a physical model and its parameters, what is our belief in a series of states?
%\item Parameter estimation ($p(\theta \mid y_{1:N}, f)$).
%Given a physical model, what is our belief in a set of parameters?
%\item Model comparison ($p(f \mid y_{1:N})$).
%Between two models, in which one do we believe more?
%\end{enumerate}
%Firstly, it is difficult to compute the joint posterior distribution due to its high dimensionality.
%Secondly, for the same reason, it would be computationally expensive to perform the necessary integrations for marginalization and Bayes' rule.
In practice, it is difficult to compute this probability distribution because of the high dimensionality of its probabilistic state space spanning multiple timesteps~\cite{Bellman2003}.
Therefore, state estimation focuses on the more direct computation of conditional probability distributions over a single timestep~(Table~\ref{tab:da:state_space:state_estimation}).
%For the state estimation, the computation of the joint probability distribution over states at all timesteps is still computationally demanding.
%Therefore, state estimation is more commonly concerned with the sequential computation of probability distributions over states at individual timesteps $k$.
%The tasks in state estimation are then distinguished as follows:
%\begin{enumerate}
%\item Filtering ($p(x_k \mid y_{1:k}, \theta, f)$).
%Given all the observations from the past and now, what is our belief in the current state?
%\item Smoothing ($p(x_k \mid y_{1:N}, \theta, f)$).
%Given all the observations from the past, the future and now, what is our belief in the current state?
%\item Prediction ($p(x_{N+1} \mid y_{1:N}, \theta, f)$).
%Given all the observations from the past and now, what is our belief in a subsequent state?
%\end{enumerate}

\begin{table}
\caption{Conditional probability distributions in data assimilation.}
\begin{center}
\begin{tabular}[hb]{p{3cm}|p{3cm}|p{6.5cm}}
Task & PDF & Description \\
\hline
State estimation & $p(x_{0:N} \mid y_{1:N}, \theta, f)$ & Given a physical model and its parameters, what is our belief in a series of states? \\
Parameter estimation & $p(\theta \mid y_{1:N}, f)$ & Given a physical model, what is our belief in a set of parameters? \\
Model comparison & $p(f \mid y_{1:N})$ & Between two physical models, in which one do we believe more?
\end{tabular}
\end{center}
\label{tab:da:state_space:data_assimilation}
\end{table}

%\begin{savenotes}
\begin{table}
\caption{Conditional probability distributions in state estimation.}
\begin{center}
\begin{tabular}[hb]{p{3cm}|p{3cm}|p{6.5cm}}
Task & PDF & Description \\
\hline
Filtering & $p(x_k \mid y_{1:k}, \theta, f)$ & Given all observations from the past and now, what is our belief in the current state? \\
%Smoothing\footnote{Smoothing in the statistical sense is unrelated to the analytical notion of smoothness and differentiability.} & $p(x_k \mid y_{1:N}, \theta, f)$ & Given all observations from the past, the future and now, what is our belief in the current state? \\
Smoothing & $p(x_k \mid y_{1:N}, \theta, f)$ & Given all observations from the past, the future and now, what is our belief in the current state? \\
Prediction & $p(x_{N+1} \mid y_{1:N}, \theta, f)$ & Given all observations from the past and now, what is our belief in a future state?
\end{tabular}
\end{center}
\label{tab:da:state_space:state_estimation}
\end{table}
%\end{savenotes}

%In the following subsections, the mathematics behind the Bayesian computations in data assimilation are outlined.
%Detailed derivations are omitted, but can be found in various monographs \cite{Evensen2009, Sarkka2013}.
%\hy{Mention older books; Stengel, Gelb, Jazwinski; maybe Bennett1992, Bennett2002.}
%More importantly, the assumptions introduced at every step of each derivation regarding the treatment of the nonlinearity in the physical models and the non-normality in the probability distributions are made explicit as they form the most critical aspects in practical data assimilation.

In the probabilistic formulation, data assimilation is easily extended to account for parameters.
In combined state and parameter estimation, the state is augmented by the parameters so that they become subject to the same inference~(Eq.~\eqref{eq:da:state_space1}):
\begin{equation}
\tilde{x}_k = \begin{pmatrix}x_k \\ \theta_k\end{pmatrix} \quad , \quad
\tilde{f}(\tilde{x}(t)) = \begin{pmatrix}f(x(t), \theta_k) \\ 0\end{pmatrix} \quad , \quad
\tilde{x}_k = \tilde{x}_{k-1} + \int_{t_{k-1}}^{t_k}{\tilde{f}(\tilde{x}(t))\,\ud{t}} \quad .
\label{tab:da:state_space:state_augmentation}
\end{equation}
The tasks in combined state and parameter estimation are given in Table~\ref{tab:da:param:state_estimation}.
The filtered and smoothed distributions in the parameters~$\theta_k$, $p(\theta_k \mid y_{1:k}, f)$ and~$p(\theta_k \mid y_{1:N}, f)$ respectively, are retrieved by marginalizing the states~$x_k$~(\ref{app:bayes}).
Note that the parameters $\theta_k$ are now time-dependent as the system traverses different regimes in state space.
This turns the strongly constrained parameter estimation into a weakly constrained combined state and parameter estimation~\cite{Evensen2009}.
Thus, the results of marginalizing the probability distributions in combined state and parameter estimation~(Table~\ref{tab:da:param:state_estimation}) are not strictly equivalent to the solutions of parameter estimation~(Table~\ref{tab:da:state_space:data_assimilation}).

\begin{table}
\caption{Conditional probability distributions in combined state and parameter estimation.}
\begin{center}
\begin{tabular}[hb]{p{3cm}|p{3cm}|p{6.5cm}}
Task & PDF & Description \\
\hline
Filtering & $p(x_k, \theta_k \mid y_{1:k}, f)$ & Given all observations from the past and now, what is our belief in the current state and set of parameters? \\
Smoothing & $p(x_k, \theta_k \mid y_{1:N}, f)$ & Given all observations from the past, the future and now, what is our belief in the current state and set of parameters?
%Prediction & $p(x_{N+1}, \theta \mid y_{1:N}, f)$ & Given all observations from the past and now, what is our belief in a future state and set of parameters?
\end{tabular}
\end{center}
\label{tab:da:param:state_estimation}
\end{table}

%------------------------------------------------------------------------------

\subsection{Bayesian filtering and smoothing}
\label{sec:da:bayes}

For the filtering problem, Bayes' rule gives
\begin{align}
p(x_k \mid y_{1:k}, \theta, f)
&= \frac{p(y_k \mid x_k, y_{1:k-1}, \theta, f)p(x_k \mid y_{1:k-1}, \theta, f)}{p(y_k \mid y_{1:k-1}, \theta, f)} \\
&= \frac{p(y_k \mid x_k)p(x_k \mid y_{1:k-1}, \theta, f)}{p(y_k \mid y_{1:k-1}, \theta, f)} \quad .
\end{align}
The prediction~$p(x_k \mid y_{1:k-1}, \theta, f)$ is given by the Chapman-Kolmogorov equation:
\begin{align}
p(x_k \mid y_{1:k-1}, \theta, f)
&= \int{p(x_k, x_{k-1} \mid y_{1:k-1}, \theta, f)\,\ud x_{k-1}} \\
&= \int{p(x_k \mid x_{k-1}, y_{1:k-1}, \theta, f)p(x_{k-1} \mid y_{1:k-1}, \theta, f)\,\ud x_{k-1}} \\
&= \int{p(x_k \mid x_{k-1}, \theta, f)p(x_{k-1} \mid y_{1:k-1}, \theta, f)\,\ud x_{k-1}} \quad .
\end{align}
The Chapman-Kolmogorov equation requires the inverse from the previous timestep~$k-1$.
In general, it is solved either numerically~\cite{Higham2001}, analytically~(Theorem~\ref{thm:da:kalman_filter}) or by a Monte-Carlo simulation~(Theorem~\ref{thm:da:ensemble_kalman_filter}).
%As a result, the Bayesian filter is sequential in nature.
The steps in the Bayesian filter may be summarized as follows:
\begin{theorem}[Bayesian filter] $ $
\label{thm:da:bayes_filter}
\begin{enumerate}
\item Prediction step:
\begin{equation}
p(x_k \mid y_{1:k-1}, \theta, f) = \int{p(x_k \mid x_{k-1}, \theta, f)p(x_{k-1} \mid y_{1:k-1}, \theta, f)\,\ud x_{k-1}} \quad .
\label{eq:da:bayes_filter:predict}
\end{equation}
\item Update step:
\begin{equation}
p(x_k \mid y_{1:k}, \theta, f) = \frac{p(y_k \mid x_k)p(x_k \mid y_{1:k-1}, \theta, f)}{p(y_k \mid y_{1:k-1}, \theta, f)} \quad .
\label{eq:da:bayes_filter:update}
\end{equation}
\end{enumerate}
\end{theorem}

At the timestep~$k=N$, the filtered and smoothed distributions are identical~(Tab.~\ref{tab:da:state_space:state_estimation}).
If the smoothed distribution at a timestep $k+1$ is known, the smoothed distribution at the previous timestep~$k$ is also known due to the Markov chain properties of the probabilistic state space model (Eq.~\eqref{eq:bayes:future}, \eqref{eq:bayes:past}).
This may be formalized as follows:
\begin{align}
p(x_k \mid y_{1:N}, \theta, f)
&= \int{p(x_k, x_{k+1} \mid y_{1:N}, \theta, f)\,\ud x_{k+1}} \\
&= \int{p(x_k \mid x_{k+1}, y_{1:N}, \theta, f)p(x_{k+1} \mid y_{1:N}, \theta, f)\,\ud x_{k+1}} \\
&= \int{p(x_k \mid x_{k+1}, y_{1:k}, \theta, f)p(x_{k+1} \mid y_{1:N}, \theta, f)\,\ud x_{k+1}} \quad .
\end{align}
$p(x_{k+1} \mid y_{1:N}, \theta, f)$ is the smoothed distribution from the subsequent timestep $k+1$.
$p(x_k \mid x_{k+1}, y_{1:k}, \theta, f)$ is computed via Bayes' rule:
\begin{align}
p(x_k \mid x_{k+1}, y_{1:k}, \theta, f)
&= \frac{p(x_{k+1} \mid x_k, y_{1:k}, \theta, f)p(x_k \mid y_{1:k}, \theta, f)}{p(x_{k+1} \mid y_{1:k}, \theta, f)} \\
&= \frac{p(x_{k+1} \mid x_k, \theta, f)p(x_k \mid y_{1:k}, \theta, f)}{p(x_{k+1} \mid y_{1:k}, \theta, f)} \quad .
\end{align}
Note that it involves the filtered distribution from the current timestep~$k$.
%Thus, the smoother 'improves' the pre-computed filtered distributions in reverse.
The steps in the Bayesian smoother may be summarized as follows:
\begin{theorem}[Bayesian smoother] $ $
\label{thm:da:bayes_smoother}
\begin{enumerate}
\item Forward sweep: Bayesian filter (Theorem~\ref{thm:da:bayes_filter}).
\item Backward step:
\begin{equation}
p(x_k \mid y_{1:N}, \theta, f) = p(x_k \mid y_{1:k}, \theta, f) \int{\frac{p(x_{k+1} \mid x_k, \theta, f)p(x_{k+1} \mid y_{1:N}, \theta, f)}{p(x_{k+1} \mid y_{1:k}, \theta, f)}\,\ud x_{k+1}} \quad .
\label{eq:da:bayes_smoother}
\end{equation}
\end{enumerate}
\end{theorem}

Note that both the Bayesian filter and smoother are sequential in nature.
In the Bayesian filter, the predicted distribution~$p(x_k \mid y_{1:k-1}, \theta, f)$ at a timestep~$k$ (Eq.~\eqref{eq:da:bayes_filter:predict}) mainly depends on the filtered distribution~$p(x_{k-1} \mid y_{1:k-1}, \theta, f)$ from the previous timestep~$k-1$~(Eq.~\eqref{eq:da:bayes_filter:update}).
In the Bayesian smoother, the smoothed distribution~$p(x_k \mid y_{1:N}, \theta, f)$ at a timestep~$k$ mainly depends on the smoothed distribution~$p(x_{k+1} \mid y_{1:N}, \theta, f)$ from the subsequent timestep~$k+1$~(Eq.~\eqref{eq:da:bayes_smoother}).
The existence of sequential algorithms for the computation of filtered and smoothed distributions significantly reduces the complexity of data assimilation~\cite{Bellman2003}.

%------------------------------------------------------------------------------

\subsection{The Kalman filter and the Rauch-Tung-Striebel smoother}
\label{sec:da:kalman}

%\hy{Derive ensemble Kalman smoother in the appendix.}

Two additional assumptions are introduced to make the computation of filtered and smoothed distributions feasible.
Firstly, the prior and the likelihood in the update step of the Bayesian filter are assumed to be normal~(Eq.~\eqref{eq:da:bayes_filter:update}):
\begin{equation}
p(x_k \mid y_{1:k-1}, \theta, f) = \mathcal{N}\left(x_k \mid \psi^f, C_{\psi\psi}^f\right) \quad ,
\label{eq:da:kalman:prior}
\end{equation}
\begin{equation}
p(y_k \mid x_k) = \mathcal{N}\left(y_k \mid Mx_k, C_{\epsilon\epsilon}\right) \quad ,
\label{eq:da:kalman:likelihood}
\end{equation}
where~$\mathcal{N}$ denotes a normal distribution with respective mean and covariance matrix.
The mean of the prior is denoted by~$\psi^f$, its covariance matrix by~$C_{\psi\psi}^f$, and the covariance matrix of the likelihood, also known as the observation error, by~$C_{\epsilon\epsilon}$.
From Eq.~\eqref{eq:da:kalman:prior} and~\eqref{eq:da:kalman:likelihood}, it follows that the filtered distribution~$p(x_k \mid y_{1:k}, \theta, f)$ is normal~(Eq.~\eqref{eq:da:bayes_filter:update}).
%For the likelihood $p(y_k \mid x_k)$, it is fair to assume that it is normal:
%Firstly, if a bias had been identifiable, the measurement functional $M$ could have been adjusted to account for the alleged bias.
%Secondly, for a sufficiently large number of observations, samples may be formed within the samples, which are subject to the central-limit theorem \cite{McKean2014}.
%The prior $p(x_k \mid y_{1:k-1}, \theta, f)$ is in general not a normal distribution.
%Treating it as such may lead to a statistically suboptimal estimation of the true state.
%Note that $\psi$, $C_{\psi\psi}$, $C_{\epsilon\epsilon}$ are in general time-dependent.
%Nevertheless, the timestep is not subscribed unless the omission would be misleading.
Secondly, the operator~$G$~(Eq.~\eqref{eq:da:state_space1}) in the prediction step~(Eq.~\eqref{eq:da:bayes_filter:predict}) is assumed to be linear in~$x$.
The result is the well-known Kalman filter~\cite{Kalman1960, Kalman1961}:

%In order to obtain the predicted distribution $p(x_k \mid y_{1:k-1}, \theta, f)$, a stochastic differential equation has to be solved in the general case \cite{Jazwinski2007}:
%\begin{equation}
%\ud{x} = f(x(t), \theta)\,\ud{t} + h(x(t))\,\ud{q(t)} \quad .
%\end{equation}
%The choice of $f$ and $\theta$ depends on the physical model.
%$h\,\ud{q}$ quantifies the unknown, epistemic uncertainty, namely due to model and parameter uncertainties \cite{Kiureghian2009}.
%\hy{There must be a better reference for aleatoric and epistemic uncertainty.}
%If $G$ is linear in $x$ (Eq.~\eqref{eq:da:state_space1}), the result is the well-known Kalman filter \cite{Kalman1960, Kalman1961}.
%In particular for a normal initial distribution $p(x_k \mid y_{1:k-1})$, it can be shown that all predicted distributions $p(x_k \mid y_{1:k-1}, \theta, f)$ are normal so that the Kalman filter gives the statistically optimal estimation of the true state at every timestep $k$.

\begin{theorem}[Kalman filter]
\label{thm:da:kalman_filter}
%\hy{Include equation for~$C_{\psi\psi}^f$?}
\begin{equation}
p(x_k \mid y_{1:k}, \theta, f) = \mathcal{N}\left(x_k \mid \psi^a, C_{\psi\psi}^a\right) \quad ,
\label{eq:da:kalman}
\end{equation}
\begin{equation}
\psi^a = \psi^f + \left(MC_{\psi\psi}^f\right)^T\left[C_{\epsilon\epsilon}+MC_{\psi\psi}^fM^T\right]^{-1}\left(y_k-M\psi^f\right) \quad ,
\label{eq:da:kalman:mean}
\end{equation}
\begin{equation}
C_{\psi\psi}^a
= C_{\psi\psi}^f - \left(MC_{\psi\psi}^f\right)^T\left[C_{\epsilon\epsilon}+MC_{\psi\psi}^fM^T\right]^{-1}\left(MC_{\psi\psi}^f\right) \quad ,
\label{eq:da:kalman:cov}
\end{equation}
\end{theorem}
where the superscript~$f$ denotes 'forecast'~(everything pertaining to the prediction), and the superscript~$a$ denotes 'analysis'~(everything pertaining to the update).
If the physical model is nonlinear, the prediction of the covariance matrix~$C_{\psi\psi}^f$ for the Kalman filter may be generalized in several ways.
In the extended Kalman filter, the predicted covariance matrix~$C_{\psi\psi}^f$ is computed by linearizing~$f$~\cite{Gelb1974}.
In strongly nonlinear dynamical systems, the predictions are found to be poor~\cite{Evensen2009}.
Higher-order extended Kalman filters are available~\cite{Miller1994}.
Nevertheless, drawbacks include their significant storage requirements, which increase exponentially with the order of approximation.
An alternative is the ensemble Kalman filter~\cite{Evensen1994, Burgers1998}.
Instead of a mean~$\psi$ and a covariance matrix~$C_{\psi\psi}$, a distribution is represented by a sample~$\left(\psi^j\right)_{j = 1 \dots n}$.
During the prediction step, the ensemble members~$\psi^j$ evolve in time independently.
Before the update step, the statistics may be recovered from the sample as follows:
\begin{equation}
\overline{\psi} \approx \frac{1}{n}\sum_{j=1}^n{\psi^j} \quad , \quad
\Psi = \begin{pmatrix}\psi^1-\overline{\psi} \quad , \quad \psi^2-\overline{\psi} \quad , \quad \cdots \quad , \quad \psi^n-\overline{\psi}\end{pmatrix} \quad , \quad
C_{\psi\psi}
%\approx C_{\psi\psi}^e
%= \frac{1}{n-1}\sum_{j=1}^n{\left(\left(\psi^f\right)_j-\overline{\psi^f}\right)\left(\left(\psi^f\right)_j-\overline{\psi^f}\right)^T} \quad ,
\approx \frac{1}{n-1}\Psi\Psi^T \quad .
\label{eq:da:ensemble}
\end{equation}
%\begin{equation}
%C_{\epsilon\epsilon}
%\approx C_{\epsilon\epsilon}^e
%= \frac{1}{n-1}\sum_{j=1}^n{\left(\left(\psi^f\right)_j-\overline{\psi^f}\right)\left(\left(\psi^f\right)_j-\overline{\psi^f}\right)^T} \quad .
%\end{equation}
The sample covariance matrix~$C_{\psi\psi}$ involves division by~$n-1$ instead of~$n$ in order to avoid a sample bias.

Various implementations of the ensemble Kalman filter exist, which differ in the update step.
%In the straightforward implementation of the ensemble Kalman filter \cite{Evensen1994}, the predicted covariance matrix $C_{\psi\psi}$ is formed (Eq.~\eqref{eq:da:ensemble}), and every ensemble member is individually updated (Eq.~\eqref{eq:da:kalman:mean}).
In the straightforward implementation of the ensemble Kalman filter~\cite{Evensen1994}, every ensemble member is individually updated~(Eq.~\eqref{eq:da:kalman:mean}).
It can be shown that the observations must be randomly perturbed in order to guarantee a statistically consistent analysis scheme~\cite{Burgers1998}.
In order to avoid the introduction of randomly generated numbers, the square-root filter is used here~\cite{Whitaker2002}.
The square-root filter belongs to a larger family of ensemble Kalman filters called ensemble-transform Kalman filters~\cite{Tippett2003, Livings2008}.
Unlike the straightforward implementation of the ensemble Kalman filter, the mean and the deviations of the ensemble members are updated.
This requires the singular value decomposition of a symmetric, positive, semi-definite matrix ($V \Sigma V^T$, where $V$ is orthonormal and $\Sigma$ diagonal), but no spurious errors due to the random perturbation of the observations are introduced.

\begin{theorem}[Square-root filter]
\label{thm:da:ensemble_kalman_filter}
\begin{equation}
\left(\psi^{a}\right)^j = \overline{\psi^a} + \left(\Psi^a\right)_j \quad ,
\end{equation}
\begin{equation}
\overline{\psi^a} = \overline{\psi^f} + \Psi^f\left(M\Psi^f\right)^T\left[(n-1)C_{\epsilon\epsilon}+M\Psi^f\left(M\Psi^f\right)^T\right]^{-1}\left(y_k-M\overline{\psi^f}\right) \quad ,
\end{equation}
\begin{equation}
\Psi^a = \Psi^fV\left[\mathbb{I} - \Sigma\right]^\frac{1}{2}V^T \quad , \quad V \Sigma V^T = \left(M\Psi^f\right)^T\left[(N-1)C_{\epsilon\epsilon}+M\Psi^f\left(M\Psi^f\right)^T\right]^{-1}M\Psi^f \quad ,
\end{equation}
where $\mathbb{I}$ is the identity matrix.
\end{theorem}
Following Eq.~\eqref{eq:da:kalman:prior} and~\eqref{eq:da:kalman:likelihood}, the Bayesian smoother~(Theorem~\ref{thm:da:bayes_smoother}) becomes the Rauch-Tung-Striebel smoother, also known as Kalman smoother~(Theorem~\ref{thm:da:kalman_smoother}).
It can be shown that the smoothed distribution becomes a normal distribution with mean~$\psi^s$ and covariance matrix~$C_{\psi\psi}^s$.
In order to again avoid the shortcomings of assuming linearity, an ensemble Kalman smoother is presented here~(Theorem~\ref{thm:da:ensemble_kalman_smoother}).
%\hy{The derivation is given in~\ref{app:bayes}.}

\begin{theorem}[Rauch-Tung-Striebel smoother]
\label{thm:da:kalman_smoother}
\begin{equation}
p(x_k \mid y_{1:k}) = \mathcal{N}\left(x_k \mid \psi^s, C_{\psi\psi}^s\right) \quad ,
\end{equation}
\begin{equation}
\psi^s = \psi^a + \left(\big(C_{\psi\psi}^f\big)_{k+1}^{-1}GC_{\psi\psi}^a\right)^T\left[\psi_{k+1}^s-\psi_{k+1}^f\right] \quad ,
\end{equation}
\begin{equation}
C_{\psi\psi}^s
= C_{\psi\psi}^a - \left(\big(C_{\psi\psi}^f\big)_{k+1}^{-1}GC_{\psi\psi}^a\right)^T\left[\big(C_{\psi\psi}^f\big)_{k+1}-\big(C_{\psi\psi}^s\big)_{k+1}\right]\left(\big(C_{\psi\psi}^f\big)_{k+1}^{-1}GC_{\psi\psi}^a\right) \quad .
\end{equation}
\end{theorem}

\begin{theorem}[Ensemble Kalman smoother]
\label{thm:da:ensemble_kalman_smoother}
\begin{equation}
\left(\psi^s\right)^j = \left(\psi^a\right)^j + \Psi^a\big(\Psi^f\big)_{k+1}^{-1}\left[\left(\psi^s\right)_{k+1}^j-\left(\psi^f\right)_{k+1}^j\right] \quad .
\label{eq:da:ensemble_kalman_smoother}
\end{equation}
\end{theorem}

\section{Level-set methods}
\label{sec:ls}

When data assimilation is applied to the motion of an interface, the question arises as to what constitutes its probabilistic state space.
Refining the distinction between front-tracking and front-capturing methods, there are at least three ways to view the motion of an interface~\cite{Sethian2001, Sethian2003}:
\begin{description}
\item[Geometric view.]
The interface is parameterized and discretized so that one follows the motion of the whole interface by solving the laws of motion for a sufficient number of points on the interface.
\item[Set-theoretic view.]
A characteristic function is defined over the whole domain.
The characteristic function assumes one of two values, depending on whether the point at the location in question is inside or outside the region enclosed by the interface.
\item[Analytic view~(also known as ``analysis view''~\cite{Sethian2003}).]
A level-set function is defined over the whole domain.
The interface is reconstructed by identifying the position of a particular level set.
\end{description}
Beyond the description of motion, these three views also pertain to the assimilation of data.
In fact, all three views have been employed in earlier level-set data assimilation frameworks:
While using a level-set method when solving for the motion of the interface, \cite{Rochoux2013} take each entry in the innovation vector~$y_k - M\psi^f$~(Eq.~\eqref{eq:da:kalman}) as the distance between one point on the predicted interface and one point on the observed interface.
This framework corresponds to the geometric view.
\cite{Gao2017} compute~$y_k$ and~$\psi^f$ in the innovation vector from a progress variable.
The progress variable assumes values between zero and one, with most intermediate values assumed near the interface.
This framework is an approximation of the set-theoretic view.
\cite{Moreno2007} compute~$y_k$ and~$\psi^f$ from a level-set function.
This framework corresponds to the analytic view, and comes closest to the level-set data assimilation framework presented in this study.
The motion of an interface depends on the values of the level-set function in the immediate environment of the interface, whereas the values of the level-set function are not unique away from the interface.
The level-set data assimilation framework presented in this study (i) introduces an additional constraint in agreement with the the analytic view, (ii) validates the choice of the constraint for two analytical test cases in one and two dimensions respectively, and (iii) addresses shortcomings in the aforementioned frameworks.

In order to construct the probabilistic state space, the derivation of the level-set method is revisited in this section.
%By taking the laws of motion, the governing equations for the analytic view are derived.
The centerpiece of this derivation is the Hamilton-Jacobi equation.
In theory, the level-set method may alternatively be derived as the transport of a passive scalar quantity~\cite[Chapter~2]{Peters2000}.
Although the value of the level-set function is well defined at the interface, the choice of values for the level-set function away from the interface is in general not unique, and would require an ad-hoc constraint.
It is demonstrated how the necessary constraint to the Hamilton-Jacobi equation naturally follows from the choice of phase space.
Finally, it is shown that the solutions to the Hamilton-Jacobi equation, the so-called generating functions, form the appropriate probabilistic state space for data assimilation.
Two complimentary level-set algorithms are combined to solve the Hamilton-Jacobi equation computationally efficiently: the narrow-band method~\cite{Peng1999} and the fast-marching method~\cite{Sethian1996}.
Along with the ensemble Kalman filter~(Theorem~\ref{thm:da:ensemble_kalman_filter}), the three algorithms form the backbone of our level-set data assimilation framework.

\subsection{Hamilton-Jacobi equation}
\label{sec:ls:hj}

The laws of motion of an interface shall be given by
\begin{equation}
\tdiff{r}{t} = u(r) \quad ,
\label{eq:ls:hj:newton}
\end{equation}
where~$r$ is the position of one point on the interface, and~$u$ is the velocity.
%For simplicity, we consider a hyperbolic problem where the magnitude of the velocity is assumed to be at most a function of the position $x$ \cite{Sethian1996}.
%\hy{Forest fire.}
In Hamiltonian mechanics, motion is described in phase space using generalized coordinates and momenta~\cite[Chapter~3]{Arnold1989}.
With interfaces in mind, a natural choice for the phase space is to use the position~$r$ as generalized coordinates and the normal vector~$n$ as generalized momenta.

%The position $x$ and the normal vector $n$ form a vector in phase space.
%In Hamiltonian mechanics, the position $x$ represents generalized coordinates, and the normal vector $n$ represents generalized momenta.
%For the Hamilton-Jacobi equation, we derive a generating function $G(x, t)$.
%With the motion of a surface and data assimilation in mind, we restrict ourselves to solutions with the following properties:
For the Hamilton-Jacobi equation, the generating function~$G(r, t)$ shall satisfy the following properties:
\begin{itemize}
\item The initial interface~$r(0)$ is given by a level set~$G(r(0), 0) = \mathrm{const}$.
\item The moving interface is identified as the level set~$G(r(t), t) = G(r(0), 0)$.
%\item At any given time, all level sets obey Huygens' principle with respect to the level set $G(x(t), t) = G(x(0), 0)$.
%Physically speaking, every point on the surface is interpreted as the source of a wave.
%As the wave propagates throughout the domain, it connects the point on the surface to every other point in the domain.
%For unity wave speed, Huygens' principle is equivalent to the eikonal equation $|\nabla G| = 1$.
%\hy{Find reference for Huygens' principle.}
\item Huygens' principle~\cite[Chapter~9]{Arnold1989} states that~$\nabla G = n$.
Given the choice of phase space, this is a necessary relationship between the generating function and the generalized coordinates and momenta.
\end{itemize}
%The first two properties imply the ability of the G-equation to capture fronts \cite{Osher1988}.
%%The third property may require more arbitration.
%The third property along with the second property implies $\nabla G = n$.
%This is a necessary relationship between the generating function and its generalized coordinates and momenta \cite{Landau1969}.
%Geometrically speaking, Huygens' principle gives an eikonal field \cite{Sethian1996}.
%This is highly relevant because data assimilation applied to the motion of a surface should involve the notion of distance.
%%Numerically speaking, maintaining a slope close to unity stabilizes the narrow-band method \cite{Sussman1994}.
%%The narrow-band method is discussed in Section~\ref{sec:ls:da}.
The first two properties are present in every front-capturing method~\cite{Peters2000}.
The third property, Huygens' principle, follows from Hamiltonian mechanics~(\ref{app:hj}).
Physically speaking, every point on the interface is interpreted as the source of a wave.
As the waves propagate through the domain, they connect each point on the interface to every other point in the domain.
Geometrically speaking, Huygens' principle gives an eikonal field~\cite{Sethian1999}.
%This notion of distance in the generating function is critical in its choice for the probabilistic state space.
%This notion of distance is critical in choosing the generating function to represent the state of an interface in the probabilistic state space.

It remains to show how the generating function translates into a level-set method.
The Lagrangian~$\mathcal{L}(r, \dot{r}, t)$ is given by~(\ref{app:hj}, Lemma~\ref{thm:hj:lagrange})
\begin{equation}
\mathcal{L}(r(t), \dot{r}(t), t) = 0 \quad ,
\label{eq:ls:hj:lagrange}
\end{equation}
\begin{equation}
\mathrm{s.t.} \quad G(r(t), t) - G(r(0), 0) = 0 \quad ,
\label{eq:ls:hj:hamilton1}
\end{equation}
\begin{equation}
\mathrm{s.t.} \quad n(t) \cdot n(t) - 1 = 0 \quad .
\label{eq:ls:hj:hamilton2}
\end{equation}
The Hamiltonian~$\mathcal{H}(r, n, t)$, subject to the same constraints as the Lagrangian~$\mathcal{L}(r, \dot{r}, t)$~(Eq.~\eqref{eq:ls:hj:hamilton1}, \eqref{eq:ls:hj:hamilton2}), is given by~(\ref{app:hj}, Lemma~\ref{thm:hj:hamilton})
\begin{equation}
\mathcal{H}(r(t), n(t), t) = u(r(t)) \cdot n(t) \quad .
\label{eq:ls:hj:hamilton}
\end{equation}
Finally, the Hamilton-Jacobi equation is given by~(\ref{app:hj}, Theorem~\ref{thm:hj:g})
\begin{equation}
\pdiff{G}{t} + u(r(t)) \cdot n(t) = 0 \quad .
\label{eq:ls:hj:g}
\end{equation}
The solution to the Hamilton-Jacobi equation, the generating function~$G(r, t)$, is defined over the whole domain.
As such, it represents the state of the interface in the probabilistic state space.
Eq.~\eqref{eq:ls:hj:g} is solved near the interface.
Away from the interface, the generating function is constrained by~(Eq.~\eqref{eq:ls:hj:hamilton2})
\begin{equation}
|\nabla G| = 1 \quad .
\label{eq:ls:hj:eikonal}
\end{equation}
Eq.~\eqref{eq:ls:hj:eikonal} is an eikonal equation~\cite{Sethian1999}.
The solutions to this eikonal equation are signed distance functions.

\subsection{Level-set data assimilation framework}
\label{sec:ls:da}

The level-set data assimilation framework presented in this study combines three algorithms:
(i)~the narrow-band method to solve the Hamilton-Jacobi equation near the interface~\cite{Peng1999}, (ii)~the fast-marching method to extend the eikonal field from the narrow band to the whole domain~\cite{Sethian1996}, and (iii)~the ensemble Kalman filter and smoother to assimilate data~\cite{Evensen2009}.
%The level-set filtering and smoothing frameworks are shown in Fig.~\ref{fig:ls:da}.
%The individual steps are outlined below.
%For details on implementation and numerical schemes, we refer to the cited publications.

\begin{figure}[ht]
\centering
\begin{subfigure}[ht]{0.45\textwidth}
\centering
\includegraphics[scale=0.6]{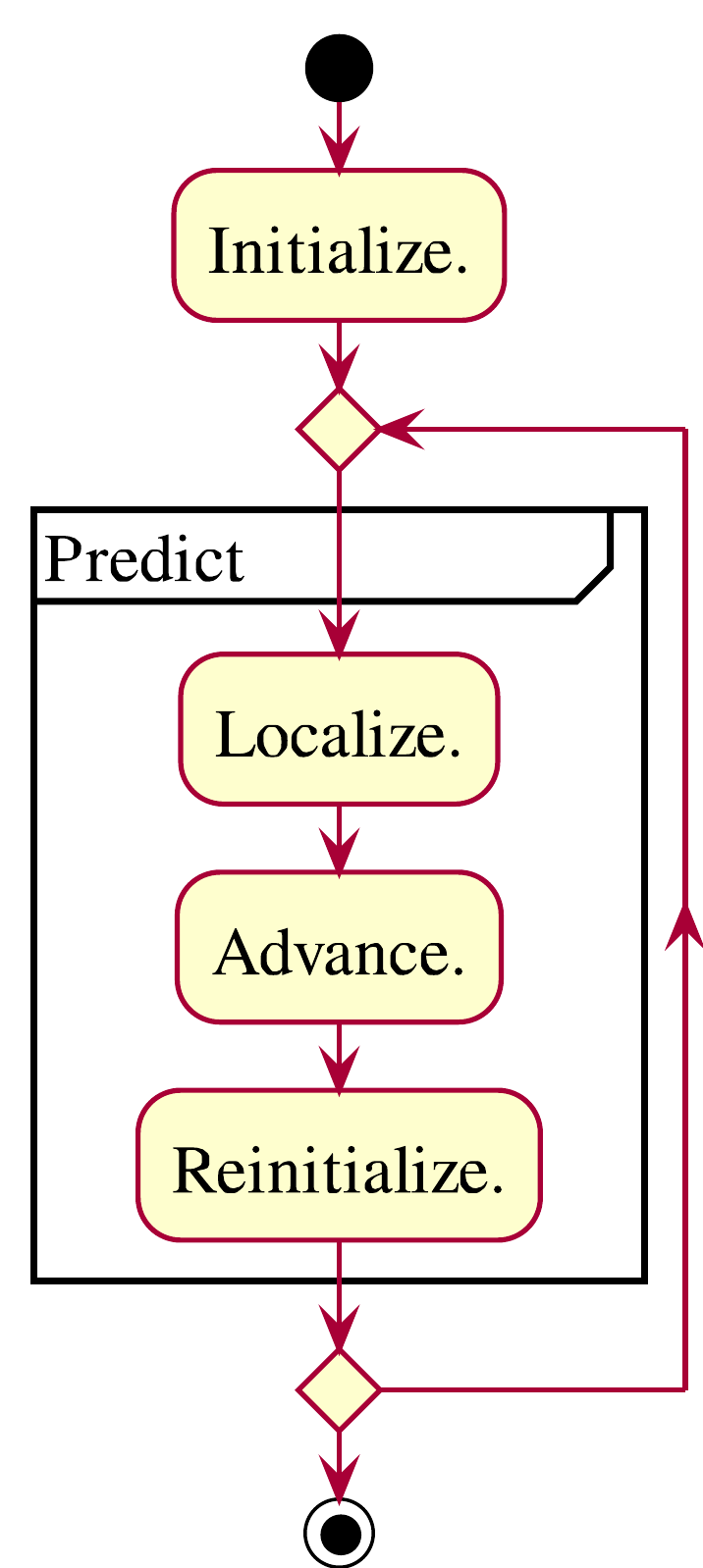}
\end{subfigure}
\begin{subfigure}[ht]{0.45\textwidth}
\centering
\includegraphics[scale=0.6]{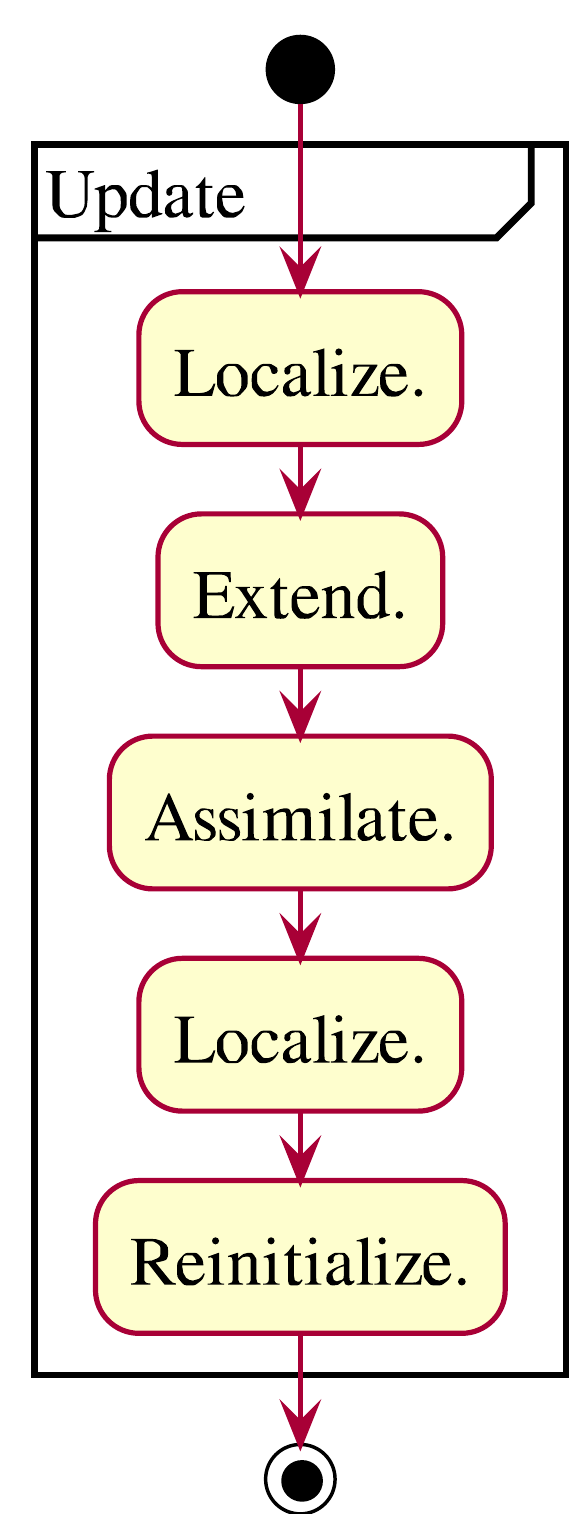}
\end{subfigure}
\caption{Activity diagrams of narrow-band method~(left) and level-set data assimilation~(right).}
%Explanations of individual activities are given in Section~\ref{sec:ls:da}.
%Note that our notion of 'Extend' differs from the one referred to in the narrow-band method \cite{Peng1999}.
%In the narrow-band method, extension is concerned with the magnitude near the surface of quantities only defined on the surface.
%In our level-set data assimilation framework, we generate an eikonal field over the whole domain using the fast-marching method \cite{Sethian1996}.}
\label{fig:ls:noda}
\end{figure}

The narrow-band method forms the backbone of the level-set data assimilation framework~\cite{Peng1999}.
It solves the Hamilton-Jacobi equation~(Eq.~\ref{eq:ls:hj:g}) in a narrow band near the interface.
%In the limit of an infinitely narrow band, the generating function would give exact information about the positions of points on the surface as well as the normal vectors at these points.
%Hence, a narrow band around the surface is sufficient to predict the position of the surface at the next timestep \cite{Osher1988}.
Using a narrow band reduces the computational cost by one order of magnitude compared to finding the generating function over the whole domain~\cite{Adalsteinsson1995}.
In a numerical simulation, the spatial resolution shall be denoted by~$\Delta x$ and the temporal resolution by~$\Delta t$.
The temporal resolution has to satisfy the Courant-Friedrich-Lewy condition~$u\Delta t < \Delta x$ at every point on the interface~\cite{Peng1999}.
The narrow-band method consists of the following steps~(Fig.~\ref{fig:ls:noda}, left):
\begin{description}
\item[Initialize.]
Before the first timestep, the signed distance field has to be known in a sufficiently wide neighborhood of the initial interface.
It may be obtained from either a previously computed generating function or exact knowledge of the position of the initial interface.
\item[Localize.]
At the timestep~$k$, a band~$T^k = \{r : G(r, t_k) < \gamma\Delta x\}$ is formed, where~$\gamma$ is a constant chosen according to the width of the narrow band.
A second, enveloping band~$N^k = \{r : G(r+r', t_k) < \gamma\Delta x \, , \, \exists |r'| < \Delta x\}$ is formed.
It follows from the Courant-Friedrich-Lewy condition that the band~$N^k$ contains the band~$T^{k+1}$ formed at the next timestep~$k+1$.
\item[Advance.]
The generating function over the band~$N^k$ is evolved by one timestep.
%In general, the speed $U$ is not only a function of position but also curvature and other local properties of the surface.
%Thus, the G-equation has to be treated as a parabolic Hamilton-Jacobi equation \cite{Osher1988}.
Note that the solution~$\widetilde{G}^{k+1}(r)$ is not the generating function at the next timestep~$k+1$.
While the position of the interface correctly coincides with the level set of the solution~$\widetilde{G}^{k+1}(r)$, Huygens' principle is no longer satisfied~\cite{Sussman1994}.
\item[Reinitialize.]
The following Hamilton-Jacobi equation is solved over the band~$N^k$ until steady state is reached~\cite{Sussman1994}:
\begin{equation}
\pdiff{\widetilde{G}}{\tilde{t}} + \mathrm{sgn}\left(\widetilde{G}\left(r, \tilde{t}\right)\right)\left(\left|\nabla\widetilde{G}\right|-1\right) = 0 \quad ,
\end{equation}
\begin{equation}
\widetilde{G}(r, 0) = \widetilde{G}^{k+1}(r) \quad .
\end{equation}
The steady-state solution~$\widetilde{G}(r, \tilde{t}\to\infty)$ gives the generating function~$G(r, t_{k+1})$ over the band~$N^k$ at the timestep~$k+1$.
Note that this Hamilton-Jacobi equation is a numerical continuation of the eikonal equation~(Eq.~\eqref{eq:ls:hj:eikonal}), where~$\widetilde{t}$ denotes a pseudo-time~\cite{Ascher1995}.
\end{description}

When observations are available, the prediction step is followed by an update step~(Fig.~\ref{fig:ls:noda}, right).
Like the prediction step, the update step involves localization and reinitialization so that the narrow band follows the position of the interface.
Instead of the advancement step, there is an extension step and an assimilation step:
\begin{description}
\item[Extend.]
While a narrow band is sufficient to capture the motion of an interface, data assimilation has to be performed in one probabilistic state space for all interfaces.
In a numerical simulation, the state comprises the values of the generating function at every grid point.
The fast-marching method extends the generating function from the narrow band to the whole domain, thus returning a well-defined state.
The complexity of the fast-marching method is~$\mathcal{O}(m^2\log(m))$ in two dimensions or~$\mathcal{O}(m^3\log(m))$ in three dimensions, where~$m$ denotes the number of grid points in one dimension~\cite{Sethian1996}.
While the generating function over the whole domain is relatively expensive to compute, the fast-marching method is only called when observations are available.
If combined with the narrow-band method (complexity~$\mathcal{O}(m)$ or~$\mathcal{O}(m^2)$~\cite{Peng1999}), this level-set method becomes affordable overall.
%Higher-order implementations of the fast-marching method exist \cite{Sethian2000}.
\item[Assimilate.]
The ensemble Kalman filter~(Theorem~\ref{thm:da:ensemble_kalman_filter}) and/or smoother~(Theorem~\ref{thm:da:ensemble_kalman_smoother}) are applied.
The assimilation step is followed by a localization step and a reinitialization step in order to prepare the narrow band for the next prediction step.
Note that the result of data assimilation is in general not an eikonal field.
While the band localized after data assimilation may in theory be too narrow, it can be shown that the slopes in the filtered solutions are in expectation less than or equal to unity~(\ref{app:ls}).
%\begin{equation}
%\mathbb{E}\left[\left|D\psi^a\right|\right] < 1 \quad ,
%\end{equation}
%where $D$ denotes a differentiation matrix.
%In the localization step following the assimilation step, the bands $T^{k,a}$ and $N^{k, a}$ are thus in expectation too wide.
%After the subsequent reinitialization step, the bands $T^{k+1}$ and $N^{k+1}$ are reduced to the correct width in the localization step at the beginning of the next timestep.
\end{description}

\begin{figure}[ht]
\centering
\begin{subfigure}[ht]{0.45\textwidth}
\centering
\includegraphics[scale=0.6]{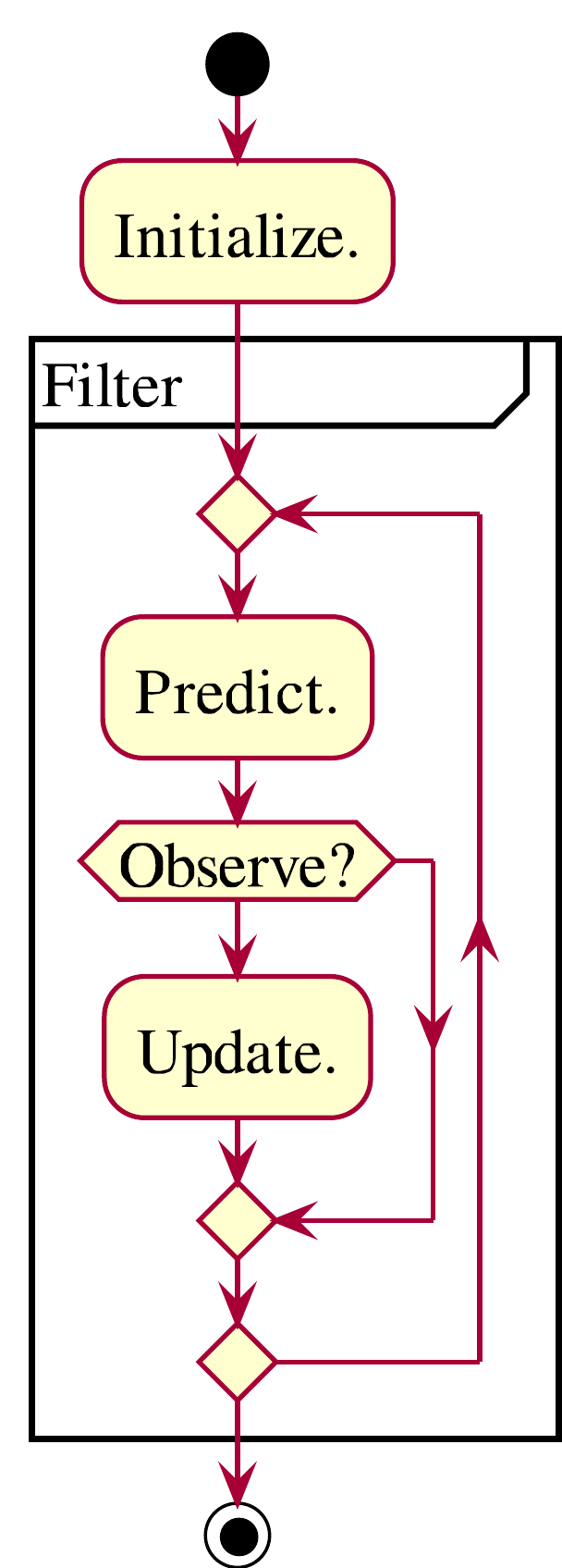}
\end{subfigure}
\begin{subfigure}[ht]{0.45\textwidth}
\centering
\includegraphics[scale=0.6]{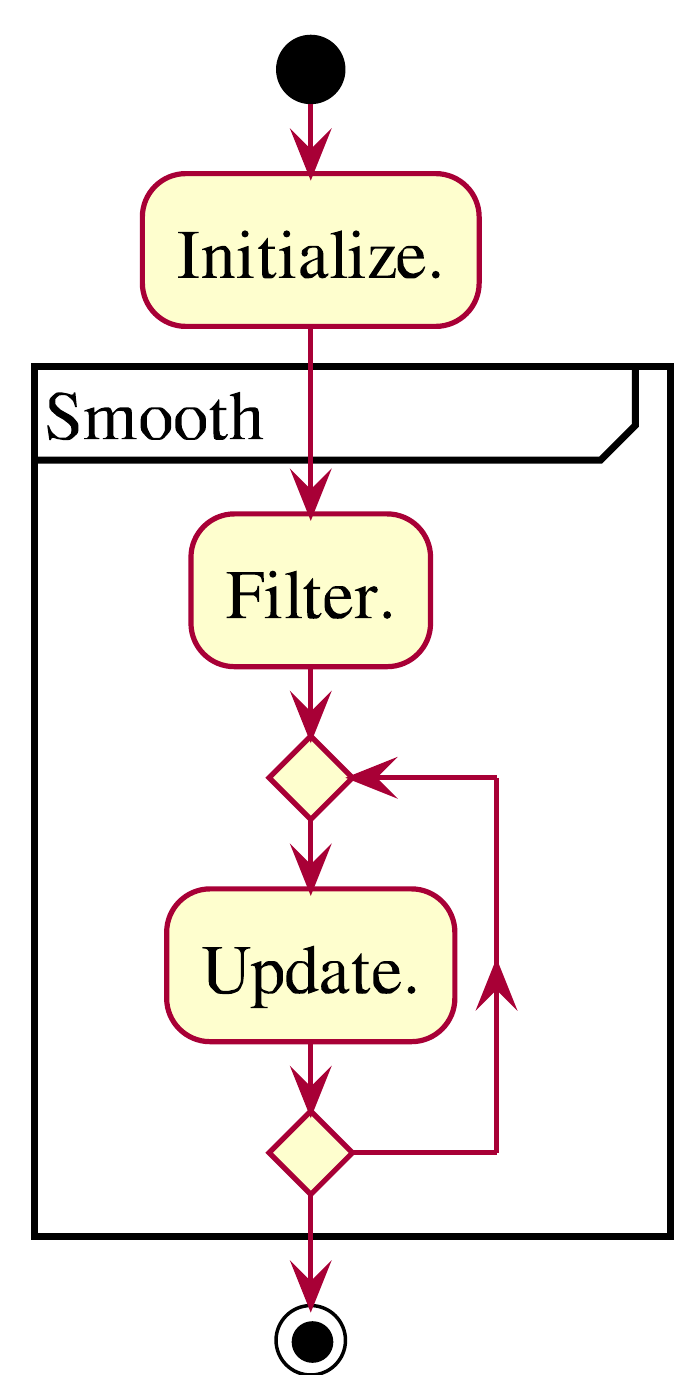}
\end{subfigure}
\caption{Activity diagrams of filtering~(left) and smoothing~(right) frameworks.
The prediction and update steps are defined in Fig.~\ref{fig:ls:noda}.
Note how the smoothing framework~(right) builds on top of the filtering framework~(left).}
\label{fig:ls:da}
\end{figure}

Combining the building blocks of the level-set method and data assimilation gives the level-set filtering and smoothing frameworks (Fig.~\ref{fig:ls:da}).
In the following subsections, the update step in the level-set filtering framework is verified using two analytical test cases:
In a one-dimensional test case, the level-set filtering framework, which represents the analytic view of data assimilation, is contrasted with the geometric and set-theoretic views of data assimilation.
It is shown that the analytic view of data assimilation works correctly, which sets the stage for its generalization to higher dimensions.
Afterwards, the level-set filtering framework is applied to a two-dimensional test case, and the effect of the observations on the shape of the interface is discussed.
It is shown how various parameters such as the number of observations and the observation error may affect the performance of the level-set data assimilation framework.
%The Kalman filter is here chosen over the square-root filter in order to study data assimilation without sampling errors.
%Similarly, the location of the surface is here continuous in order to avoid discretization errors.
%The derivation of the Kalman filter generalized to state spaces with functions as vectors is given in \ref{app:bayes}.

%------------------------------------------------------------------------------

\subsection{One-dimensional test case}
\label{sec:ls:ex1d}

The predicted and observed positions~$x$ and~$y$ of an interface in one dimension are respectively given by~(Fig.~\ref{fig:ls:ex1d})
\begin{equation}
x \sim \mathcal{N}(x \mid \mu_\psi, C_{\psi\psi}) \quad , \quad \mu_\psi = 0 \quad , \quad C_{\psi\psi} = 0.64 \quad ,
\label{eq:ls:ex1d:psi}
\end{equation}
\begin{equation}
y = x \sim \mathcal{N}(x \mid \mu_\epsilon, C_{\epsilon\epsilon}) \quad , \quad \mu_\epsilon = 1 \quad , \quad C_{\epsilon\epsilon} = 0.36 \quad .
\label{eq:ls:ex1d:epsilon}
\end{equation}
In one dimension, an interface reduces to a point.
The normal vector on the interface is unique up to its orientation.
Geometric data assimilation using the Kalman filter~(Theorem~\ref{thm:da:kalman_filter}) reduces to optimal interpolation of the positions~$\mu_\psi$ and~$\mu_\epsilon$~(Eq.~\eqref{eq:da:kalman:mean}):
\begin{equation}
x = \mu_\psi + \frac{C_{\psi\psi}}{C_{\epsilon\epsilon}+C_{\psi\psi}}\left(\mu_\epsilon-\mu_\psi\right) = 0.64 \quad .
\label{eq:ls:ex1d:phase_space}
\end{equation}
The filtered position~$x$ is closer to the observation than to the prediction in accordance with the variances~$C_{\psi\psi}$ and~$C_{\epsilon\epsilon}$~(Fig.~\ref{fig:ls:ex1d}, left).

For set-theoretic data assimilation, the characteristic function~$\chi$ is assumed to be bounded by~$\pm 1$ with the interface at~$\chi=0$:
%The surface is marked by a sharp gradient, e.g.\ in shock waves or premixed flames.
%In the limiting case, the progress variable resembles a Heaviside function:
\begin{equation}
\chi(x) =
\begin{cases}
-1 & x < c \\
0 & x = c \\
1 & x > c
\end{cases}
\quad , \quad c \sim \mathcal{N}(c \mid \mu_\psi, C_{\psi\psi}) \quad .
\end{equation}
The characteristic function~$\chi$ follows a Bernoulli distribution, which is the appropriate probability distribution for a random variable with two possible outcomes, at every location~$x$~\cite[Chapter~6]{MacKay2003}:
\begin{equation}
p\left(\chi(x)=1\right) = \Phi\left(\frac{x-\mu_\psi}{\sqrt{C_{\psi\psi}}}\right) \quad , \quad p\left(\chi(x)=-1\right) = 1 - \Phi\left(\frac{x-\mu_\psi}{\sqrt{C_{\psi\psi}}}\right) \quad ,
\end{equation}
where $\Phi$ denotes the cumulative density function of the normal distribution.
The characteristic function~$\chi$ does not follow a normal distribution whereas the Kalman filter requires the mean~$\overline{\chi}$ and the covariance function~$\mathcal{C}_{\chi\chi}$:
\begin{equation}
\overline{\chi}(x) = 2\Phi\left(\frac{x-\mu_\psi}{\sqrt{C_{\psi\psi}}}\right)-1 \quad , \quad \mathcal{C}_{\chi\chi}(x_1, x_2) = 4\Phi\left(\frac{\min(x_1, x_2)-\mu_\psi}{\sqrt{C_{\psi\psi}}}\right)\left(1-\Phi\left(\frac{\max(x_1, x_2)-\mu_\psi}{\sqrt{C_{\psi\psi}}}\right)\right) \quad ,
\end{equation}
Note that the mean~$\overline{\chi}$ is not a characteristic function.
%Also note that the covariance function $\mathcal{C}_{\chi\chi}$ is not differentiable at $x_1=x_2$.
Set-theoretic data assimilation using the Kalman filter~(Theorem~\ref{thm:da:kalman_filter}) gives
\begin{align}
\chi(x)
&= 2\Phi\left(\frac{x-\mu_\psi}{\sqrt{C_{\psi\psi}}}\right)-1 - \frac{4\Phi\left(\frac{\min(x, \mu_\epsilon)-\mu_\psi}{\sqrt{C_{\psi\psi}}}\right)\left(1-\Phi\left(\frac{\max(x, \mu_\epsilon)-\mu_\psi}{\sqrt{C_{\psi\psi}}}\right)\right)}{1+4\Phi\left(\frac{\mu_\epsilon-\mu_\psi}{\sqrt{C_{\psi\psi}}}\right)\left(1-\Phi\left(\frac{\mu_\epsilon-\mu_\psi}{\sqrt{C_{\psi\psi}}}\right)\right)}\left(2\Phi\left(\frac{\mu_\epsilon-\mu_\psi}{\sqrt{C_{\psi\psi}}}\right)-1\right) \\
&= 2\Phi\left(1.25x\right)-1 - \frac{4\Phi\left(1.25\min(x, 1)\right)\left(1-\Phi\left(1.25\max(x, 1)\right)\right)}{1+4\Phi\left(1.25\right)\left(1-\Phi\left(1.25\right)\right)}\left(2\Phi\left(1.25\right)-1\right) \quad .
\label{eq:ls:ex1d:progress_variable}
\end{align}
The interface is localized at~$x\approx0.139$~(Fig.~\ref{fig:ls:ex1d}, middle).
The filtered position~$x$ is closer to the prediction than to the observation although the prediction variance~$C_{\psi\psi}$ is larger than the observation variance~$C_{\epsilon\epsilon}$.
Note that the filtered position~$x$ is independent of the variance~$C_{\epsilon\epsilon}$.

For analytic data assimilation, the predicted and observed positions~$x$ and~$y$ are viewed as points on the level sets of the predicted and observed generating functions respectively.
It follows from Eq.~\eqref{eq:ls:ex1d:psi} that the predicted generating function is given by
\begin{equation}
G(x) \sim \mathcal{N}(G(x) \mid x-\mu_\psi, C_{\psi\psi}) \quad .
\end{equation}
The mean~$\overline{G}$ and the covariance function~$\mathcal{C}_{GG}$ are given by
\begin{equation}
\overline{G}(x) = x - \mu_\psi \quad , \quad \mathcal{C}_{GG}(x_1, x_2) \equiv C_{\psi\psi} \quad .
\end{equation}
Analytic data assimilation using the Kalman filter~(Theorem~\ref{thm:da:kalman_filter}) on the generating function gives
\begin{equation}
G(x) = x - \mu_\psi - \frac{C_{\psi\psi}}{C_{\epsilon\epsilon}+C_{\psi\psi}}\left(\mu_\epsilon-\mu_\psi\right) = x - 0.64 \quad .
\label{eq:ls:ex1d:generating_function}
\end{equation}
The filtered position~$x$ is given by the level set of the filtered generating function~$G$~(Eq.~\eqref{eq:ls:ex1d:generating_function}), which is identical to the result from geometric data assimilation~(Fig.~\ref{fig:ls:ex1d}, right).%~(Eq.~\eqref{eq:ls:ex1d:phase_space}).
%The level-set data assimilation framework is thus verified in one dimension~(Fig.~\ref{fig:ls:ex1d}, right).

\begin{figure}[ht]
\centering
\includegraphics[width=0.8\textwidth]{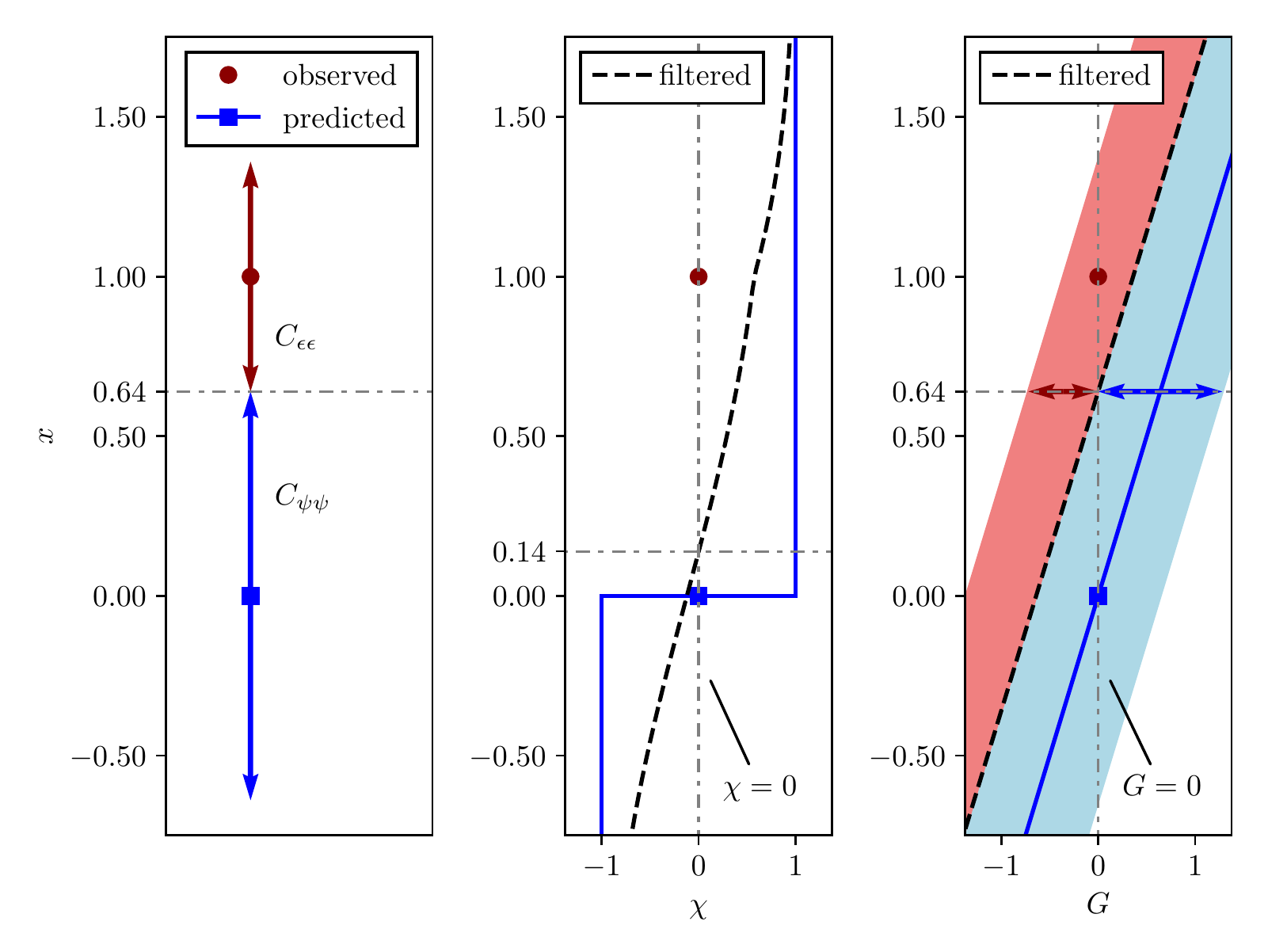}
\centering
\begin{tabular}{p{0.05\textwidth}p{0.225\textwidth}p{0.225\textwidth}p{0.225\textwidth}}
& \centering (a) geometric view & \centering (b) set-theoretic view & \centering (c) analytic view
\end{tabular}
\caption{Data assimilation from geometric~(left), set-theoretic~(middle) and analytic view~(right).
%The set-theoretic view relies on a characteristic function~$\chi$.
%The analytic view relies on a level-set function~G.
The blue square at~$x=0$ marks the predicted position of the interface.
The red circle at~$x=1$ marks the observed position of the interface.
The grey, dash-dotted lines give the coordinates of the filtered position of the interface.
In one dimension, data assimilation from the geometric view (left) reduces to optimal interpolation of the $x$ coordinates, weighted by $C_{\psi\psi}$ and $C_{\epsilon\epsilon}$ respectively.
The geometric solution may be considered the most intuitive one.
In the middle figure, the Kalman filter is applied to the characteristic function~$\chi$~(blue line).
The set-theoretic solution (dashed line) localizes the interface at~$x\approx0.139$, significantly differently from the geometric solution.
This shows that oft-used physical quantities, which exhibit behavior similar to a characteristic function in the vicinity of the interface of interest, are not suitable for data assimilation.
In the right figure, the Kalman filter is applied to the generating function~$G$~(blue line).
The analytic solution (dashed line) localizes the interface at~$x=0.64$, identically to the the geometric solution.
}
\label{fig:ls:ex1d}
\end{figure}

In summary, geometric, set-theoretic and analytic data assimilation are compared.
Set-theoretic data assimilation gives implausible results in one dimension, and is thus discarded from any further consideration.
Geometric and analytic data assimilation give consistent results in one dimension.
Nevertheless, the generalization of geometric data assimilation to higher dimensions is not straightforward because it requires a one-to-one correspondence between the predicted interface and the observations.
In the next test case, the generalization of analytic data assimilation, as the most feasible framework, to higher dimensions is investigated.

\FloatBarrier

%------------------------------------------------------------------------------

\subsection{Two-dimensional test case}
\label{sec:ls:ex2d}

The distance field $S$ of a corner in $(x_0, y_0)$ is given by
\begin{equation}
S(x, y; x_0, y_0) =
\begin{cases}
x - x_0 & \frac{\pi}{4} \leq \theta \leq \pi \\
y - y_0 & -\frac{\pi}{2} \leq \theta \leq \frac{\pi}{4} \\
-\sqrt{(x-x_0)^2 + (y-y_0)^2} & \pi \leq \theta \leq \frac{3\pi}{2}
\end{cases} \quad .
\label{eq:ls:ex2d:S}
\end{equation}
For prediction, the coordinates of the corner $(x_0, y_0)$ are taken to be independent and normally distributed:
\begin{equation}
x_0 \sim \mathcal{N}(x_0 \mid 0, 1) \quad , \quad y_0 \sim \mathcal{N}(y_0 \mid 0, 1) \quad .
\label{eq:ls:ex2d:S0}
\end{equation}
A sketch of the corner and its median distance field $S(x, y; 0, 0)$ are given in Fig.~\ref{fig:ls:ex2d}.

\begin{figure}[ht]
\centering
\raisebox{-0.5\height}{\includegraphics[width=0.325\linewidth]{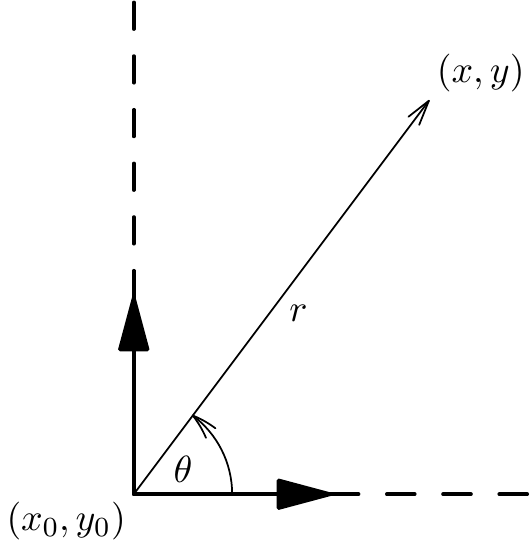}}
~
\raisebox{-0.5\height}{\includegraphics[width=0.625\linewidth]{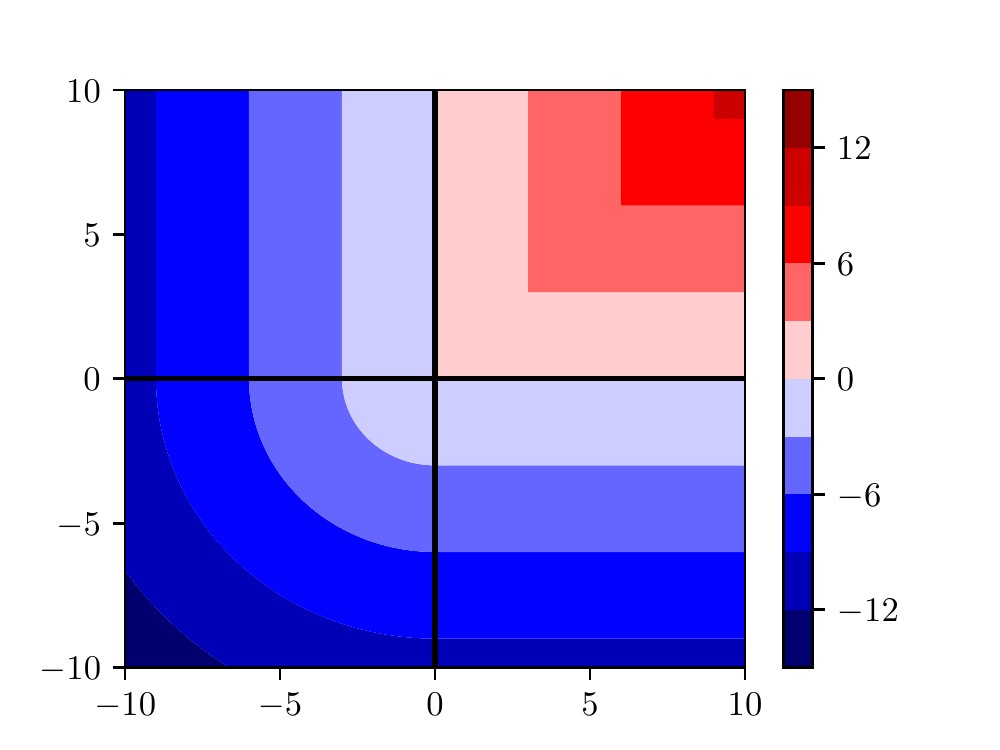}}
\caption{Sketch of a corner (left) and its median distance field $S(x, y; 0, 0)$ (right).
The walls of the corner (left, dashed line), which are represented by the zero-level set in analytic data assimilation, are parallel to the positive $x$ and $y$ axes and share their orientations (right).
The points with $0<\theta<\frac{\pi}{2}$ are defined to have positive distance.}
\label{fig:ls:ex2d}
\end{figure}

For analytic data assimilation, the predicted position~$x$ is again viewed as the level set of a predicted generating function.
It follows from Eq.~\eqref{eq:ls:ex2d:S} and \eqref{eq:ls:ex2d:S0} that the predicted generating function is given by
\begin{equation}
G(x, y) \sim \mathcal{N}(G(x, y) \mid \overline{G}(x, y), C_{GG}(x_1, y_1, x_2, y_2)) \quad ,
\end{equation}
where the mean $\overline{G}$ and the covariance function $C_{GG}$ are given by
\begin{align}
\overline{G}(x, y)
&= \mathbb{E}_{x_0, y_0}[S](x, y) \\
&= \int_{-\infty}^{\infty}{\int_{-\infty}^{\infty}{S(x, y; \xi, \eta)p(x_0=\xi, y_0=\eta)\,\ud\xi}\ud\eta} \quad ,
\label{eq:ls:ex2d:G_bar}\\
C_{GG}(x_1, y_1, x_2, y_2)
&= \mathbb{E}_{x_0, y_0}\left[\left(S-\mathbb{E}_{x_0, y_0}[S]\right)\times\left(S-\mathbb{E}_{x_0, y_0}[S]\right)\right](x_1, y_1, x_2, y_2) \\
&= \mathbb{E}_{x_0, y_0}\left[S \times S\right](x_1, y_1, x_2, y_2) - \mathbb{E}_{x_0, y_0}[S](x_1, y_1) \mathbb{E}_{x_0, y_0}[S](x_2, y_2) \\
&= \int_{-\infty}^{\infty}{\int_{-\infty}^{\infty}{S(x_1, y_1; \xi, \eta)S(x_2, y_2; \xi, \eta)p(x_0=\xi, y_0=\eta)\,\ud\xi}\ud\eta} \nonumber\\
&\quad - \mathbb{E}_{x_0, y_0}[S](x_1, y_1) \mathbb{E}_{x_0, y_0}[S](x_2, y_2) \quad ,
\label{eq:ls:ex2d:G_prime}
\end{align}
where~$\mathbb{E}_{x_0, y_0}$ denotes the expected value over the random variables~$x_0$ and~$y_0$, and the '$\times$'~symbol denotes the Cartesian product of two random variables.
The mean $\overline{G}$ and the covariance function $C_{GG}$ are computed from Eq.~\eqref{eq:ls:ex2d:G_bar} and \eqref{eq:ls:ex2d:G_prime} using Gauss-Hermite quadrature \cite{Szegoe1939}.
Alternatively, the mean and the variance function can be computed from the following one-dimensional integrals on the unit circle:
\begin{align}
\mathbb{E}_{x_0, y_0}[S](x, y)
&= \int_{\pi/4}^{\pi}{\cos(\theta)g(\theta)\,\ud{\theta}} + \int_{-\pi/2}^{\pi/4}{\sin(\theta)g(\theta)\,\ud{\theta}} - \int_{\pi}^{3\pi/2}{g(\theta)\,\ud{\theta}} \quad , \\
\mathbb{E}_{x_0, y_0}\left[S^2\right](x, y)
&= \int_{\pi/4}^{\pi}{\cos(\theta)^2h(\theta)\,\ud{\theta}} + \int_{-\pi/2}^{\pi/4}{\sin(\theta)^2h(\theta)\,\ud{\theta}} + \int_{\pi}^{3\pi/2}{h(\theta)\,\ud{\theta}} \quad .
\end{align}
The auxiliary functions $g$ and $h$ are given by
\begin{equation}
g(\theta) = \varphi(w) \left[\left(v^2 + 1\right) \Phi(v) + v \varphi(v)\right] \quad , \quad h(\theta) = \varphi(w) \left[\left(v^2 + 2\right) \varphi(v) + \left(3v + v^3\right) \Phi(v)\right] \quad ,
\end{equation}
where $\varphi$ and $\Phi$ denote the probability density function and cumulative distribution function of the normal distribution.
%\begin{equation*}
%\varphi(x) = \frac{1}{\sqrt{2\pi}}\exp\left(-\frac{x^2}{2}\right) \quad , \quad
%\Phi(x) = \int_{-\infty}^{x}{\varphi(z)\,\ud{z}} \quad ,
%\end{equation*}
The natural coordinates $(v, w)$ are given by
\begin{equation}
v = x\cos(\theta) + y\sin(\theta) \quad , \quad
w = x\sin(\theta) - y\cos(\theta) \quad .
\end{equation}
The mean and the standard deviation of the predicted generating function are shown in Fig.~\ref{fig:ls:ex2d:prior}.
Artefacts from the interaction between the nonlinearity of distance fields and the assumption of normal distributions are visible:
The mean of the predicted generating function~(Fig.~\ref{fig:ls:ex2d:prior}, left) is rounded in the corner whereas distance fields are supposed to have right angles~(Fig.~\ref{fig:ls:ex2d}).
The standard deviation of the predicted generating function~(Fig.~\ref{fig:ls:ex2d:prior}, right) assumes non-unity values along~$\theta=45^\circ$ although the coordinates of the corner are independent and normally distributed~(Eq.~\eqref{eq:ls:ex2d:S0}).
The effects of the artefacts in the predicted generating function are now examined in various observation scenarios.

\begin{figure}[ht]
\centering
\includegraphics[width=0.45\linewidth]{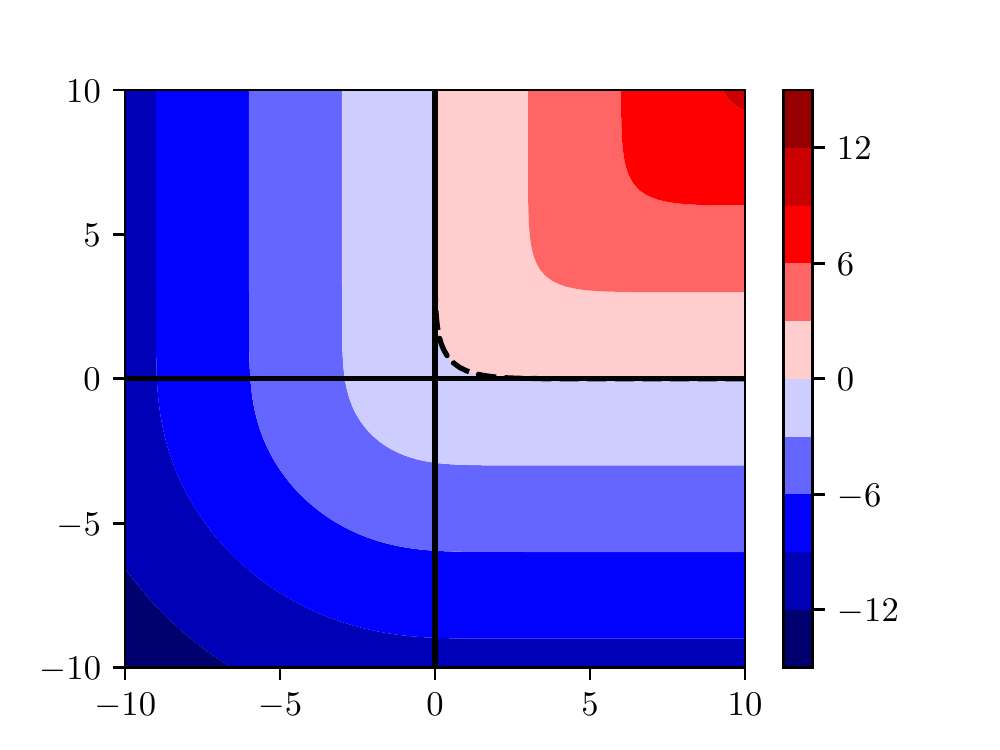}
~
\includegraphics[width=0.45\linewidth]{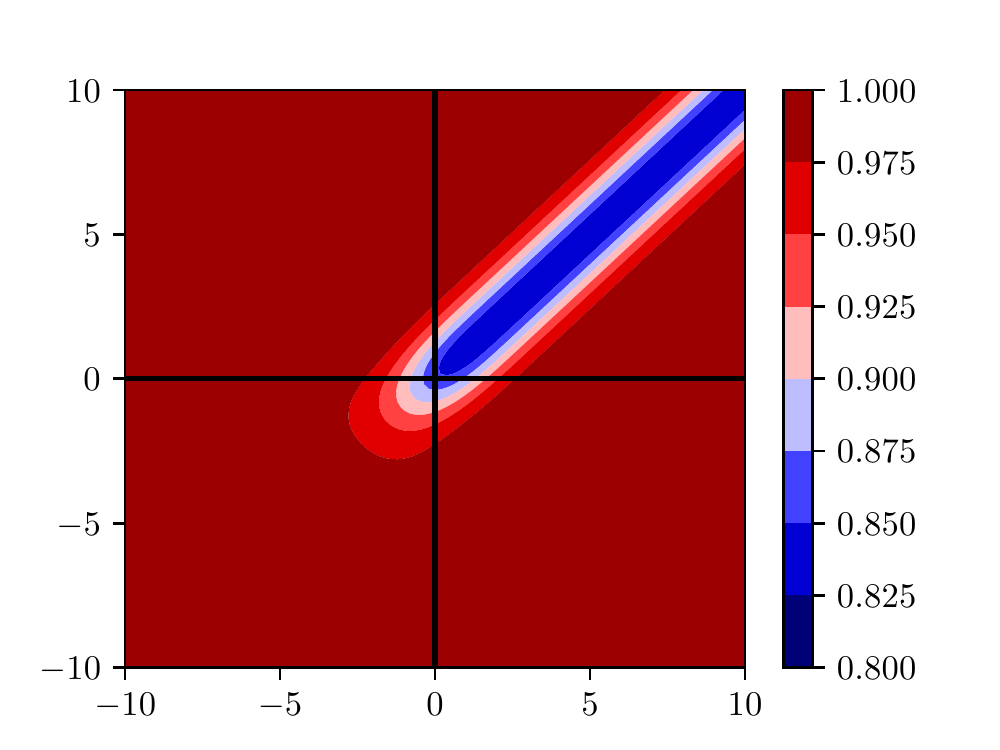}
\caption{Mean (left) and standard deviation (right) of predicted generating function.
The zero-level set in the mean is marked by a dashed line.
Away from $\theta=45^\circ$, the mean of the generating function comes close to its median (Fig.~\ref{fig:ls:ex2d}, right).
Note that the distance field $S$ (Eq.~\eqref{eq:ls:ex2d:S}) is smooth but not differentiable on $\theta=45^\circ$.
This explains why the right angles observed in the median are not preserved but are rounded in the mean.
At the origin, the right angle is thus offset by $\frac{1}{2\sqrt{\pi}}+\frac{1}{4}\sqrt{\frac{\pi}{2}} \approx 0.5954$.
Away from the origin, the right angles are asymptotically offset by $\frac{1}{\sqrt{\pi}} \approx 0.5642$.
The same can be observed in the standard deviation.
Away from $\theta=45^\circ$, the standard deviation is close to unity, which is identical to the standard deviation in the coordinates of the corner (Eq.~\eqref{eq:ls:ex2d:S0}).
At the origin, the standard deviation assumes a value of $\sqrt{\frac{5}{4} - \frac{3}{4\pi} - \frac{\pi}{32} - \frac{1}{4\sqrt{2}}} \approx 0.8581$.
Away from the origin, the standard deviation asymptotically assumes a value of $\sqrt{1-\frac{1}{\pi}} \approx 0.8256$.}
\label{fig:ls:ex2d:prior}
\end{figure}

Unlike the one-dimensional example, where one point fully characterizes an interface, observations in two (and higher) dimensions vary in resolution, ranging from the position of one point on the interface to a full sampling of the interface.
In Fig.~\ref{fig:ls:ex2d:post:1}, the influence of one observation depending on its location is investigated.
It is evident that one observation is insufficient to represent a corner.
Nevertheless, analytic data assimilation gives qualitatively meaningful results in that the observation shifts the vertical and horizontal edges proportionally depending on the location of the observation.
This is fundamentally different from geometric data assimilation~(Fig.~\ref{fig:ls:ex1d}, left), where every point on the interface is required to correspond to an observation.
This makes the presented (analytic) level-set data assimilation framework fully front-capturing (rather than front-tracking), both with respect to the description of the motion of the interface and the assimilation of data.
As the number of observation points increases, the right angle in the corner is retrieved~(Fig.~\ref{fig:ls:ex2d:post:n}).

\begin{figure}[ht]
\centering
\begin{subfigure}[t]{0.2\textwidth}
\includegraphics[height=\textwidth]{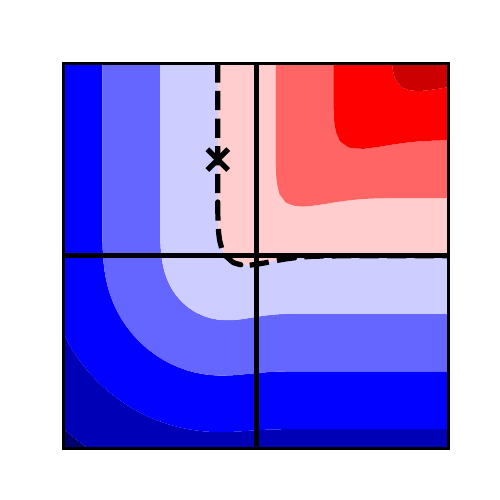}
\caption{$(-2, 5)$}
\end{subfigure}
\begin{subfigure}[t]{0.2\textwidth}
\includegraphics[height=\textwidth]{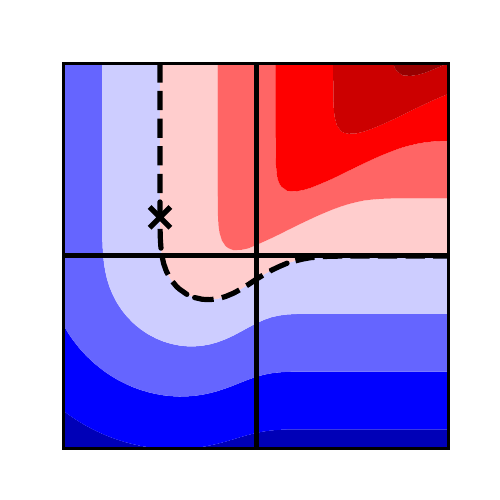}
\caption{$(-5, 2)$}
\end{subfigure}
\begin{subfigure}[t]{0.2\textwidth}
\includegraphics[height=\textwidth]{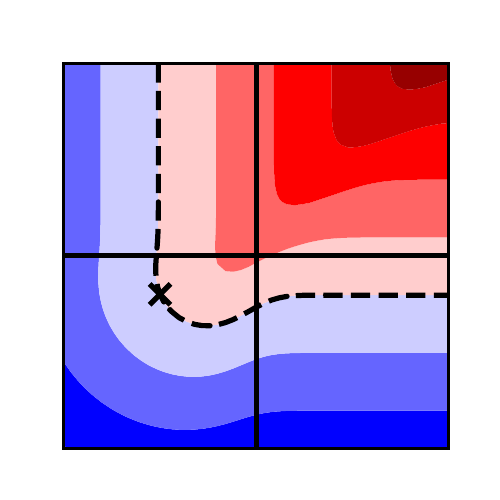}
\caption{$(-5, -2)$}
\end{subfigure}
\begin{subfigure}[t]{0.2\textwidth}
\includegraphics[height=\textwidth]{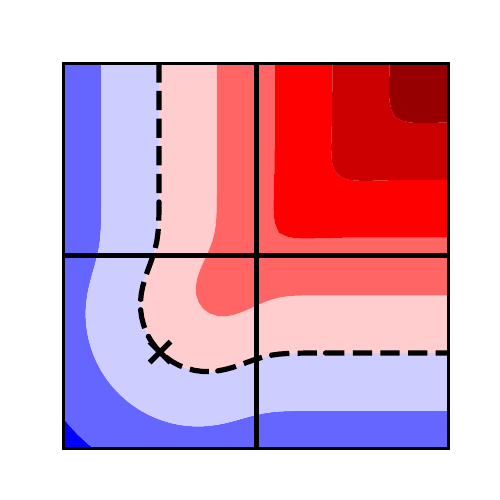}
\caption{$(-5, -5)$}
\end{subfigure}
\begin{subfigure}[t]{0.1\textwidth}
\raisebox{-0.01\linewidth}{\includegraphics[height=2\textwidth]{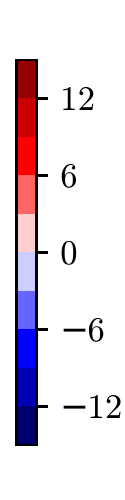}}
\end{subfigure}
\caption{
Means of filtered generating functions for one observation of the interface at various locations.
The observed location is marked by a cross.
Its coordinates are specified in each caption.
The zero-level set is marked by a dashed line.
Without loss of generality, the observation error is set to zero.
(a)~The observation is located significantly closer to the vertical than the horizontal edge of the predicted generating function~(Fig.~\ref{fig:ls:ex2d:prior}, left).
Hence, the observation is primarily associated with points on the vertical edge.
Data assimilation thus mainly shifts the points on the vertical edge while the vertex and the horizontal edge remain largely unaffected.
(b) As the location of the observation moves downwards, the observation gets more strongly associated with the vertex of the corner while the horizontal edge remains largely unaffected.
(c) As the location of the observation crosses the $x$-axis, data assimilation also shifts the points on the horizontal edge.
(d) As the location of the observation reaches the diagonal, data assimilation equally affects the vertical and horizontal axes of the corner.
}
\label{fig:ls:ex2d:post:1}
\end{figure}

\begin{figure}[ht]
\centering
\begin{subfigure}[t]{0.2\textwidth}
\includegraphics[height=\textwidth]{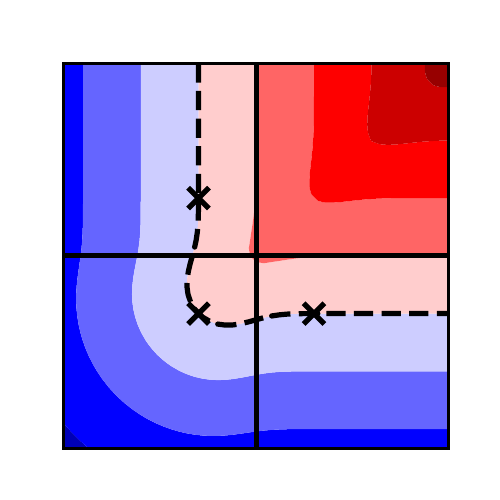}
\caption{3 observations.}
\end{subfigure}
\begin{subfigure}[t]{0.2\textwidth}
\includegraphics[height=\textwidth]{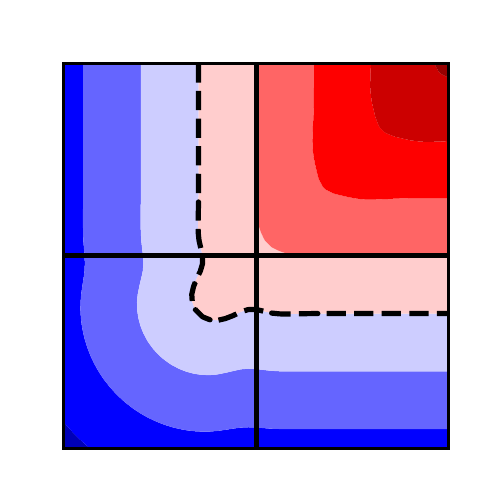}
\caption{7 observations.}
\label{fig:ls:ex2d:post:n:7}
\end{subfigure}
\begin{subfigure}[t]{0.2\textwidth}
\includegraphics[height=\textwidth]{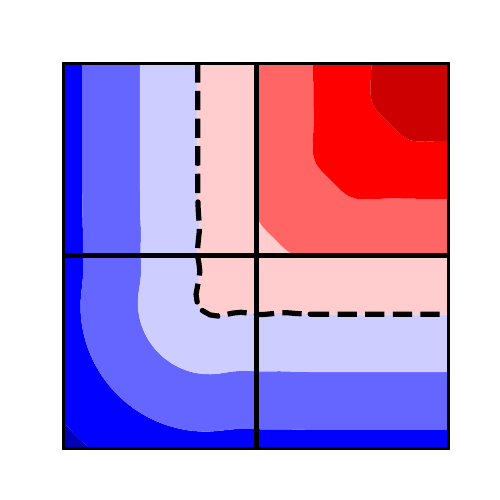}
\caption{15 observations.}
\end{subfigure}
\begin{subfigure}[t]{0.2\textwidth}
\includegraphics[height=\textwidth]{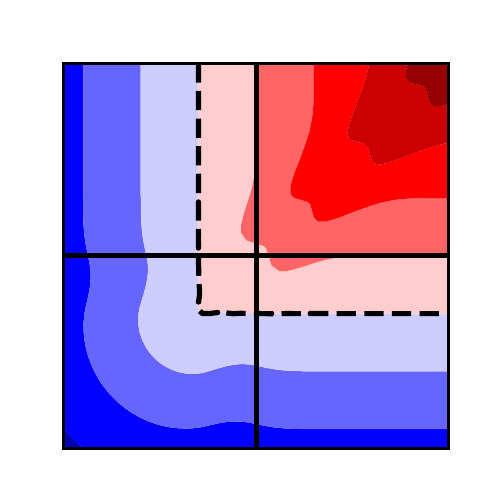}
\caption{31 observations.}
\end{subfigure}
\begin{subfigure}[t]{0.1\textwidth}
\raisebox{-0.01\linewidth}{\includegraphics[height=2\textwidth]{ex2d-cbar.pdf}}
\end{subfigure}
\caption{Means of filtered generating functions for multiple observations of the interface at various resolutions.
The observed locations are marked by crosses and the zero-level sets by dashed lines.
The observations are equidistantly spaced between $(-3, 3)$ and $(3, -3)$ via $(-3, -3)$.
The number of observations is specified in each caption.
For better readability, the observations are not marked in (b)-(d).
Without loss of generality, the observation error is set to zero.
From (a) to (d), as the number of the observation increases, the observations give a clearer image of the observed corner.
Hence, the zero-level set of the the filtered generating function more closely resembles a corner.}
\label{fig:ls:ex2d:post:n}
\end{figure}

It is worth pointing out the probabilistic nature of the presented level-set data assimilation framework.
Similar to soft clustering \cite{MacKay2003}, the points on the predicted interface are not associated with distinct observations but with all observations to varying degrees.
If we consider $\psi^a$ to be a function of $\psi^f$, $C_{\psi\psi}^f$, $y$, $C_{\epsilon\epsilon}$ and $M$ (Eq.~\eqref{eq:da:kalman:mean}), the sensitivity to data is given by
\begin{equation}
\pdiff{\psi^a}{y} = \left(MC_{\psi\psi}^f\right)^T\left[C_{\epsilon\epsilon}+MC_{\psi\psi}^fM^T\right]^{-1} \quad ,
\label{eq:ls:ex2d:post:modes}
\end{equation}
where the columns represent the sensitivities to individual observations.
The shift in the zero-level set (Fig.~\ref{fig:ls:ex2d:post:n}) is the superposition of the contributions due to the individual observations.
For example, while the points on the vertical edge far away from the vertex are mainly affected by the uppermost observation (Fig.~\ref{fig:ls:ex2d:post:modes:1}), the other observations have a vanishing effect on them as well (Fig.~\ref{fig:ls:ex2d:post:modes}\subref{fig:ls:ex2d:post:modes:2}-\subref{fig:ls:ex2d:post:modes:4}).
Nevertheless, the shape deformations are more localized compared to Fig.~\ref{fig:ls:ex2d:post:1}.
While Eq.~\eqref{eq:ls:ex2d:post:modes} does not formally depend on the observations~$y$, it includes information about the number and the locations of the observations signified by the measurement operator~$M$.
Note that each observation of interest leaves the values of the generating function at the other observation locations intact.

\begin{figure}[ht]
\centering
\begin{subfigure}[t]{0.2\textwidth}
\includegraphics[height=\textwidth]{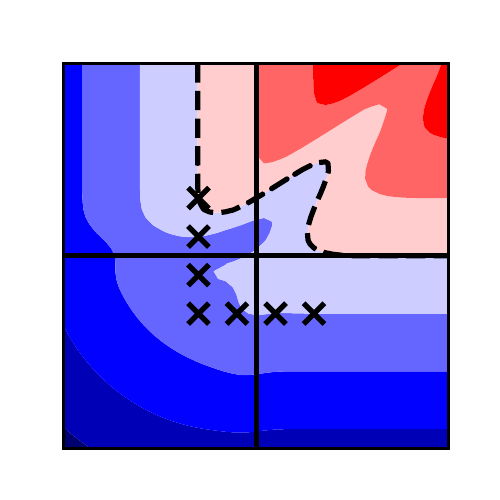}
\caption{$(-3, 3)$}
\label{fig:ls:ex2d:post:modes:1}
\end{subfigure}
\begin{subfigure}[t]{0.2\textwidth}
\includegraphics[height=\textwidth]{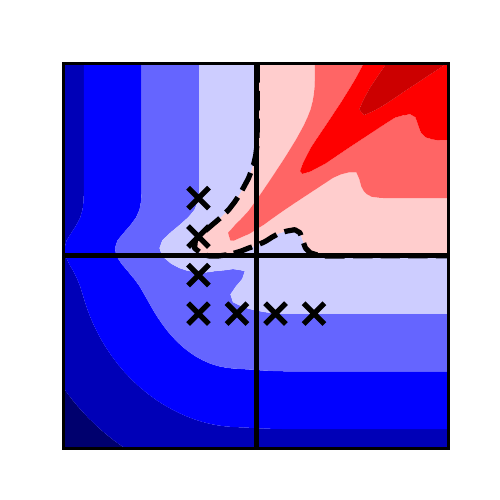}
\caption{$(-3, 1)$}
\label{fig:ls:ex2d:post:modes:2}
\end{subfigure}
\begin{subfigure}[t]{0.2\textwidth}
\includegraphics[height=\textwidth]{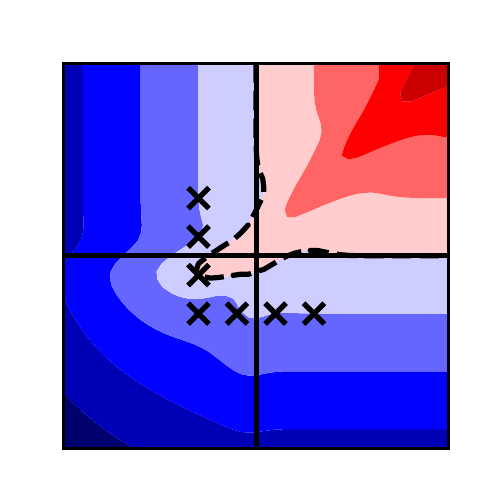}
\caption{$(-3, -1)$}
\label{fig:ls:ex2d:post:modes:3}
\end{subfigure}
\begin{subfigure}[t]{0.2\textwidth}
\includegraphics[height=\textwidth]{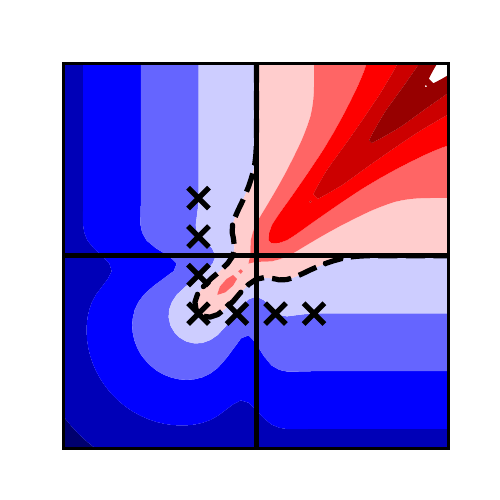}
\caption{$(-3, -3)$}
\label{fig:ls:ex2d:post:modes:4}
\end{subfigure}
\begin{subfigure}[t]{0.1\textwidth}
\raisebox{-0.01\linewidth}{\includegraphics[height=2\textwidth]{ex2d-cbar.pdf}}
\end{subfigure}
\caption{Means of filtered generating functions due to individual observations.
The observed locations are marked by crosses and the zero-level sets by dashed lines.
The observations are equidistantly spaced between $(-3, 3)$ and $(3, -3)$ via $(-3, -3)$.
The location of the observation of interest is specified in each caption.
Without loss of generality, the observation error is set to zero.
%The shape deformations are depicted for observations on the vertical edge.
%The shape deformations for observations on the horizontal edge are their mirror images.
}
\label{fig:ls:ex2d:post:modes}
\end{figure}

In summary, analytic data assimilation is applied to a prototypical two-dimensional shape, a corner.
Within the level-set data assimilation framework, the differences between distance fields, predicted generating functions and filtered generating functions are illustrated.
In particular, the effects of the number of observation points and their locations are studied.
In the following section, the level-set data assimilation framework is applied to a more realistic experiment.
In general, the predicted generating function is taken as a composition of straight sections, which exhibit the behavior discussed for the one-dimensional test case, and strongly curved sections insufficiently resolved by the observations, which exhibit the behavior discussed for the two-dimensional test case.
The theoretical insights from this section will be referenced throughout the next section.

\FloatBarrier

%------------------------------------------------------------------------------

%\subsection{Filter divergence and model errors}
%
%\begin{itemize}
%\item For large number of observations, matrix inversion in the computation of the Kalman gain may become computationally expensive.
%A subspace inversion scheme may introduce sampling errors, even into the square-root filter.
%\item Local analysis.
%Explore how the signed distance function in the farfield correlates to its value on the level set.
%\end{itemize}
%
%Reduced-order Kalman filter is an extreme case of local analysis \cite{Meldi2017}.
%
%Read more about viscosity solutions.
%
%Explore the structure of the covariance matrix using adjoint sensitivity analysis with respect to the boundary conditions.
%
%What do you do if the number of observations exceeds the size of the ensemble?

%------------------------------------------------------------------------------

%\section{Test cases}
%
%Sampling error can result in spurious correlations.
%Compare perturbed observations to square-root filter.
%
%Spend more time on showing data assimilation of topology (vs shape)?
%
%Not eikonal field, but curvature dependent speed?

%------------------------------------------------------------------------------

\section{Twin experiments for a premixed flame inside a duct}
\label{sec:flame}

%Few applications of data assimilation in supersonic or reacting flow:
%\begin{itemize}
%\item Ensemble Kalman filter in one-dimensional convection-diffusion-reaction problem and two-dimensional flame propagation \cite{Gao2017}.
%Estimation of only one parameter.
%'Perfect' observations synthesized from exact solutions and identical CFD models.
%Study of observation frequency and error as well as spatial distribution.
%Two-dimensional flame propagation grossly oversimplified: progress variable convected as passive scalar, velocity field and fuel mass fraction prescribed, estimating one flame speed parameter.
%Naive data assimilation based on zero observation error.
%\item Reduced-order Kalman filter in premixed flame and wildfire propagation \cite{Rochoux2013}.
%Level-set method based on progress variable.
%Innovation vector formed from distances of observed isocontour points projected onto background isocontour along normal direction of observed contour.
%Not consistent with definition of Eulerian framework.
%\item Deflagration-to-detonation transition inside shock-focusing geometry for pulse detonation combustor \cite{Lemke2016}.
%Pressure evolution inside detonation tube measured at various locations over time.
%Arrhenius and diffusion parameters are estimated.
%Requires mollification of discontinuities in the objective functional.
%\end{itemize}

% Introduction.
Thermoacoustic instabilities are a persistent challenge in the design of jet and rocket engines:
Velocity and pressure oscillations inside the combustion chamber interact with the flame and cause an unsteady heat release rate.
If moments of higher heat release rate coincide with moments of higher pressure (and lower heat release rate with lower pressure), acoustic oscillations arise.
This can lead to large-amplitude oscillations causing structural damage in the jet or rocket engine~\cite{Lieuwen2005, Culick2006}.

The so-called $G$-equation model is a reduced-order model to study the flame dynamics which lead to heat release rate perturbations~\cite{Fleifil1996, Dowling1999}.
The premixed flame is modeled as an interface captured by a level set.
The velocity of a point in the level set is the sum of the flame speed, which is normal to the interface and points towards the unburnt gas, and the underlying flow field.
The underlying flow field includes both hydrodynamic and acoustic contributions.
The $G$-equation model is used to compute flame transfer functions~(FTF) or flame describing functions~(FDF), which give the heat release perturbation to a given velocity perturbation.
Coupled with linear acoustics models, the $G$-equation model has been very successful in qualitatively characterizing the linear and nonlinear dynamics of self-excited thermoacoustic oscillations~\cite{Kashinath2014, Waugh2014, Orchini2016a, Semlitsch2017}.

We demonstrate our level-set data assimilation framework by performing twin experiments for the $G$-equation model applied to a ducted premixed Bunsen flame under acoustic forcing.
In our twin experiments, both the model predictions and the observations come from $G$-equation simulations, but with different sets of parameters~\cite[Chapter~9]{Reich2015}.
This leads to uncertainties in the parameters and, as a result, in the states.
In the absence of uncertainties in the model, the twin experiment is an important benchmark in the quantitative assessment of a data assimilation framework because it becomes possible to compare filtered and smoothed solutions to a reference solution at all times~\cite[Chapter~8]{Reich2015}.
An example for the application of our level-set data assimilation framework to more realistic observations from a direct numerical simulation can be found in the proceedings of the 2018 CTR Summer Program~\cite{Yu2018}.

\FloatBarrier

%------------------------------------------------------------------------------

\subsection{Reference solution}
\label{sec:res:ref}

% Governing equations.
The $G$-equation is given by~\cite[Chapter~2]{Peters2000}
\begin{equation}
\pdiff{G}{t} + \left[u(r, t)-s_\mathrm{L}n\right] \cdot \nabla G = 0 \quad ,
\label{eq:res:est:g}
\end{equation}
where the underlying flow field is denoted by~$u(r, t)$ and the laminar flame speed by~$s_\mathrm{L}$.
For $|\nabla G| = 1$, the $G$-equation is formally equivalent to the Hamilton-Jacobi equation~(Eq.~\ref{eq:ls:hj:g}).
The underlying flow field~$u(r, t)$ is a superposition of the steady, inviscid flow field~$U$ and the time-dependent perturbation velocity field~$u'(r, t)$.
In the reduced-order model of~$u'(r, t)$, the perturbation has an amplitude~$\varepsilon U$, and travels downstream at the phase speed~$U/K$, where~$K$ and~$\varepsilon$ are model parameters~\cite{Kashinath2013}.
In the axisymmetric case, the cylindrical components of the perturbation velocity field are given by
\begin{equation}
u'(r, t)
= \begin{pmatrix} u_\rho' \\ u_\theta' \\ u_z' \end{pmatrix}
= \begin{pmatrix}
\frac{1}{2} k\rho\varepsilon U \sin(kz - \omega t) \\
0 \\
\varepsilon U \cos(kz - \omega t)
\end{pmatrix} \quad ,
\end{equation}
where the wavenumber~$k$ and the angular frequency~$\omega$ satisfy the dispersion relation $\omega(k) = Uk/K$.
The perturbation velocity field~$u'(r, t)$ is, mathematically speaking, divergence-free, and, physically speaking, satisfies the continuity equation.

% Computational set-up.
In Fig.~\ref{fig:res:burner}, a sketch of the ducted premixed flame is shown.
The ducted premixed flame is forced at an angular frequency of $277.78\,\mathrm{rad}/\mathrm{s}$.
Tab.~\ref{tab:res:param} shows the measurements of the burner and the tube.
The $G$-equation is numerically solved using the narrow-band level-set method with distance reinitialization~\cite{Peng1999, Sussman1994}:
The computational domain is discretized using a weighted essentially non-oscillatory (WENO) scheme in space and a total-variation diminishing (TVD) version of the Runge-Kutta scheme in time, which give up to fifth-order accuracy in space and third-order accuracy in time~\cite{Liu1994, Jiang1996, Shu1988}.
At the base of the flame, a rotating boundary condition is used~\cite{Waugh2013}.
The $G$-equation solver has been verified in a number of studies~\cite{Preetham2008, Kashinath2014, Waugh2014, Orchini2016, Semlitsch2017}.
These studies showed that the $G$-equation reproduces the dynamics of premixed flames only qualitatively, not quantitatively.
With our level-set data assimilation framework, significant improvements on quantitative predictions are achieved.
In axial and radial coordinates, a uniform $401 \times 401$ Cartesian grid is used, which corresponds to $30\,\mu\mathrm{m} \times 180\,\mu\mathrm{m}$ grid cells.
An advection Courant-Friedrich-Lewy (CFL) number of 0.02 is chosen, which corresponds to a timestep of 3.6\,$\mu$s or approximately 6,283 timesteps over one period.
%Cusps and pinched-off fuel-air pockets are reliably captured.
%The generating function is reconstructed from the solution to the $G$-equation using a fast-marching method~\cite{Sethian1996}.

\begin{figure}[ht]
\centering
\raisebox{-0.5\height}{\includegraphics[width=0.4\linewidth]{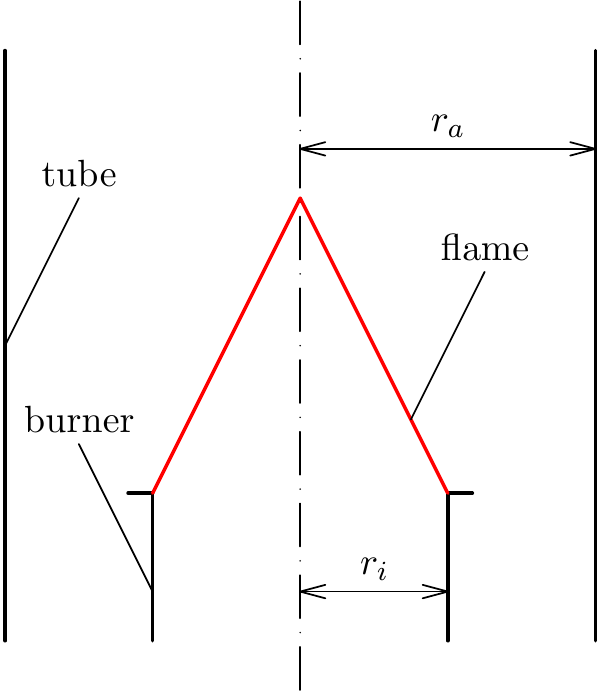}}
\qquad
\raisebox{-0.5\height}{\includegraphics[width=0.4\linewidth]{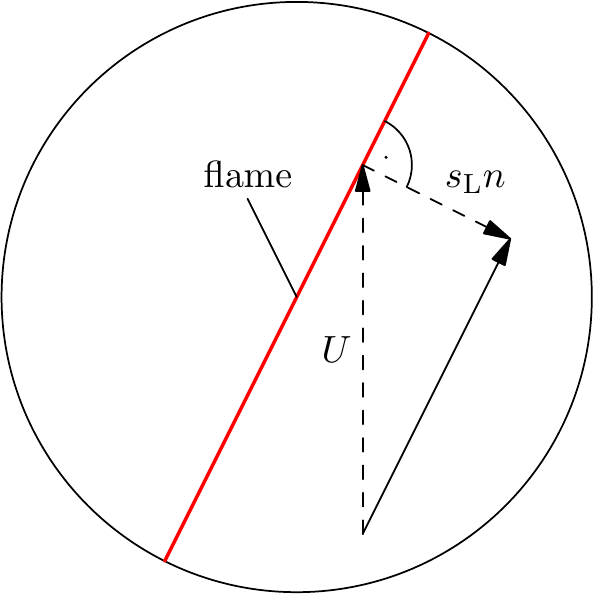}}
\caption{
Sketch of the ducted premixed flame without perturbations (left).
The superposition of the underlying flow field and the propagation of the premixed flame into the unburnt gas gives velocity vectors, which are locally tangential to the flame surface (right).
}
\label{fig:res:burner}
\end{figure}

\begin{table}
\caption{$G$-equation parameters for the reference simulation.}
\centering
\begin{tabular}{l|cccc|ccc}
Parameter & $r_i$ [m] & $r_a$ [m] & $U$ [m/s] & $s_\mathrm{L}$ [m/s] & $\omega$ [rad/s] & $K$ [-] & $\varepsilon$ [-] \\
\hline
Parameter value & 0.006 & 0.012 & 1.0 & 0.164 & 277.78 & 1.0 & 0.1
\end{tabular}
\label{tab:res:param}
\end{table}

% Results.
Snapshots of the reference solution to the $G$-equation for the ducted premixed flame are shown in Fig.~\ref{fig:res:g}.
The flame is attached to the burner lip, while the perturbations travel from the base of the flame to the tip.
When the perturbations are sufficiently large, a fuel-air pocket pinches off.
In our reduced-order model, the perturbations are mainly governed by two non-dimensional parameters~$K$ and~$\varepsilon$, which govern phase speed and amplitude respectively~\cite{Kashinath2013}.
In practice, neither parameter is accurately known a priori, which is a major source of uncertainty.

\begin{figure}
\centering
\includegraphics[width=0.225\linewidth]{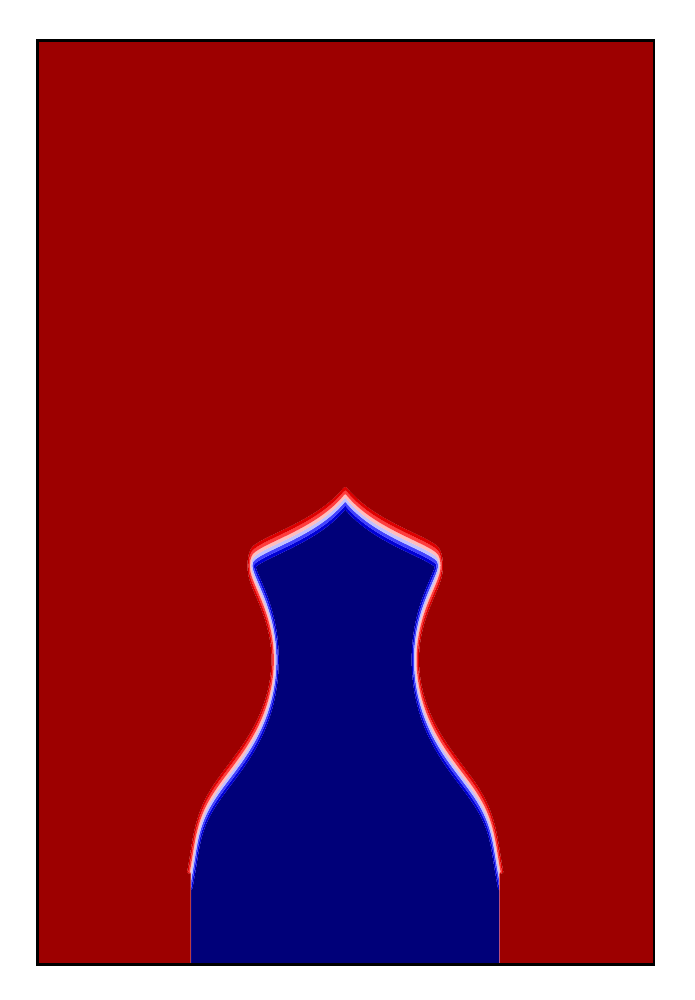}
\includegraphics[width=0.225\linewidth]{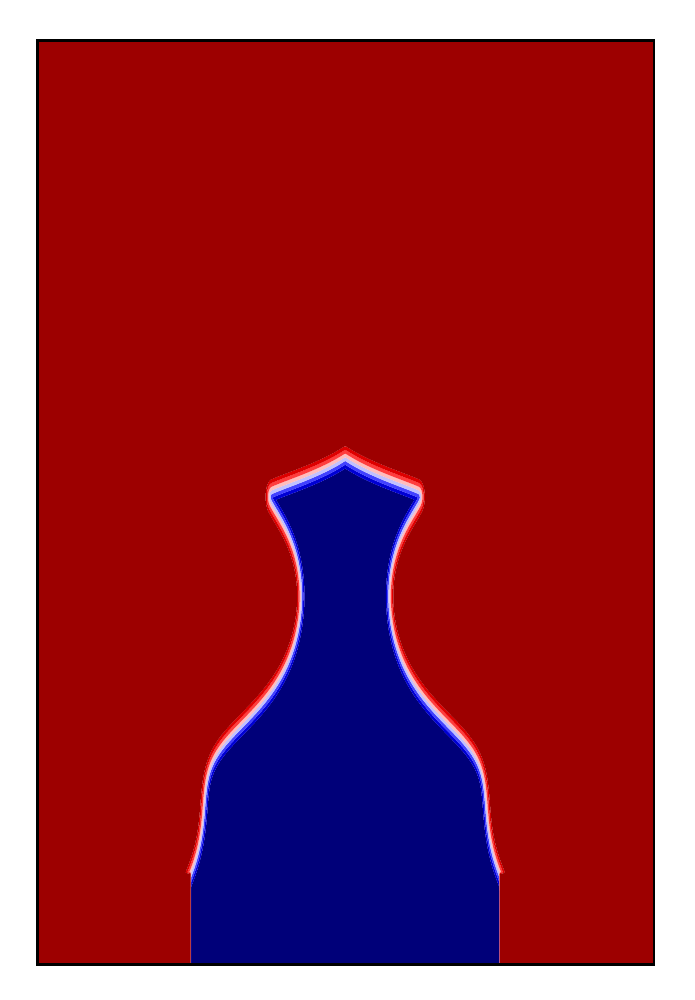}
\includegraphics[width=0.225\linewidth]{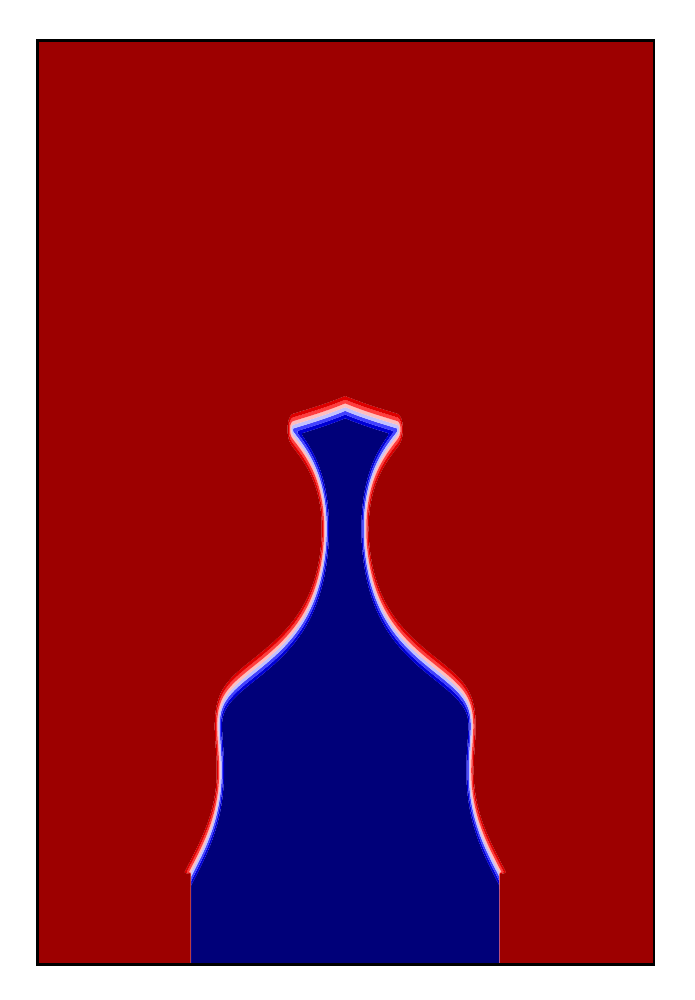}
\includegraphics[width=0.225\linewidth]{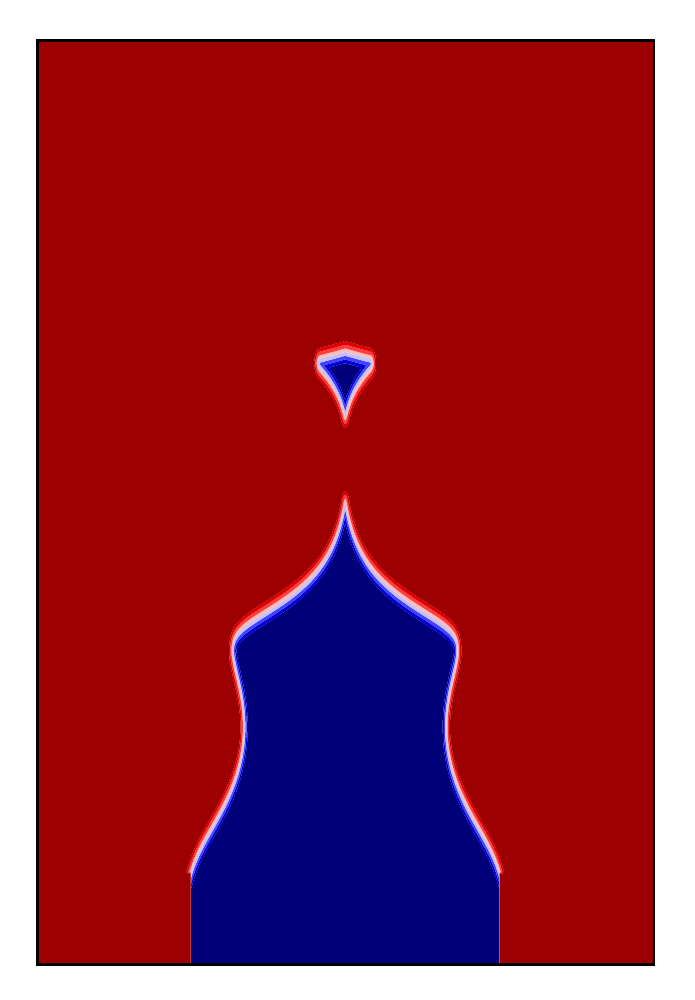}
\caption{
Snapshots of $G$-equation simulation over one period of harmonic forcing.
The fuel-air mixture leaves the burner at the bottom of each frame.
The infinitely thin flame surface separates the burnt gas (red) from the unburnt gas (blue).
Each snapshot from left to right is distanced by a quarter of the forcing period.
}
\label{fig:res:g}
\end{figure}

%------------------------------------------------------------------------------

\subsection{Combined state and parameter estimation}
\label{sec:res:est}

%In the following subsections, two types of problems are addressed:
%The forward problem quantifies the propagation from the input (the uncertainty in the parameters) to the output (the uncertainty in the location of the flame surface).
%In the inverse problem, data assimilation is used to identify the location of the flame surface, and reduce the uncertainty in the location and the parameters.
%In the first subsection, the forward problem is solved by treating the Hamilton-Jacobi equation as a stochastic process.
%In the second subsection, the inverse problem is solved by using synthetic data from the same $G$-equation solver.
%In the third subsection, the inverse problem is solved by using data from a direct numerical simulation of the premixed flame, which provides the reference solution.

Three variations of the twin experiment are performed:
\begin{description}
\item[Monte-Carlo experiment.]
An ensemble of 32 $G$-equation simulations is performed.
Each simulation has a different set of parameters $K$ and $\varepsilon$.
They are independently sampled from a normal distribution with a standard deviation of 20\,\%.
\item[State estimation.]
In addition to the procedure described for the Monte-Carlo experiment, the ensemble Kalman filter~(Theorem~\ref{thm:da:ensemble_kalman_filter}) is applied every 1,000 timesteps to update the discretized generating function~$G$ while leaving each set of parameters~$K$ and~$\varepsilon$ unaltered.
The observations are extracted from the reference solution~(Section~\ref{sec:res:ref}):
The measurement operator~$M$ is an indicator matrix which identifies all grid points in the reference solution adjacent to its zero-level set.
All aforementioned grid points are treated as observations of the flame surface.
As such, all entries in the vector~$y_k$ are set to 0.
The observation error~$\sigma_\epsilon$ is set to 60\,$\mu$m, which corresponds to two grid cells.
\item[Combined state and parameter estimation.]
In addition to the procedure described for the state estimation, the discretized generating function~$G$ is augmented by appending the parameters~$K$ and~$\varepsilon$ to the state vector~(Eq.~\ref{tab:da:state_space:state_augmentation}).
Thus, both the state and the parameters are updated whenever data in the form of observations from the reference solution is assimilated.
\end{description}
In each variation of the twin experiment, the $k$-th entry in the state vector~$\psi$ is marginally distributed according to
\begin{equation}
\psi[k] \sim \mathcal{N}(\psi[k] \mid \overline{\psi}[k], C_{\psi\psi}[k, k]) \quad .
\end{equation}
The mean~$\overline{\psi}[k]$ and the variance $C_{\psi\psi}[k, k]$ are computed from Eq.~\eqref{eq:da:ensemble}.
Explicitly, the likelihood for the flame surface to be found at the location of the $k$-th entry is given by
\begin{equation}
p[k] = \frac{1}{\sqrt{2\pi C_{\psi\psi}[k, k]}}\exp\left(-\frac{\overline{\psi}[k]^2}{2C_{\psi\psi}[k, k]}\right) \quad .
\label{eq:res:est:lik}
\end{equation}
Alternatively, the logarithm of the normalized likelihood is given by \cite{Rasmussen2006}
\begin{equation}
\log\left(\frac{p[k]}{p_0[k]}\right) = -\frac{\overline{\psi}[k]^2}{2C_{\psi\psi}[k, k]} \quad .
\label{eq:res:est:loglik}
\end{equation}

In Fig.~\ref{fig:res:uq:heat}, the logarithm of the normalized likelihood is shown for the three variations of the twin experiment.
Its zero-level set gives the maximum-likelihood location of the flame surface.
The more negative the value at a location is, the less likely the flame surface is to be found there.
In Fig.~\ref{fig:res:uq:heat:a}, the logarithm of the normalized likelihood is shown for the Monto-Carlo experiment.
As the perturbation travels from the base of the flame to the tip, the high-likelihood region for the location of the flame surface spreads out.
The high-likelihood region is largest when fuel-air pockets pinch off, which represents maximal uncertainty.
%Physically speaking, this means that the accurate detection of pinch-off events is the most challenging task.
In Fig.~\ref{fig:res:uq:heat:b}, the logarithm of the normalized likelihood is shown for the twin experiment with state estimation.
A qualitative comparison to Fig.~\ref{fig:res:uq:heat:a} shows that the high-likelihood region for the location of the flame surface is significantly tighter.
While the high-likelihood region still grows as the perturbation travels, the regular assimilation of data suppresses it to the vicinity of the observed flame surface.
In Fig.~\ref{fig:res:uq:heat:c}, the logarithm of the normalized likelihood is shown for the twin experiment with combined state and parameter estimation.
After few data assimilation cycles, the state-augmented ensemble Kalman filter has learnt the parameters~$K$ and~$\varepsilon$ almost exactly.
Knowledge of the parameters enables highly precise predictions of the location of the flame surface, even during pinch-off events.
Unlike in Fig.~\ref{fig:res:uq:heat:a} and \subref{fig:res:uq:heat:b}, no growth in the high-likelihood region is qualitatively discernable.

\begin{figure}
\begin{tabular}{p{0.03\linewidth}p{0.175\linewidth}p{0.175\linewidth}p{0.175\linewidth}p{0.175\linewidth}}
 & \centering $t=10,000$ & \centering $t=11,500$ & \centering $t=13,000$ & \centering $t=14,500$
\end{tabular}
\begin{subfigure}[t]{\textwidth}
\centering
\includegraphics[width=0.2\linewidth]{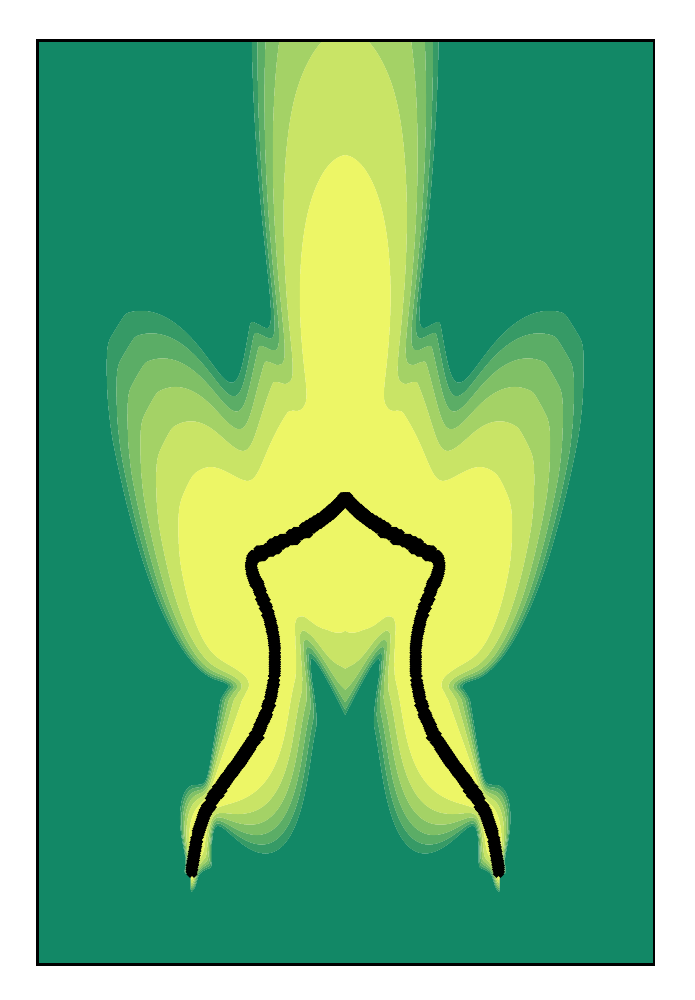}
\includegraphics[width=0.2\linewidth]{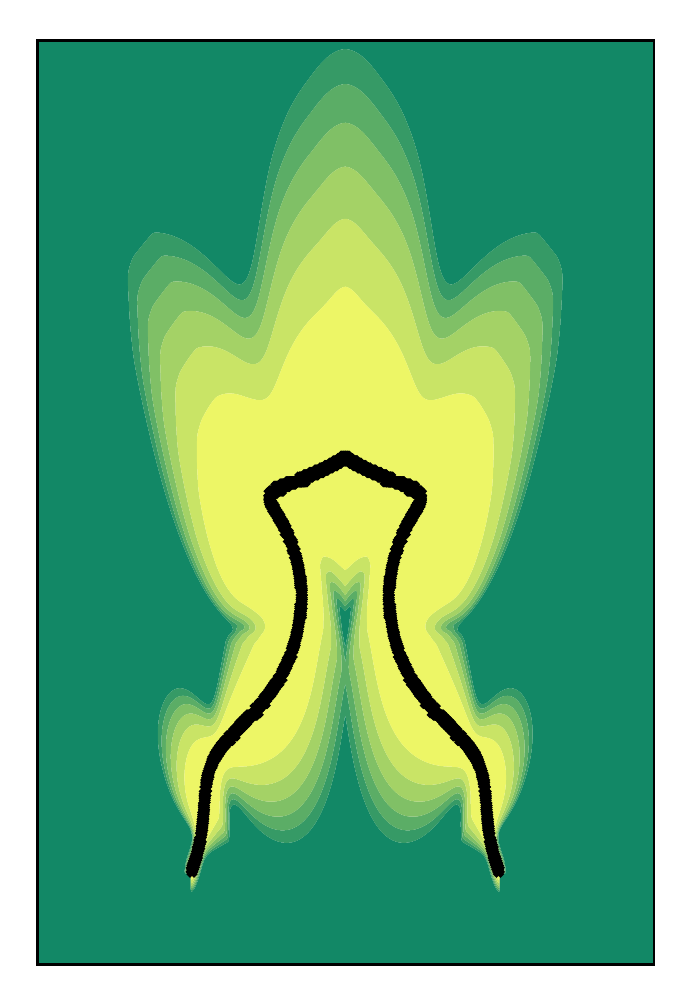}
\includegraphics[width=0.2\linewidth]{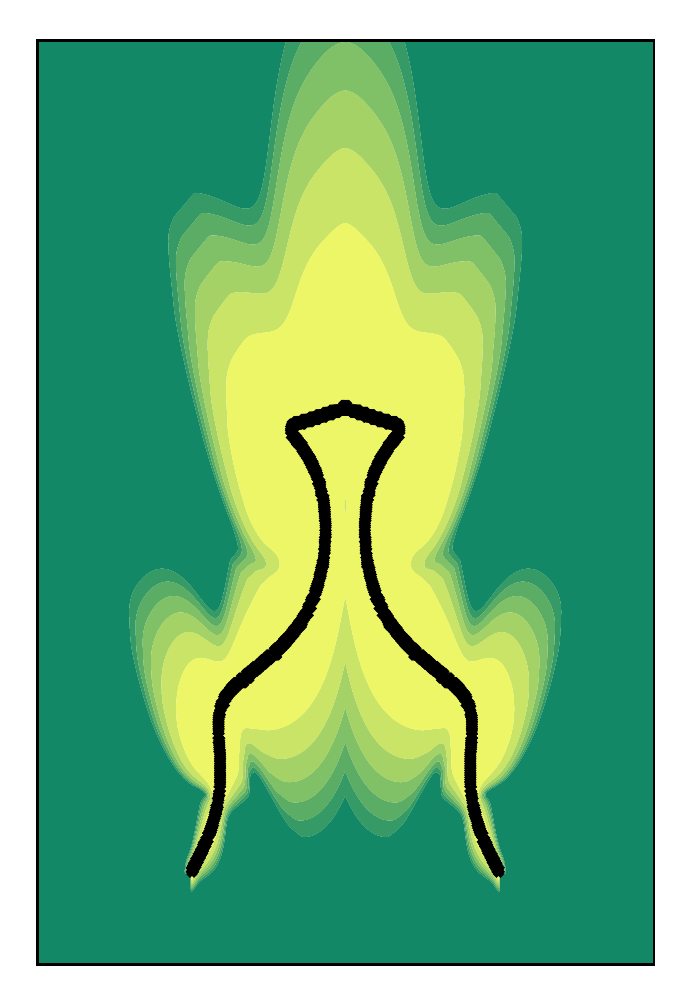}
\raisebox{-0.005\linewidth}{\includegraphics[width=0.27\linewidth]{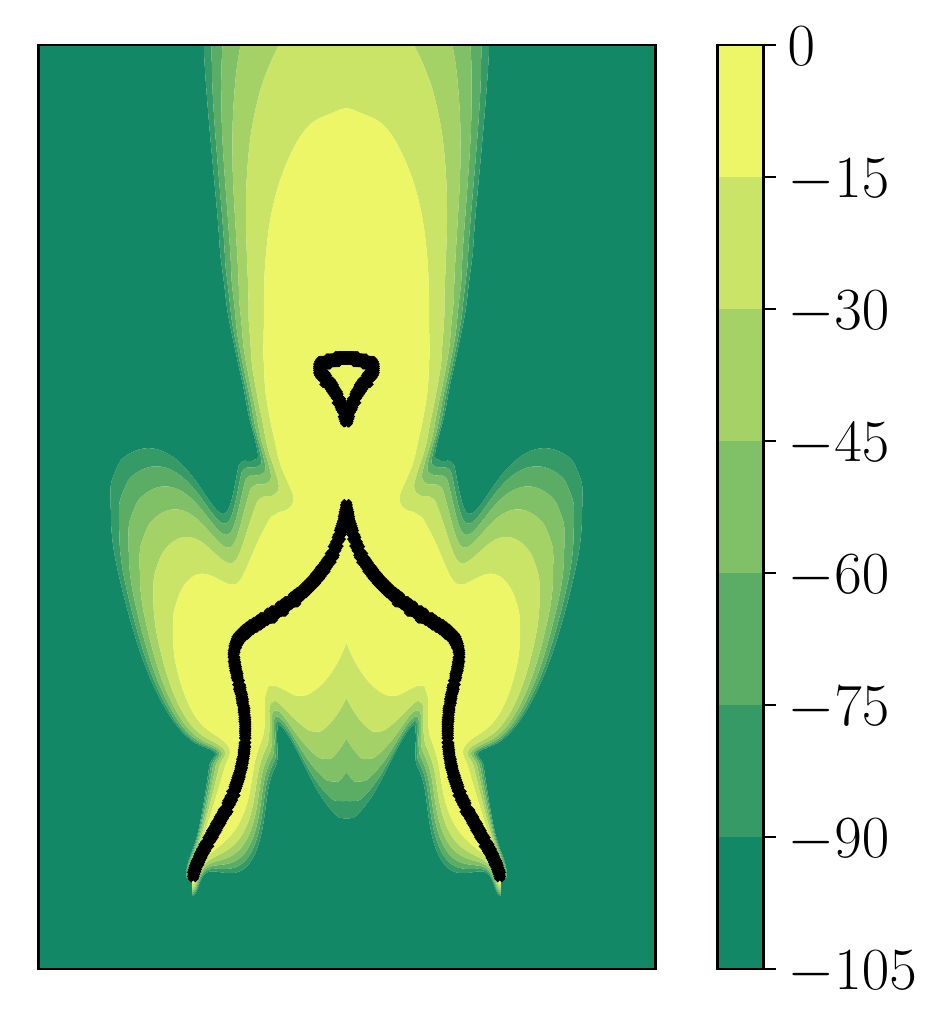}}
\caption{Monte-Carlo experiment.}
\label{fig:res:uq:heat:a}
\end{subfigure}

\begin{subfigure}[t]{\textwidth}
\centering
\includegraphics[width=0.2\linewidth]{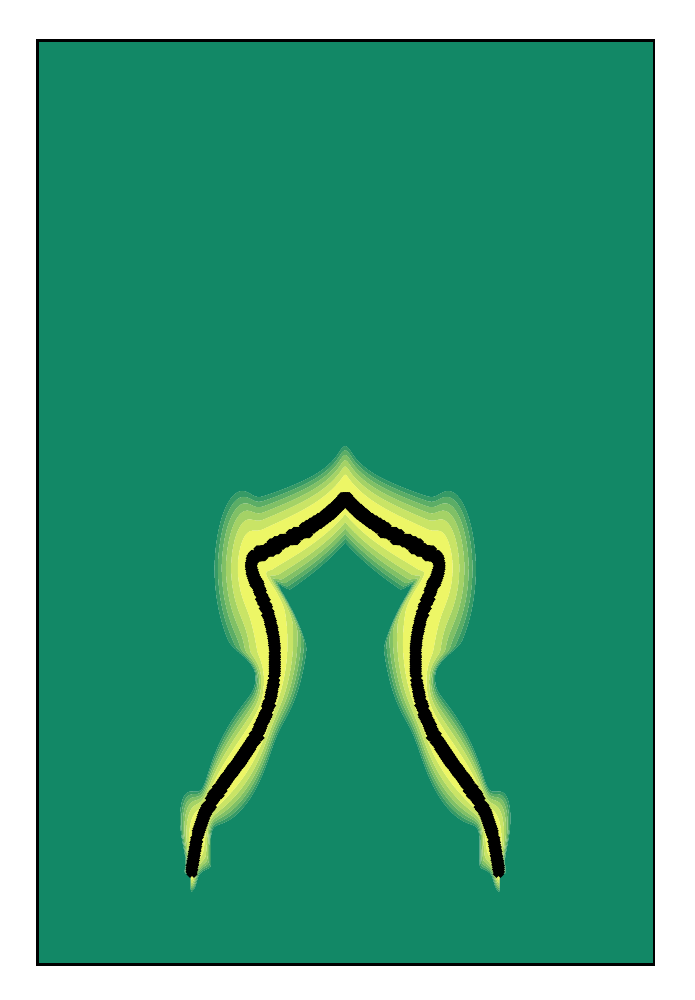}
\includegraphics[width=0.2\linewidth]{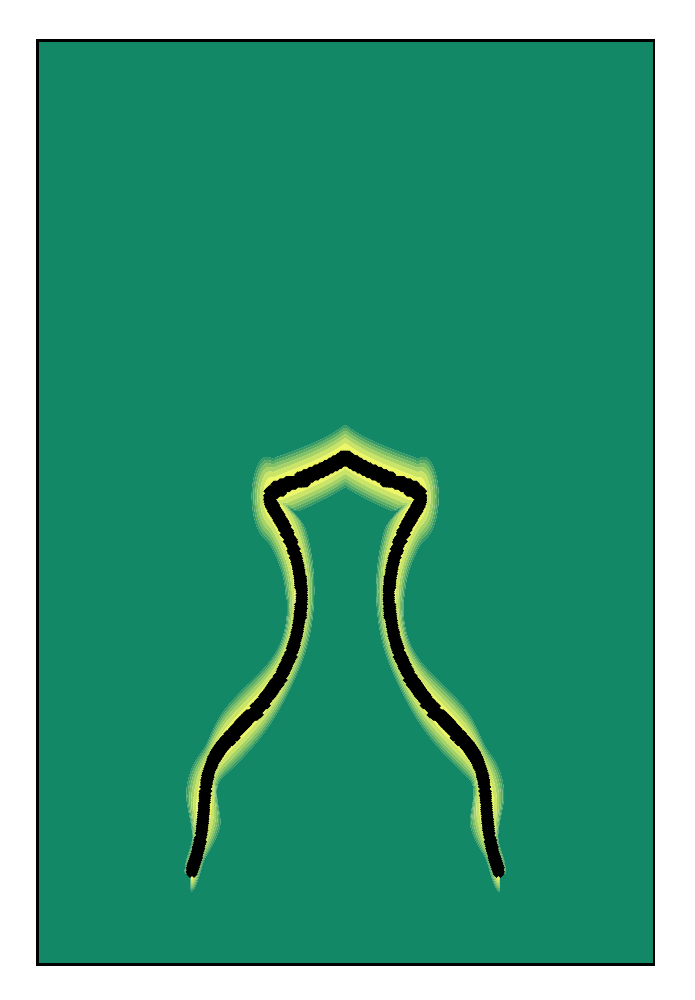}
\includegraphics[width=0.2\linewidth]{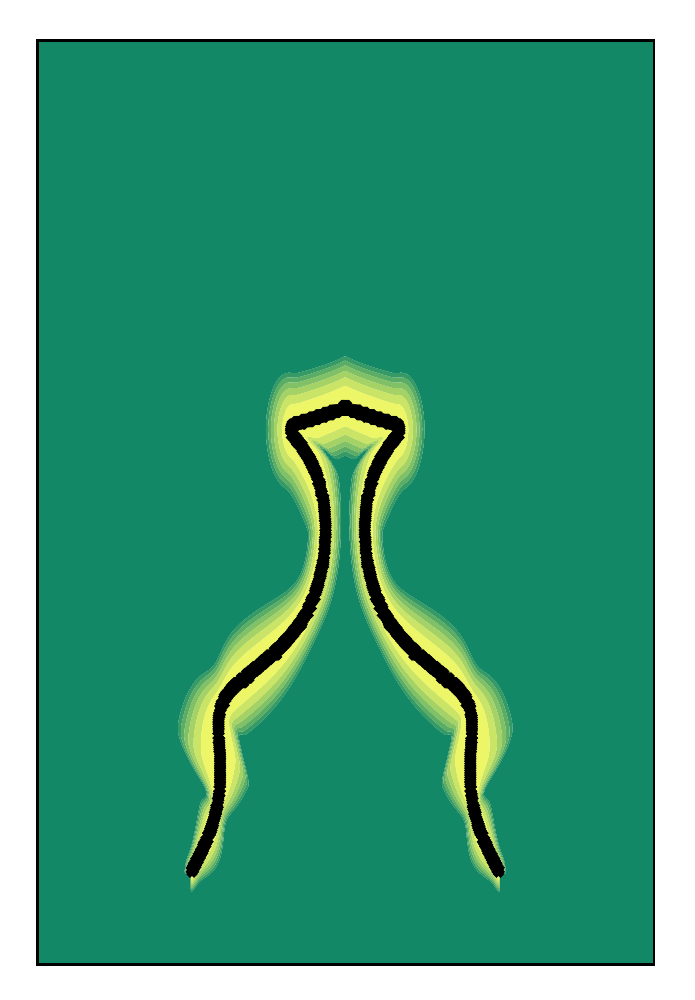}
\raisebox{-0.005\linewidth}{\includegraphics[width=0.27\linewidth]{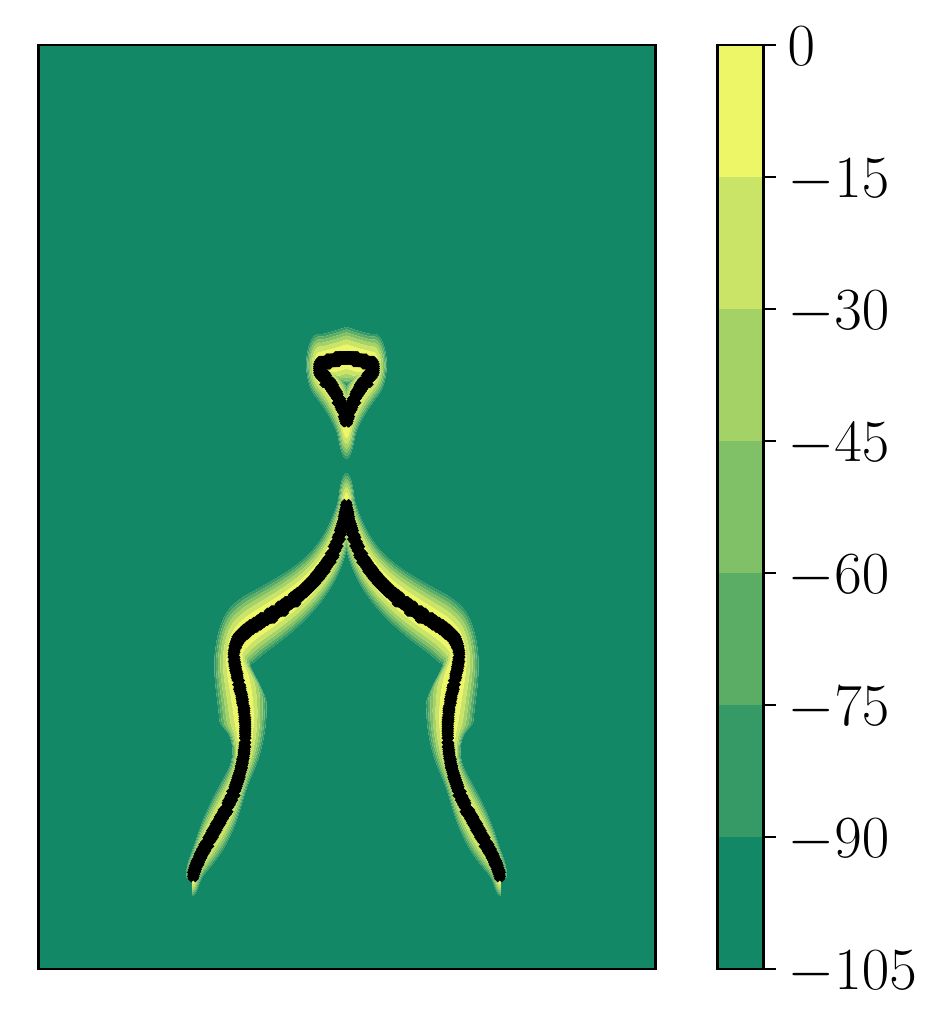}}
\caption{State estimation.}
\label{fig:res:uq:heat:b}
\end{subfigure}

\begin{subfigure}[t]{\textwidth}
\centering
\includegraphics[width=0.2\linewidth]{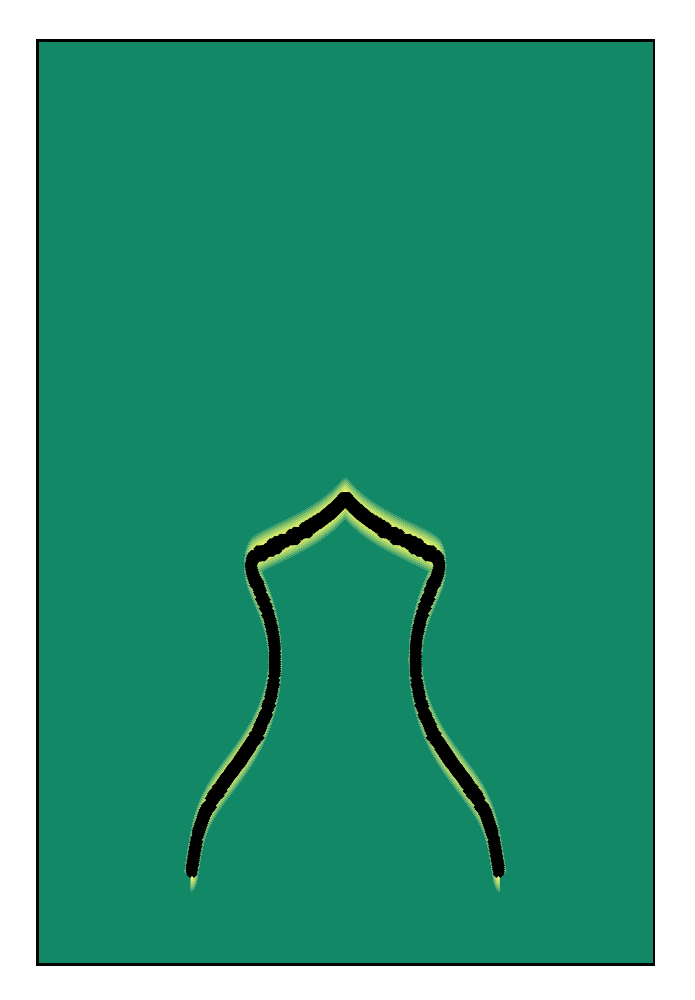}
\includegraphics[width=0.2\linewidth]{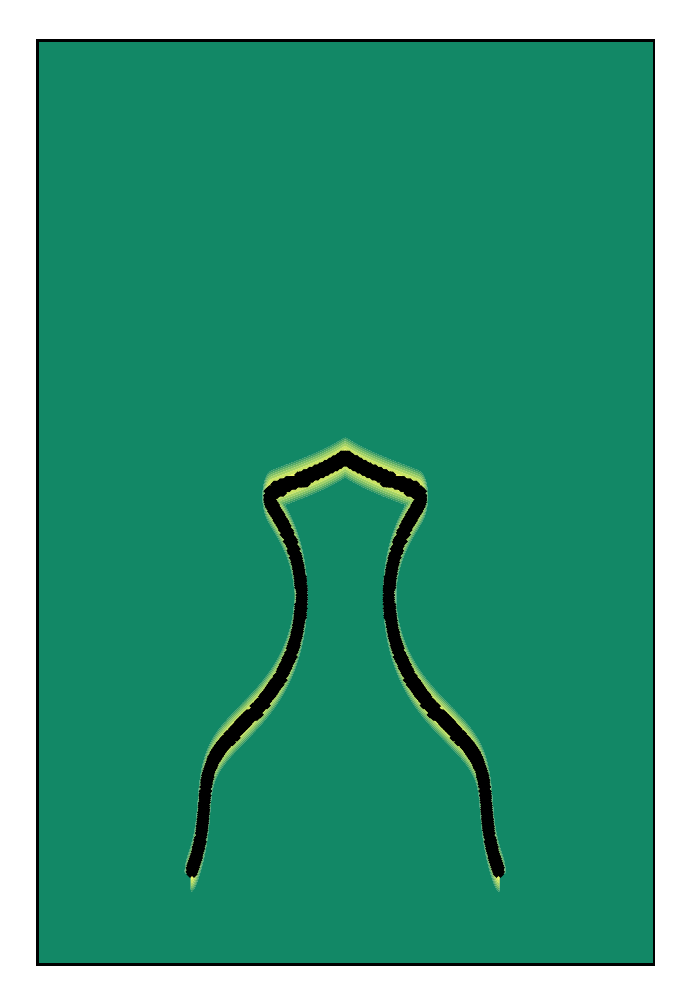}
\includegraphics[width=0.2\linewidth]{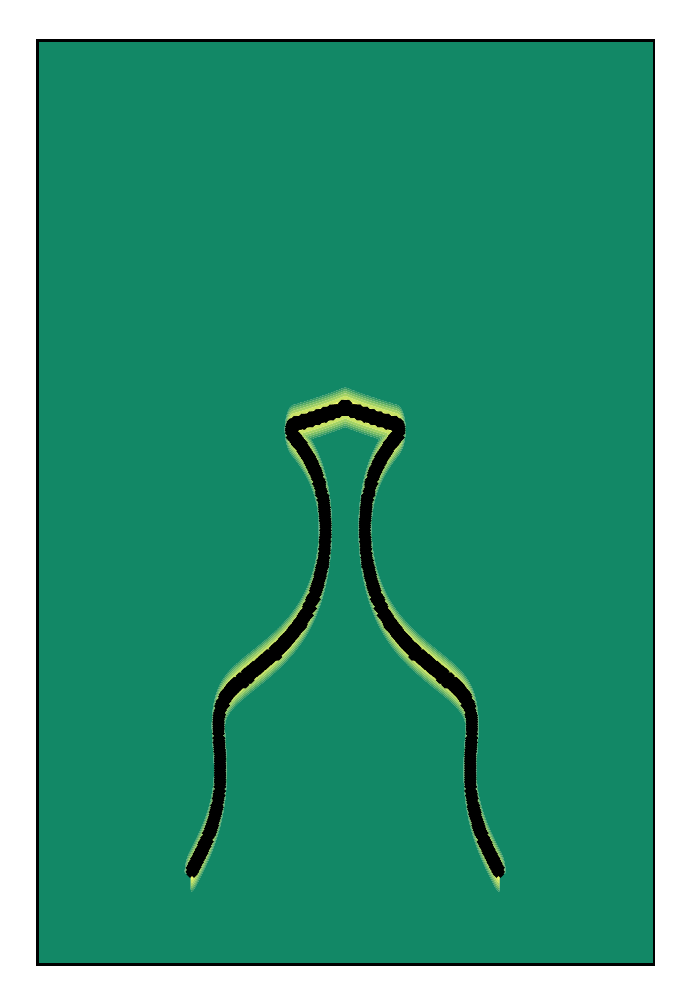}
\raisebox{-0.005\linewidth}{\includegraphics[width=0.27\linewidth]{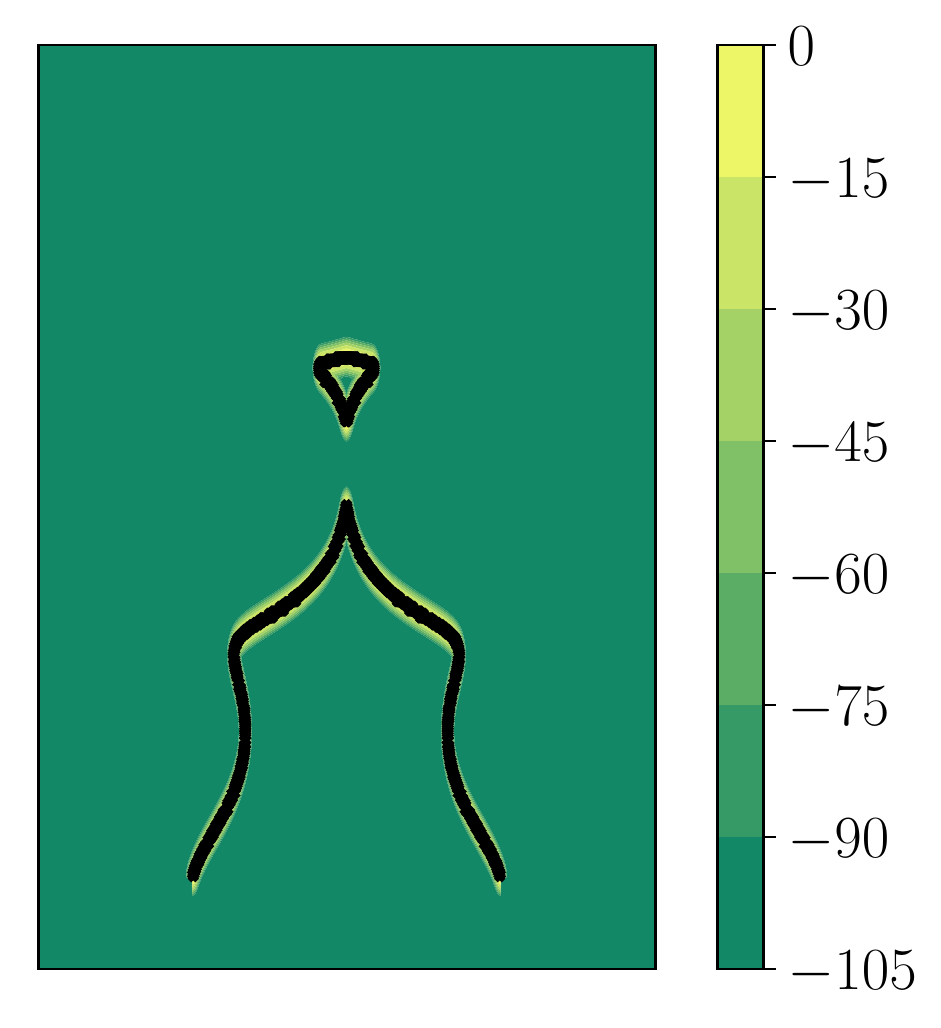}}
\caption{Combined state and parameter estimation.}
\label{fig:res:uq:heat:c}
\end{subfigure}
\caption{
Snapshots of the logarithm of the normalized likelihood~(Eq.~\eqref{eq:res:est:loglik}) over one period for the three variations of the twin experiment respectively.
The observations~(black) are extracted from a reference solution to the $G$-equation~(Fig.~\ref{fig:res:g}).
High-likelihood~(yellow) and low-likelihood~(green) regions are shown.
Each snapshot from left to right is distanced by a quarter of the forcing period.
}
\label{fig:res:uq:heat}
\end{figure}

A global, more quantitative measure for the uncertainty in the location of the flame surface is the root mean square (RMS) error, which is defined as the square-root of the trace of the covariance matrix of the ensemble~(Eq.~\eqref{eq:da:ensemble}):
\begin{equation}
\mathrm{RMS~error} = \sqrt{\frac{1}{n-1}\sum_{j=1}^N{\left(\psi_j - \overline{\psi}\right)^T\left(\psi_j - \overline{\psi}\right)}} \quad .
\end{equation}
In Fig.~\ref{fig:res:twin:error}, the RMS error is plotted over time for the three variations of the twin experiment from $t = 0$ for eight cycles.
The RMS error is initially zero in all three twin experiments because all members in each ensemble share the same initial condition.
In the Monte-Carlo experiment~(Fig.~\ref{fig:res:twin:error}, blue line), the RMS error subsequently grows until it reaches a high-uncertainty plateau.
Momentary spikes in the uncertainty occur approximately every 6000 timesteps, and coincide with the pinch-off events observed in the reference solution~(Fig.~\ref{fig:res:g}).
With state estimation~(Fig.~\ref{fig:res:twin:error}, orange line), the assimilation of data regularly suppresses the uncertainty as qualitatively observed in Fig.~\ref{fig:res:uq:heat:b}.
As a result, the predictions of the location of the flame surface significantly improve.
With combined state and parameter estimation~(Fig.~\ref{fig:res:twin:error}, green line), the predictions further improve.
Beginning with the first instance of data assimilation at timestep 1,000, the uncertainty steadily decreases until it reaches a low-uncertainty plateau after approximately 10,000 timesteps.
At this point, the state-augmented ensemble Kalman filter has optimally calibrated the parameters, and the state has reached a statistically stationary state.

\begin{figure}
\centering
\includegraphics[width=0.8\linewidth]{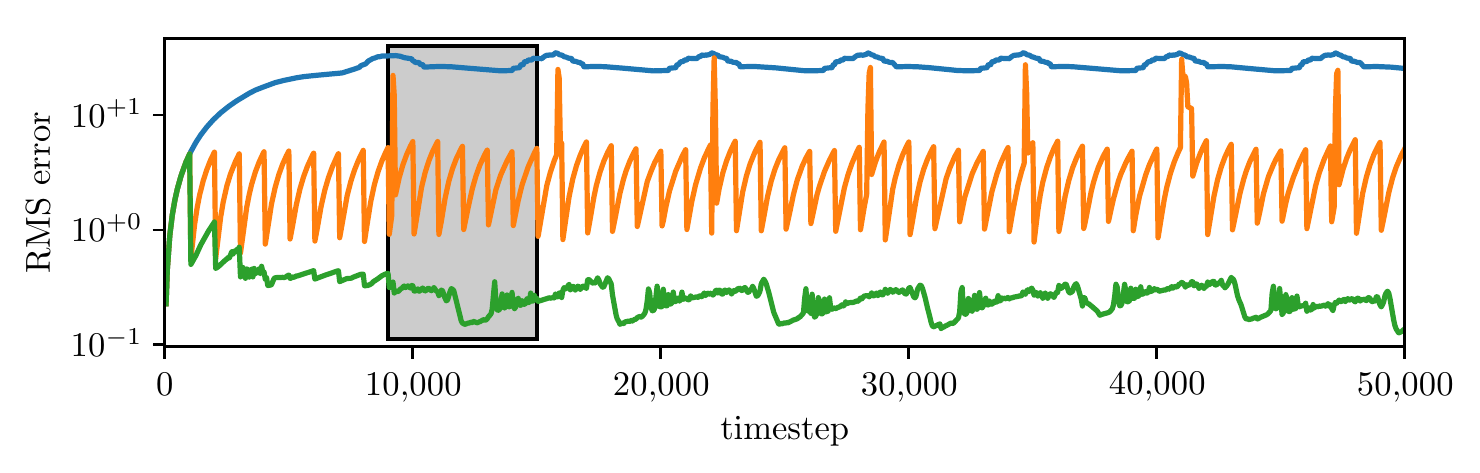}

\textcolor{mpl_blue}{\textbf{---}} Monte-Carlo experiment \quad \textcolor{mpl_orange}{\textbf{---}} state estimation \quad \textcolor{mpl_green}{\textbf{---}} combined state and parameter estimation
\caption{
Logarithmic plot of root mean square~(RMS) error over 50,000 timesteps for the Monte-Carlo experiment~(blue line) as well as the twin experiments with state estimation~(orange line) and combined state and parameter estimation~(green line).
6,283 timesteps correspond to one period of harmonic forcing.
Data is assimilated every 1,000 timesteps.
The black rectangle marks the time window depicted in Fig.~\ref{fig:res:uq:heat}.}
\label{fig:res:twin:error}
\end{figure}

To gain insight into the effect of combined state and parameter estimation, it is instructive to analyze the probability distributions in the parameters.
In Fig.~\ref{fig:res:twin:param}, the marginal distributions in the normalized residuals of~$K$ and~$\varepsilon$ are shown, signified by their means and three-sigma confidence levels at every timestep.
By timestep 15,000, no improvement is visible in either the estimates of the means or the uncertainties, which matches the evolution of the RMS error~(Fig.~\ref{fig:res:twin:error}).
While both~$K$ and~$\varepsilon$ quickly converge to the values used in the reference simulation~(Tab.~\ref{tab:res:param}), the low uncertainty in~$K$ compared to~$\varepsilon$ reflects the physical significance of this parameter:
The nonlinear dynamics of a premixed flame are highly sensitive to the timings of pinch-off events~\cite{Juniper2018}, which in turn depend on the phase speed at which perturbations travel along with the flame.
In our reduced-order model based on the $G$-equation, the phase speed is regulated by the parameter~$K$, which makes it highly observable even in the presence of observation noise.
Similar to sensitivity analysis~\cite{Magri2016, Magri2016a,Magri2019_amr}, this physical insight is gained from the inspection of the uncertainties, which here exceed the residuals between the means and the set of reference values by several orders of magnitude.

%When the first fuel-air pocket is pinched off, at timestep 4,000, the uncertainty grows rapidly until the combined state and parameter estimation updates the state and the parameters again.
%The update step coincides with the moment of pinch-off, which is also a moment of high uncertainty.
%Thus the parameters can be found very accurately, which leads to a low-uncertainty plateau.
%High uncertainties at later pinch-off events are suppressed without much change in the parameters.

\begin{figure}
\centering
\includegraphics[width=0.45\linewidth]{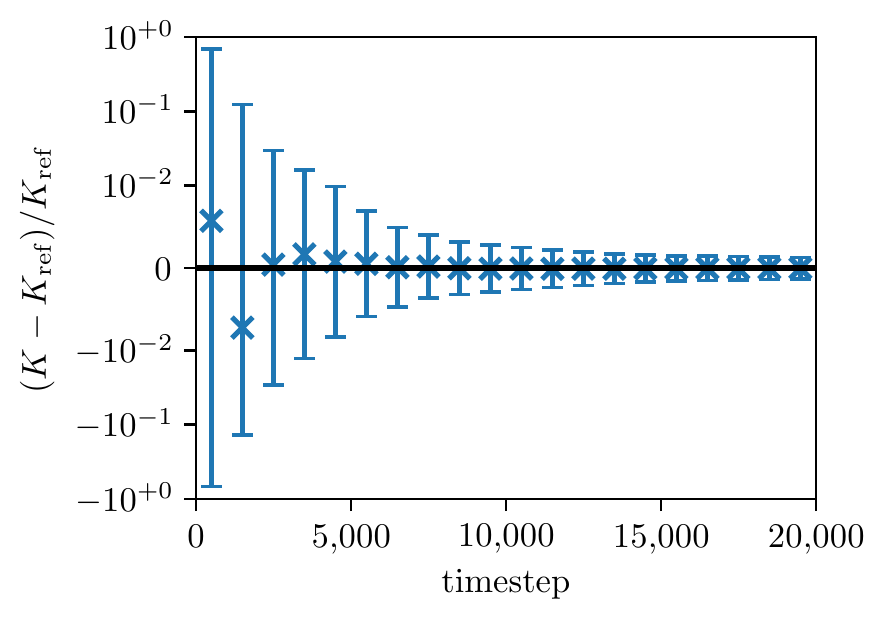}
~
\includegraphics[width=0.45\linewidth]{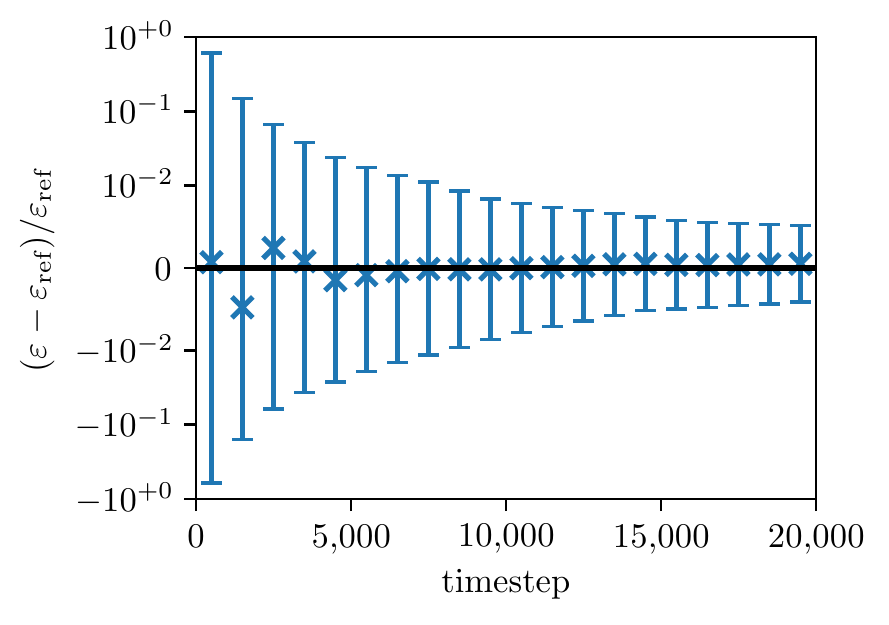}
\caption{
Logarithmic plots of normalized residuals of~$K$~(left) and~$\varepsilon$~(right) over 20,000 timesteps for combined state and parameter estimation; linear scales for normalized likelihoods between~$\pm 10^{-2}$.
The marginal distributions in~$K$ and~$\varepsilon$ are represented by their means~(crosses) and three-sigma confidence levels~(error bars).
}
\label{fig:res:twin:param}
\end{figure}

Finally, combined state and parameter estimation using the ensemble Kalman smoother is performed~(Theorem~\ref{thm:da:ensemble_kalman_smoother}).
In Fig.~\ref{fig:res:twin:param}, the marginal distributions in the normalized residuals of~$K$ and~$\varepsilon$ are shown for both the ensemble Kalman filter and smoother.
The effect of the ensemble Kalman filter has already been described in Fig.~\ref{fig:res:twin:param}.
In the forward-backward implementation, the ensemble Kalman smoother takes the last solution of the ensemble Kalman filter and works itself backwards in time~(Theorem~\ref{thm:da:bayes_smoother}).
In terms of information theory, the ensemble Kalman smoother takes at every timestep~$k$ all observations, from the past, present and future, into account compared to the ensemble Kalman filter, which only takes the past and the present into account.
Statistically speaking, the smoothed distributions are more strictly conditioned than the filtered distributions, over~$y_{1:N}$ compared to~$y_{1:k}$~(Tab.~\ref{tab:da:param:state_estimation}).
This surplus in information is evident in the form of lower error bars not just towards the end of the combined state and parameter estimation, but extending all the way to the very first instances of data assimilation.
The evaluation of Eq.~\ref{thm:da:ensemble_kalman_smoother} does not rely on the solution of the governing equations, but instead relies on the predicted and filtered distributions in storage.
The ensemble Kalman smoother is thus a computationally inexpensive tool to quantify the uncertainties especially at the beginning of the simulation, where little retrospective data is available.
In analogy to direct-adjoint looping~\cite{Juniper2011, Kerswell2018}, combined state and parameter estimation based on filtering and smoothing can be used to obtain otherwise inevitably ad-hoc initial conditions and parameters.

\begin{figure}
\centering
\includegraphics[width=0.8\linewidth]{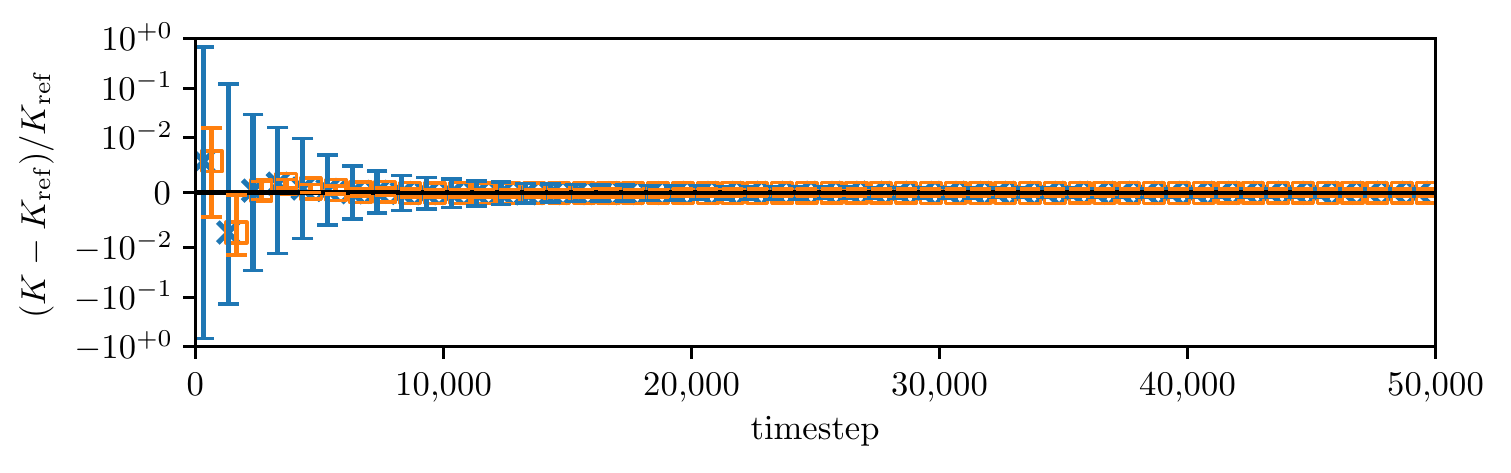}

\includegraphics[width=0.8\linewidth]{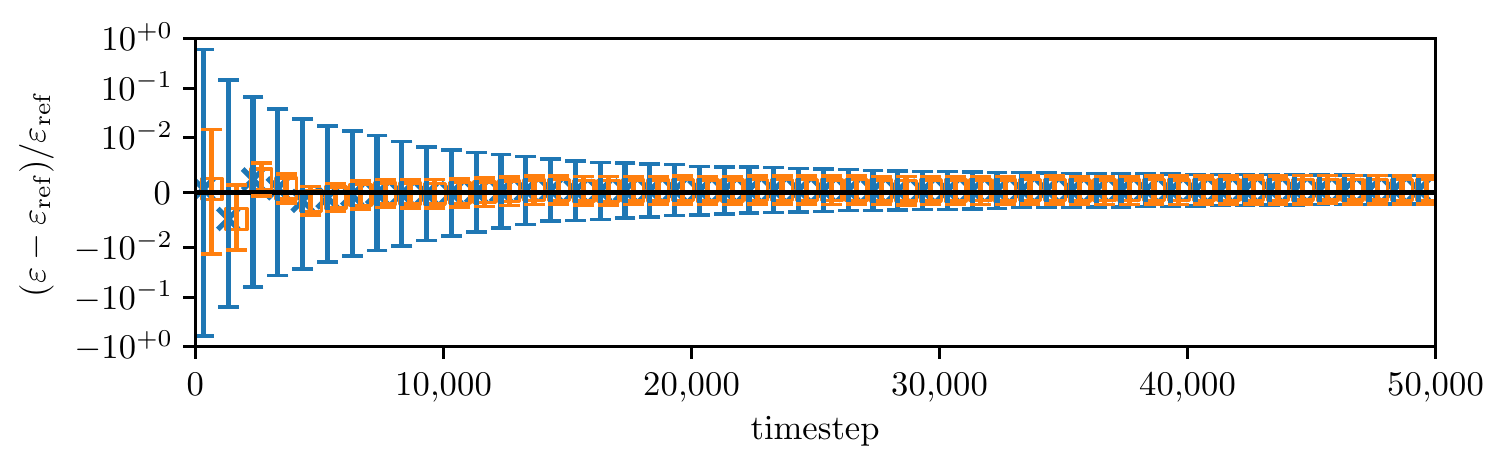}

\textcolor{mpl_blue}{\textbf{---}} ensemble Kalman filter \quad \textcolor{mpl_orange}{\textbf{---}} ensemble Kalman smoother
\caption{
Logarithmic plots of normalized residuals of~$K$~(top) and~$\varepsilon$~(bottom) over 50,000 timesteps for combined state and parameter estimation using the ensemble Kalman filter~(blue cross) and the ensemble Kalman smoother~(orange square); linear scales for normalized likelihoods between~$\pm 10^{-2}$.
}
\label{fig:res:twin:smooth}
\end{figure}

\FloatBarrier

%------------------------------------------------------------------------------

%\subsection{Stochastic convergence}
%
%\begin{figure}
%\centering
%\includegraphics[width=0.8\linewidth]{error_ensemble.pdf}
%\caption{Root mean square (RMS) error over time for combined state and parameter estimation with different experimental errors ($\sigma_\epsilon^2 = 10^{-4}$ blue line, $\sigma_\epsilon^2 = 10^{-6}$ orange dash, $\sigma_\epsilon^2 = 10^{-8}$ green dash-dot).}
%\label{fig:res:twin:error_ensemble}
%\end{figure}
%
%\begin{figure}
%\centering
%\includegraphics[width=0.45\linewidth]{K_ensemble.pdf}
%~
%\includegraphics[width=0.45\linewidth]{varepsilon_ensemble.pdf}
%\caption{Residuals of (a) $K$ and (b) $\varepsilon$ over time for combined state and parameter estimation with different experimental errors ($\sigma_\epsilon^2 = 10^{-4}$ blue cross, $\sigma_\epsilon^2 = 10^{-6}$ orange triangle, $\sigma_\epsilon^2 = 10^{-8}$ green square).}
%\label{fig:res:twin:param_ensemble}
%\end{figure}
%
%\FloatBarrier

%------------------------------------------------------------------------------

\subsection{Effect of observation noise}

In a first parameter study, the observation error~$\sigma_\epsilon$ is varied to investigate the effect of observation noise on the level-set data assimilation framework.
In Section~\ref{sec:res:est}, it has been established that combined state and parameter estimation eventually leads to a statistically stationary state.
In Fig.~\ref{fig:res:twin:error_epsilon}, the same RMS error is plotted over time for different observation errors~$\sigma_\epsilon$.
This reveals that the sustained low uncertainties are epistemic in nature~\cite{Kiureghian2009}:
Every reduction in the observation error~$\sigma_\epsilon$ by one order of magnitude reduces the RMS error by one order of magnitude.
It is clear that the observation error poses a lower epistemic bound on how low the uncertainty can be reduced by combined state and parameter estimation.
%most likely attributed to numerical error.

%Hence, uncertainty in the optimally calibrated parameters remains unless the experimental error vanishes.
%Furthermore, moments of pinch-off regularly induce significant uncertainty regardless of the experimental error.
%In each case, combined state and parameter estimation intervenes, and caps the RMS error before it grows out of bounds.
%The results from data assimilation confirm the extreme sensitivity observed elsewhere that undermines the a-priori calibration of models in favor of optimal, on-the-fly calibration~\cite{Juniper2018}.
%In Fig.~\ref{fig:res:twin:param_epsilon}, the residuals of $K$ and $\varepsilon$ are plotted over time.
%While $\varepsilon$ remains virtually constant after 10,000 timesteps, $K$ significantly varies in sync with the pinch-off cycles.
%In terms of an inverse problem, the uncertainty in the state (Fig.~\ref{fig:res:twin:error_param}) has propagated upstream to the parameters and mainly to $K$.
%The results from combined state and parameter estimation confirm the physical significance of the perturbation convection speed $K$ as a model parameter~\cite{Kashinath2013}.

\begin{figure}
\centering
\includegraphics[width=0.8\linewidth]{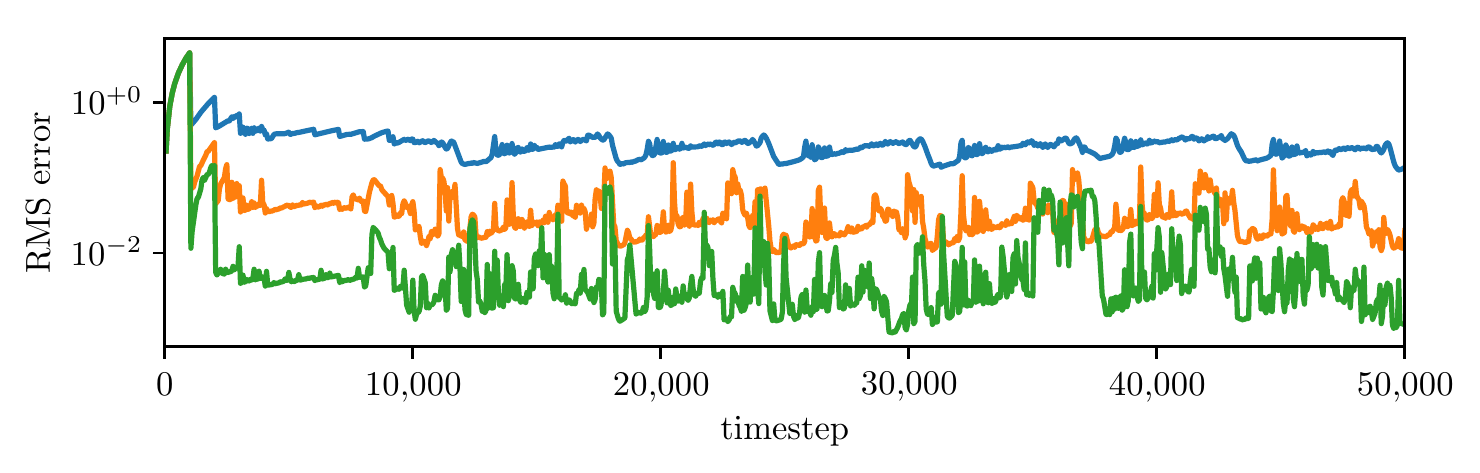}

\textcolor{mpl_blue}{\textbf{---}} $\sigma_\epsilon = 60\,\mu\mathrm{m}$ \quad \textcolor{mpl_orange}{\textbf{---}} $\sigma_\epsilon = 6\,\mu\mathrm{m}$ \quad \textcolor{mpl_green}{\textbf{---}} $\sigma_\epsilon = 0.6\,\mu\mathrm{m}$
\caption{Logarithmic plot of root mean square (RMS) error over 50,000 timesteps for combined state and parameter estimation with different observation errors~$\sigma_\epsilon$~($\sigma_\epsilon = 60\,\mu\mathrm{m}$ blue line, $\sigma_\epsilon = 6\,\mu\mathrm{m}$ orange line, $\sigma_\epsilon = 0.6\,\mu\mathrm{m}$ green line).}
\label{fig:res:twin:error_epsilon}
\end{figure}

In Fig.~\ref{fig:res:twin:param_epsilon}, the marginal distributions in the normalized residuals of~$K$ and~$\varepsilon$ are plotted over time for the different observation errors.
The means of the residuals quickly vanish for all three observation errors~$\sigma_\epsilon$.
For the lower observation errors, the error bars increasingly fail to contain the zero residual.
Considering that three sigmas correspond to a confidence of 99.7\,\% under the assumption of normal distributions~(Section~\ref{sec:da:kalman}), the frequency at which the error bars fail leads to the conclusion that the level-set data assimilation framework underpredicts the uncertainties for low observation errors.
This is in agreement with the theoretical analysis of the level-set data assimilation framework for the two-dimensional test case~(Section~\ref{sec:ls:ex2d}):
Although the individual simulations predict the formation of sharp cusps as the perturbations travel on the respective flame surfaces, the cusps are rounded in the mean of the ensemble~(Fig.~\ref{fig:ls:ex2d:prior}, left).
Furthermore, the observable sharpness of a cusp depends on the local number of observations of points on the flame surface.
A decrease in the observation error~$\sigma_\epsilon$ without an increase in the number of observations amounts to a relative decrease in the resolution of the observed cusp~(Fig.~\ref{fig:ls:ex2d:post:n}).
Therefore, the local values of the eikonal fields of the individual simulations significantly deviate from normal distributions~(Fig.~\ref{fig:ls:ex2d:prior}, right), which leads to incorrect distributions and uncertainties.

\begin{figure}
\centering
\includegraphics[width=0.45\linewidth]{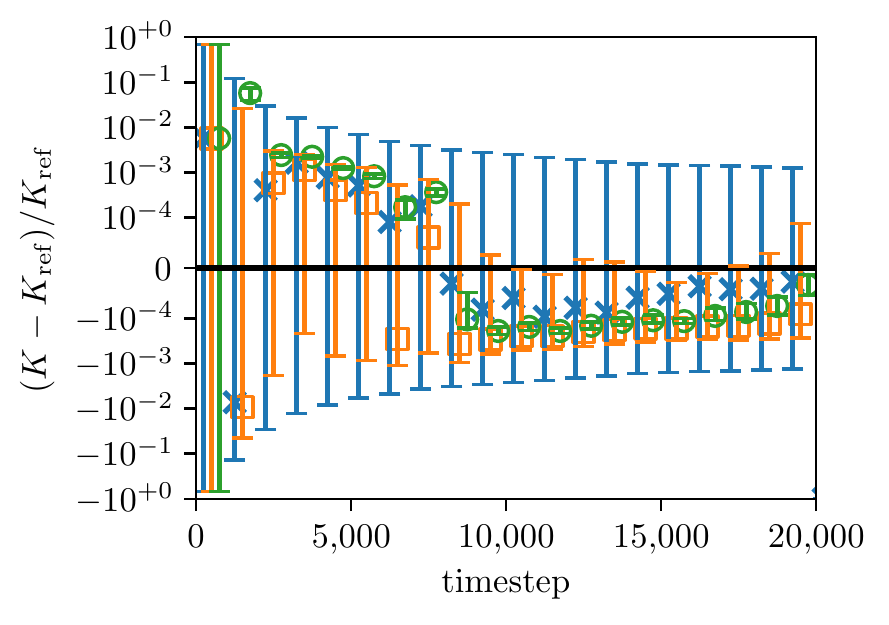}
~
\includegraphics[width=0.45\linewidth]{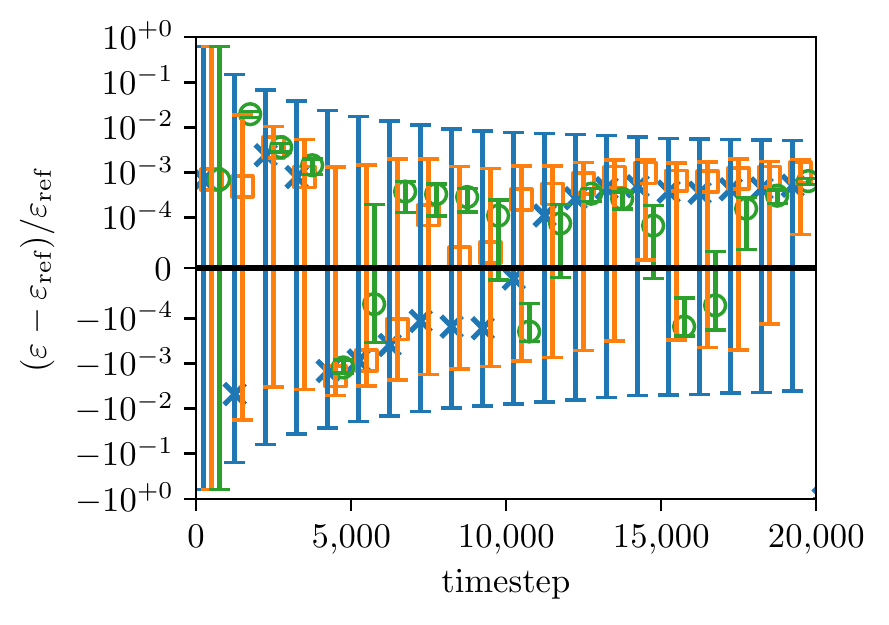}

\textcolor{mpl_blue}{\textbf{---}} $\sigma_\epsilon = 60\,\mu\mathrm{m}$ \quad \textcolor{mpl_orange}{\textbf{---}} $\sigma_\epsilon = 6\,\mu\mathrm{m}$ \quad \textcolor{mpl_green}{\textbf{---}} $\sigma_\epsilon = 0.6\,\mu\mathrm{m}$
\caption{Logarithmic plots of normalized residuals of $K$~(left) and $\varepsilon$~(right) over 20,000 timesteps for combined state and parameter estimation with different observation errors~($\sigma_\epsilon = 60\,\mu\mathrm{m}$ blue cross, $\sigma_\epsilon = 6\,\mu\mathrm{m}$ orange square, $\sigma_\epsilon = 0.6\,\mu\mathrm{m}$ green circle); linear scales for normalized likelihoods between~$\pm 10^{-4}$.}
\label{fig:res:twin:param_epsilon}
\end{figure}

\FloatBarrier

%------------------------------------------------------------------------------

\subsection{Effect of number of observations}

In a second parameter study, subsamples of the observation points are used to investigate the effect of the number of observations on the level-set data asssimilation framework.
Furthermore, this serves to illustrate the superiority of our analytic view over the geometric view in data assimilation~(Section~\ref{sec:ls}):
In the geometric view, the observed interface must be parametric or at least as highly resolved as the grid used in the level-set method to allow the optimal interpolation of a sufficiently large number of points on the predicted interfaces.
This restriction does not exist in the analytic view.
Hence, the effect of dispersed observation points on the level-set data assimilation framework is also studied here.

In Fig.~\ref{fig:res:twin:error_random}, the RMS error is plotted over time for different randomly sampled subsets of the observation points.
In Section~\ref{sec:res:est}, it has been established that combined state and parameter estimation eventually leads to a statistically stationary state after approximately 10,000~timesteps.
With 10\,\% of the observation points, it takes more than 30,000~timesteps to reach a comparable low-uncertainty plateau.
With 2\,\% of the observation points, the RMS error is still decreasing after 50,000~timesteps while repeatedly failing to correctly predict the pinch-off events for the whole ensemble.

\begin{figure}
\centering
\includegraphics[width=0.8\linewidth]{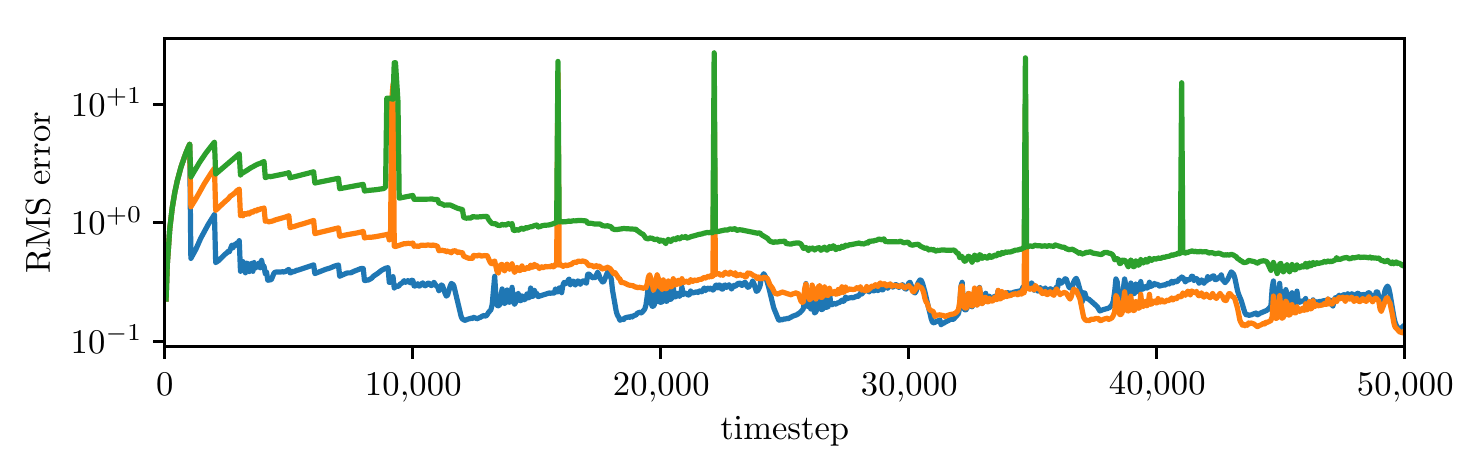}

\textcolor{mpl_blue}{\textbf{---}} $100\,\%$ observed \quad \textcolor{mpl_orange}{\textbf{---}} $10\,\%$ observed \quad \textcolor{mpl_green}{\textbf{---}} $2\,\%$ observed
\caption{Logarithmic plot of root mean square (RMS) error over 50,000 timesteps for combined state and parameter estimation with different subsamples of observation points~($100\,\%$ blue line, $10\,\%$ orange line, $2\,\%$ green line).}
\label{fig:res:twin:error_random}
\end{figure}

In Fig.~\ref{fig:res:twin:param_random}, the marginal distributions in the normalized residuals of~$K$ and~$\varepsilon$ are plotted over time for the different subsamples.
While the means of the residuals quickly vanish for all three subsamples, the confidence levels for the incomplete subsamples improve less rapidly than for the original sample of observation points.
This is in agreement with the theoretical analysis of the level-set data assimilation framework for the two-dimensional test case~(Section~\ref{sec:ls:ex2d}):
In particular for the 2\,\%-subsample, certain features in the observed interfaces remain underresolved at the individual timesteps~(Fig.~\ref{fig:ls:ex2d:post:1}).
Nevertheless, the ensemble at a timestep~$k$ representing the filtered probability distribution~$p(x_k, \theta \mid y_{1:k}, f)$~(Tab.~\ref{tab:da:param:state_estimation}) reflects the knowledge of the observations at all previous timesteps due to the Markov chain property of the probabilistic state space model~(Section~\ref{sec:da:state_space}).
In theory, the lack of resolution in the observations is compensated by more instances of data assimilation to accumulate the same amount of information.
The result is a similar statistically stationary state reached at a later time.

\begin{figure}
\centering
\includegraphics[width=0.45\linewidth]{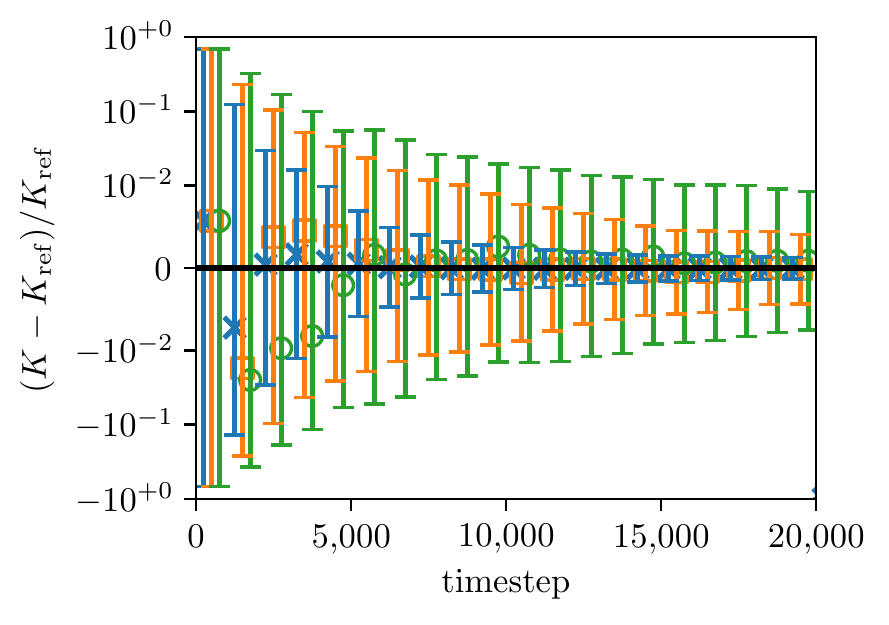}
~
\includegraphics[width=0.45\linewidth]{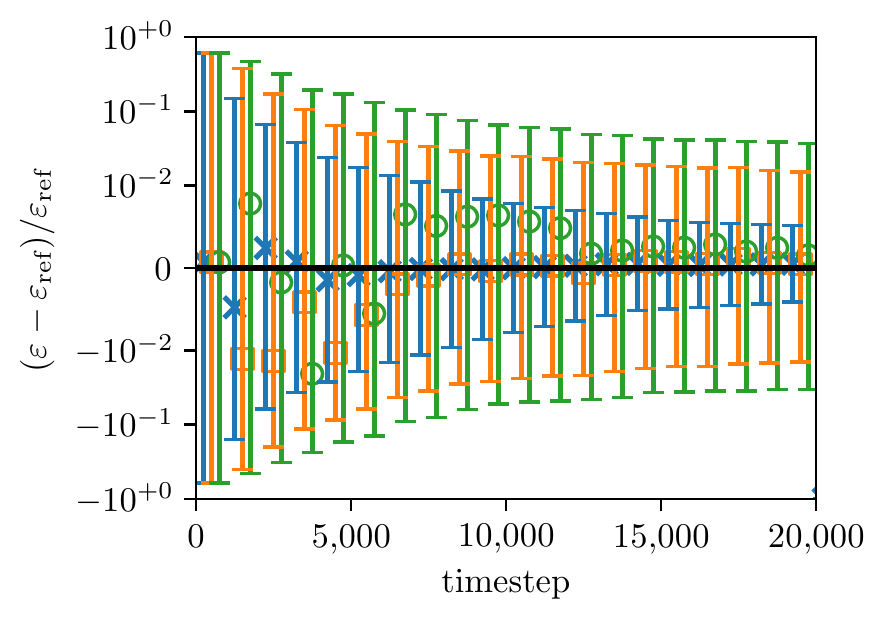}

\textcolor{mpl_blue}{\textbf{---}} $100\,\%$ observed \quad \textcolor{mpl_orange}{\textbf{---}} $10\,\%$ observed \quad \textcolor{mpl_green}{\textbf{---}} $2\,\%$ observed
\caption{Logarithmic plots of normalized residuals of $K$~(left) and $\varepsilon$~(right) over 20,000 timesteps for combined state and parameter estimation with different subsamples of observation points~($100\,\%$ blue cross, $10\,\%$ orange square, $2\,\%$ green circle); linear scales for normalized likelihoods between~$\pm 10^{-2}$.}
\label{fig:res:twin:param_random}
\end{figure}

\FloatBarrier

%------------------------------------------------------------------------------

\section{Conclusion}
\label{sec:conclusion}
Reduced-order models based on level-set methods are widely used tools to qualitatively capture and track the nonlinear dynamics of an interface. In this paper, we enhance such a level-set model with a statistical learning technique in order to make the model quantitatively predictive. The statistical learning method is Bayesian, so the uncertainty of the outputs is naturally included in this framework. The main ingredients of the level-set statistical learning are
(i) a reduced-order model of an interface based on a level-set method with a law of motion;
(ii) the ensemble Kalman filter and smoother; and
(iii) external reference data with its uncertainty. The ensemble Kalman filter uses Bayes' rule as its first principle while assuming normal distributions when data is assimilated. The output of the algorithm is the combined state and parameters estimation, which provides the most likely position of the interface and the model's parameters given observations. These observations can originate from experimental data or high-fidelity simulations.

To show and verify the algorithm, we perform twin experiments, where the observations are produced by the same code. The assimilation of data from experimental data or high-fidelity simulation requires virtually no modification to the proposed algorithm. The algorithm is applied to three test cases, which are of increasing complexity.

First, a one-dimensional case is presented. We choose the Hamilton-Jacobi generating functions over the whole domain to perform data assimilation, which is called the``analytic approach''. The one-dimensional example shows that the analytic approach overcomes the inaccuracy of the set-theoretic approach, which is based on characteristic functions. Moreover, it is argued that the analytic approach is more straightforward to implement that the geometric approach when moving to higher dimensions.

Second, a two-dimensional, time-independent case is presented, where the position of a sharp corner has to be predicted. It is shown that the level-set data assimilation framework proposed is fully front capturing. The effect that the number of observations has on capturing a sharp corner is investigated. The sensitivity of the analysis solution to each observation is calculated. The one- and two-dimensional cases show that the analytic approach we adopt is more robust than the geometric and set-theoretic approaches that have been used in the past.

Third, the time-dependent nonlinear dynamics of a conical forced premixed flame is studied. The most uncertain states are pinch-offs and the formation of sharp cusps, which are highly nonlinear topological changes of the interface. The two most uncertain parameters are the amplitude of the velocity perturbation at the flame's base and the perturbation phase speed that wrinkles the flame. In the twin experiment, the combined state and parameter estimation fully recovers the reference solution, which validates the algorithm. Finally, the effect of uncertainty and number of observations is analysed. The uncertainty on the observation most influences the detection of the pinch-off events or the formation of sharp cusps. The data assimilation analytic approach is shown to be versatile because it can assimilate data any time it becomes available, avoiding the cumbersome parameterization required in a geometric approach.

The level-set data assimilation framework is applicable to a number of problems in computational physics involving the motion of an interface~\cite{Sethian2001}. For example, a physical model coupling the dynamics of a premixed flame with an acoustic model of a duct may be used for a data-driven investigation of the Rijke tube~\cite{Fleifil1996, Dowling1999}. More generally, a number of computational fluid dynamics (CFD) models of premixed combustion, short of direct numerical simulations, combine the laws of motion with conservation laws for the inert flow, e.g. the Bray-Moss-Libby model or the flamelet model~\cite{Peters2000}. Although it is tempting to directly apply data assimilation to the primitive variables in the Navier-Stokes equations, the discontinuities in density, temperature and species concentrations in light of the theoretical analysis of the set-theoretic approach to level-set data assimilation raise doubts as to whether the kinematics of a premixed flame, which are crucial to thermoacoustics, will be correctly predicted.

Future work will focus on extending the applications of the proposed level-set data assimilation framework. When applied to experiments, the effects of imperfect models will need to be mitigated~\cite{Reich2015}, for example with techniques such as covariance inflation~\cite{Anderson1999, Ott2004} and localization~\cite{Gaspari1999}. Future work will also assess the uncertainties in the model, as opposed to uncertainties in the state and parameters.
%
%[The danger of the next line is that a reviewer will ask us to do this in the current paper, which would be a huge amount of extra work. To be safe, I would leave it out.]
%Finally, the ensemble Kalman filter in the level-set data assimilation framework will be benchmarked against its alternatives, e.g. four-dimensional variational data assimilation [76, 77] and particle filters [24].

%------------------------------------------------------------------------------

\section*{Acknowledgments}

The authors would like to thank the Data Assimilation Research Centre at the University of Reading for the useful discussions at their 2018 Data Assimilation Training.
H.\ Yu is supported the Cambridge Commonwealth, European \& International Trust under a Schlumberger Cambridge International Scholarship.
L.\ Magri is supported by the Royal Academy of Engineering under the RAEng Research Fellowships scheme and partially by the visiting fellowship at the Technical University of Munich -- Institute for Advanced Study, funded by the German Excellence Initiative and the European Union Seventh Framework Programme under grant agreement n. 291763.

%------------------------------------------------------------------------------

%\bibliography{refs}

\begin{thebibliography}{88}
\expandafter\ifx\csname natexlab\endcsname\relax\def\natexlab#1{#1}\fi
\providecommand{\url}[1]{\texttt{#1}}
\providecommand{\href}[2]{#2}
\providecommand{\path}[1]{#1}
\providecommand{\DOIprefix}{doi:}
\providecommand{\ArXivprefix}{arXiv:}
\providecommand{\URLprefix}{URL: }
\providecommand{\Pubmedprefix}{pmid:}
\providecommand{\doi}[1]{\href{http://dx.doi.org/#1}{\path{#1}}}
\providecommand{\Pubmed}[1]{\href{pmid:#1}{\path{#1}}}
\providecommand{\bibinfo}[2]{#2}
\ifx\xfnm\relax \def\xfnm[#1]{\unskip,\space#1}\fi
%Type = Article
\bibitem[{Sethian(2001)}]{Sethian2001}
\bibinfo{author}{J.~Sethian},
\newblock \bibinfo{title}{Evolution, {Implementation}, and {Application} of
  {Level} {Set} and {Fast} {Marching} {Methods} for {Advancing} {Fronts}},
\newblock \bibinfo{journal}{Journal of Computational Physics}
  \bibinfo{volume}{169} (\bibinfo{year}{2001}) \bibinfo{pages}{503--555}.
%Type = Article
\bibitem[{Gibou et~al.(2018)Gibou, Fedkiw, and Osher}]{Gibou2018}
\bibinfo{author}{F.~Gibou}, \bibinfo{author}{R.~Fedkiw},
  \bibinfo{author}{S.~Osher},
\newblock \bibinfo{title}{A review of level-set methods and some recent
  applications},
\newblock \bibinfo{journal}{Journal of Computational Physics}
  \bibinfo{volume}{353} (\bibinfo{year}{2018}) \bibinfo{pages}{82--109}.
%Type = Article
\bibitem[{Osher and Sethian(1988)}]{Osher1988}
\bibinfo{author}{S.~Osher}, \bibinfo{author}{J.~A. Sethian},
\newblock \bibinfo{title}{Fronts propagating with curvature-dependent speed:
  {Algorithms} based on {Hamilton}-{Jacobi} formulations},
\newblock \bibinfo{journal}{Journal of Computational Physics}
  \bibinfo{volume}{79} (\bibinfo{year}{1988}) \bibinfo{pages}{12--49}.
%Type = Article
\bibitem[{Sethian and Smereka(2003)}]{Sethian2003}
\bibinfo{author}{J.~A. Sethian}, \bibinfo{author}{P.~Smereka},
\newblock \bibinfo{title}{Level {Set} {Methods} for {Fluid} {Interfaces}},
\newblock \bibinfo{journal}{Annual Review of Fluid Mechanics}
  \bibinfo{volume}{35} (\bibinfo{year}{2003}) \bibinfo{pages}{341--372}.
%Type = Book
\bibitem[{Sethian(1999)}]{Sethian1999a}
\bibinfo{author}{J.~A. Sethian}, \bibinfo{title}{Level set methods and fast
  marching methods: evolving interfaces in computational geometry, fluid
  mechanics, computer vision, and materials science},
  number~\bibinfo{number}{3} in \bibinfo{series}{Cambridge monographs on
  applied and computational mathematics}, \bibinfo{edition}{2nd ed} ed.,
  \bibinfo{publisher}{Cambridge University Press}, \bibinfo{address}{Cambridge,
  U.K. ; New York}, \bibinfo{year}{1999}.
%Type = Article
\bibitem[{Osher and Fedkiw(2001)}]{Osher2001}
\bibinfo{author}{S.~Osher}, \bibinfo{author}{R.~P. Fedkiw},
\newblock \bibinfo{title}{Level {Set} {Methods}: {An} {Overview} and {Some}
  {Recent} {Results}},
\newblock \bibinfo{journal}{Journal of Computational Physics}
  \bibinfo{volume}{169} (\bibinfo{year}{2001}) \bibinfo{pages}{463--502}.
%Type = Book
\bibitem[{Kollmann(2010)}]{Kollmann2010}
\bibinfo{author}{W.~Kollmann}, \bibinfo{title}{Fluid {Mechanics} in {Spatial}
  and {Material} {Description}}, \bibinfo{publisher}{University Readers},
  \bibinfo{year}{2010}.
%Type = Book
\bibitem[{Peters(2000)}]{Peters2000}
\bibinfo{author}{N.~Peters}, \bibinfo{title}{Turbulent {Combustion}},
  \bibinfo{publisher}{Cambridge University Press}, \bibinfo{year}{2000}.
  \URLprefix \url{http://ebooks.cambridge.org/ref/id/CBO9780511612701}.
  \DOIprefix\doi{10.1017/CBO9780511612701}.
%Type = Article
\bibitem[{Fleifil et~al.(1996)Fleifil, Annaswamy, Ghoneim, and
  Ghoniem}]{Fleifil1996}
\bibinfo{author}{M.~Fleifil}, \bibinfo{author}{A.~M. Annaswamy},
  \bibinfo{author}{Z.~A. Ghoneim}, \bibinfo{author}{A.~F. Ghoniem},
\newblock \bibinfo{title}{Response of a laminar premixed flame to flow
  oscillations: {A} kinematic model and thermoacoustic instability results},
\newblock \bibinfo{journal}{Combustion and Flame} \bibinfo{volume}{106}
  (\bibinfo{year}{1996}) \bibinfo{pages}{487--510}.
%Type = Article
\bibitem[{Dowling(1999)}]{Dowling1999}
\bibinfo{author}{A.~P. Dowling},
\newblock \bibinfo{title}{A kinematic model of a ducted flame},
\newblock \bibinfo{journal}{Journal of Fluid Mechanics} \bibinfo{volume}{394}
  (\bibinfo{year}{1999}) \bibinfo{pages}{51--72}.
%Type = Article
\bibitem[{Kashinath et~al.(2014)Kashinath, Waugh, and Juniper}]{Kashinath2014}
\bibinfo{author}{K.~Kashinath}, \bibinfo{author}{I.~C. Waugh},
  \bibinfo{author}{M.~P. Juniper},
\newblock \bibinfo{title}{Nonlinear self-excited thermoacoustic oscillations of
  a ducted premixed flame: bifurcations and routes to chaos},
\newblock \bibinfo{journal}{Journal of Fluid Mechanics} \bibinfo{volume}{761}
  (\bibinfo{year}{2014}) \bibinfo{pages}{399--430}.
%Type = Article
\bibitem[{Waugh et~al.(2014)Waugh, Kashinath, and Juniper}]{Waugh2014}
\bibinfo{author}{I.~C. Waugh}, \bibinfo{author}{K.~Kashinath},
  \bibinfo{author}{M.~P. Juniper},
\newblock \bibinfo{title}{Matrix-free continuation of limit cycles and their
  bifurcations for a ducted premixed flame},
\newblock \bibinfo{journal}{Journal of Fluid Mechanics} \bibinfo{volume}{759}
  (\bibinfo{year}{2014}) \bibinfo{pages}{1--27}.
%Type = Article
\bibitem[{Zhao et~al.(1996)Zhao, Chan, Merriman, and Osher}]{Zhao1996}
\bibinfo{author}{H.-K. Zhao}, \bibinfo{author}{T.~Chan},
  \bibinfo{author}{B.~Merriman}, \bibinfo{author}{S.~Osher},
\newblock \bibinfo{title}{A {Variational} {Level} {Set} {Approach} to
  {Multiphase} {Motion}},
\newblock \bibinfo{journal}{Journal of Computational Physics}
  \bibinfo{volume}{127} (\bibinfo{year}{1996}) \bibinfo{pages}{179--195}.
%Type = Article
\bibitem[{Osher and Santosa(2001)}]{Osher2001a}
\bibinfo{author}{S.~J. Osher}, \bibinfo{author}{F.~Santosa},
\newblock \bibinfo{title}{Level {Set} {Methods} for {Optimization} {Problems}
  {Involving} {Geometry} and {Constraints}},
\newblock \bibinfo{journal}{Journal of Computational Physics}
  \bibinfo{volume}{171} (\bibinfo{year}{2001}) \bibinfo{pages}{272--288}.
%Type = Article
\bibitem[{Sethian and Wiegmann(2000)}]{Sethian2000a}
\bibinfo{author}{J.~A. Sethian}, \bibinfo{author}{A.~Wiegmann},
\newblock \bibinfo{title}{Structural {Boundary} {Design} via {Level} {Set} and
  {Immersed} {Interface} {Methods}},
\newblock \bibinfo{journal}{Journal of Computational Physics}
  \bibinfo{volume}{163} (\bibinfo{year}{2000}) \bibinfo{pages}{489--528}.
%Type = Article
\bibitem[{Allaire et~al.(2004)Allaire, Jouve, and Toader}]{Allaire2004}
\bibinfo{author}{G.~Allaire}, \bibinfo{author}{F.~Jouve},
  \bibinfo{author}{A.-M. Toader},
\newblock \bibinfo{title}{Structural optimization using sensitivity analysis
  and a level-set method},
\newblock \bibinfo{journal}{Journal of Computational Physics}
  \bibinfo{volume}{194} (\bibinfo{year}{2004}) \bibinfo{pages}{363--393}.
%Type = Article
\bibitem[{Chantalat et~al.(2009)Chantalat, Bruneau, Galusinski, and
  Iollo}]{Chantalat2009}
\bibinfo{author}{F.~Chantalat}, \bibinfo{author}{C.-H. Bruneau},
  \bibinfo{author}{C.~Galusinski}, \bibinfo{author}{A.~Iollo},
\newblock \bibinfo{title}{Level-set, penalization and cartesian meshes: {A}
  paradigm for inverse problems and optimal design},
\newblock \bibinfo{journal}{Journal of Computational Physics}
  \bibinfo{volume}{228} (\bibinfo{year}{2009}) \bibinfo{pages}{6291--6315}.
%Type = Book
\bibitem[{Evensen(2009)}]{Evensen2009}
\bibinfo{author}{G.~Evensen}, \bibinfo{title}{Data {Assimilation}},
  \bibinfo{publisher}{Springer Berlin Heidelberg}, \bibinfo{year}{2009}.
  \URLprefix \url{http://link.springer.com/10.1007/978-3-642-03711-5}.
  \DOIprefix\doi{10.1007/978-3-642-03711-5}.
%Type = Book
\bibitem[{Jazwinski(2007)}]{Jazwinski2007}
\bibinfo{author}{A.~H. Jazwinski}, \bibinfo{title}{Stochastic {Processes} and
  {Filtering} {Theory}}, \bibinfo{publisher}{Dover Publications},
  \bibinfo{year}{2007}.
%Type = Book
\bibitem[{Gelb(1974)}]{Gelb1974}
\bibinfo{editor}{A.~Gelb} (Ed.), \bibinfo{title}{Applied optimal estimation},
  \bibinfo{publisher}{MIT Press}, \bibinfo{year}{1974}.
%Type = Book
\bibitem[{Stengel(1994)}]{Stengel1994}
\bibinfo{author}{R.~F. Stengel}, \bibinfo{title}{Optimal control and
  estimation}, \bibinfo{publisher}{Dover Publications}, \bibinfo{year}{1994}.
%Type = Article
\bibitem[{Evensen(1994)}]{Evensen1994}
\bibinfo{author}{G.~Evensen},
\newblock \bibinfo{title}{Sequential data assimilation with a nonlinear
  quasi-geostrophic model using {Monte} {Carlo} methods to forecast error
  statistics},
\newblock \bibinfo{journal}{Journal of Geophysical Research}
  \bibinfo{volume}{99} (\bibinfo{year}{1994}) \bibinfo{pages}{10143}.
%Type = Article
\bibitem[{Burgers et~al.(1998)Burgers, van Leeuwen, and Evensen}]{Burgers1998}
\bibinfo{author}{G.~Burgers}, \bibinfo{author}{P.~J. van Leeuwen},
  \bibinfo{author}{G.~Evensen},
\newblock \bibinfo{title}{Analysis scheme in the ensemble {Kalman} filter},
\newblock \bibinfo{journal}{Monthly Weather Review} \bibinfo{volume}{126}
  (\bibinfo{year}{1998}) \bibinfo{pages}{1719--1724}.
%Type = Book
\bibitem[{Doucet et~al.(2001)Doucet, Freitas, and Gordon}]{Doucet2001}
\bibinfo{editor}{A.~Doucet}, \bibinfo{editor}{N.~Freitas},
  \bibinfo{editor}{N.~Gordon} (Eds.), \bibinfo{title}{Sequential {Monte}
  {Carlo} {Methods} in {Practice}}, \bibinfo{publisher}{Springer New York},
  \bibinfo{year}{2001}. \URLprefix
  \url{http://link.springer.com/10.1007/978-1-4757-3437-9}, \bibinfo{note}{doi:
  10.1007/978-1-4757-3437-9}.
%Type = Article
\bibitem[{Miller et~al.(1994)Miller, Ghil, and Gauthiez}]{Miller1994}
\bibinfo{author}{R.~N. Miller}, \bibinfo{author}{M.~Ghil},
  \bibinfo{author}{F.~Gauthiez},
\newblock \bibinfo{title}{Advanced data assimilation in strongly nonlinear
  dynamical systems},
\newblock \bibinfo{journal}{Journal of the Atmospheric Sciences}
  \bibinfo{volume}{51} (\bibinfo{year}{1994}) \bibinfo{pages}{1037--1056}.
%Type = Article
\bibitem[{Rozier et~al.(2007)Rozier, Birol, Cosme, Brasseur, Brankart, and
  Verron}]{Rozier2007}
\bibinfo{author}{D.~Rozier}, \bibinfo{author}{F.~Birol},
  \bibinfo{author}{E.~Cosme}, \bibinfo{author}{P.~Brasseur},
  \bibinfo{author}{J.~M. Brankart}, \bibinfo{author}{J.~Verron},
\newblock \bibinfo{title}{A {Reduced}-{Order} {Kalman} {Filter} for {Data}
  {Assimilation} in {Physical} {Oceanography}},
\newblock \bibinfo{journal}{SIAM Review} \bibinfo{volume}{49}
  (\bibinfo{year}{2007}) \bibinfo{pages}{449--465}.
%Type = Article
\bibitem[{Colburn et~al.(2011)Colburn, Cessna, and Bewley}]{Colburn2011}
\bibinfo{author}{C.~H. Colburn}, \bibinfo{author}{J.~B. Cessna},
  \bibinfo{author}{T.~R. Bewley},
\newblock \bibinfo{title}{State estimation in wall-bounded flow systems. {Part}
  3. {The} ensemble {Kalman} filter},
\newblock \bibinfo{journal}{Journal of Fluid Mechanics} \bibinfo{volume}{682}
  (\bibinfo{year}{2011}) \bibinfo{pages}{289--303}.
%Type = Article
\bibitem[{Kato et~al.(2015)Kato, Yoshizawa, Ueno, and Obayashi}]{Kato2015}
\bibinfo{author}{H.~Kato}, \bibinfo{author}{A.~Yoshizawa},
  \bibinfo{author}{G.~Ueno}, \bibinfo{author}{S.~Obayashi},
\newblock \bibinfo{title}{A data assimilation methodology for reconstructing
  turbulent flows around aircraft},
\newblock \bibinfo{journal}{Journal of Computational Physics}
  \bibinfo{volume}{283} (\bibinfo{year}{2015}) \bibinfo{pages}{559--581}.
%Type = Article
\bibitem[{Mons et~al.(2016)Mons, Chassaing, Gomez, and Sagaut}]{Mons2016}
\bibinfo{author}{V.~Mons}, \bibinfo{author}{J.-C. Chassaing},
  \bibinfo{author}{T.~Gomez}, \bibinfo{author}{P.~Sagaut},
\newblock \bibinfo{title}{Reconstruction of unsteady viscous flows using data
  assimilation schemes},
\newblock \bibinfo{journal}{Journal of Computational Physics}
  \bibinfo{volume}{316} (\bibinfo{year}{2016}) \bibinfo{pages}{255--280}.
%Type = Article
\bibitem[{Xiao et~al.(2016)Xiao, Wu, Wang, Sun, and Roy}]{Xiao2016}
\bibinfo{author}{H.~Xiao}, \bibinfo{author}{J.-L. Wu}, \bibinfo{author}{J.-X.
  Wang}, \bibinfo{author}{R.~Sun}, \bibinfo{author}{C.~Roy},
\newblock \bibinfo{title}{Quantifying and reducing model-form uncertainties in
  {Reynolds}-averaged {Navier}-{Stokes} simulations: {A} data-driven,
  physics-informed {Bayesian} approach},
\newblock \bibinfo{journal}{Journal of Computational Physics}
  \bibinfo{volume}{324} (\bibinfo{year}{2016}) \bibinfo{pages}{115--136}.
%Type = Article
\bibitem[{Darakananda et~al.(2018)Darakananda, da~Silva, Colonius, and
  Eldredge}]{Darakananda2018a}
\bibinfo{author}{D.~Darakananda}, \bibinfo{author}{A.~F. d.~C. da~Silva},
  \bibinfo{author}{T.~Colonius}, \bibinfo{author}{J.~D. Eldredge},
\newblock \bibinfo{title}{Data-assimilated low-order vortex modeling of
  separated flows},
\newblock \bibinfo{journal}{Physical Review Fluids} \bibinfo{volume}{3}
  (\bibinfo{year}{2018}).
%Type = Article
\bibitem[{Labahn et~al.(2018)Labahn, Wu, Coriton, Frank, and Ihme}]{Labahn2018}
\bibinfo{author}{J.~W. Labahn}, \bibinfo{author}{H.~Wu},
  \bibinfo{author}{B.~Coriton}, \bibinfo{author}{J.~H. Frank},
  \bibinfo{author}{M.~Ihme},
\newblock \bibinfo{title}{Data assimilation using high-speed measurements and
  {LES} to examine local extinction events in turbulent flames},
\newblock \bibinfo{journal}{Proceedings of the Combustion Institute}
  (\bibinfo{year}{2018}).
%Type = Article
\bibitem[{Beezley and Mandel(2008)}]{Beezley2008}
\bibinfo{author}{J.~D. Beezley}, \bibinfo{author}{J.~Mandel},
\newblock \bibinfo{title}{Morphing ensemble {Kalman} filters},
\newblock \bibinfo{journal}{Tellus A: Dynamic Meteorology and Oceanography}
  \bibinfo{volume}{60} (\bibinfo{year}{2008}) \bibinfo{pages}{131--140}.
%Type = Article
\bibitem[{Mandel et~al.(2008)Mandel, Bennethum, Beezley, Coen, Douglas, Kim,
  and Vodacek}]{Mandel2008}
\bibinfo{author}{J.~Mandel}, \bibinfo{author}{L.~S. Bennethum},
  \bibinfo{author}{J.~D. Beezley}, \bibinfo{author}{J.~L. Coen},
  \bibinfo{author}{C.~C. Douglas}, \bibinfo{author}{M.~Kim},
  \bibinfo{author}{A.~Vodacek},
\newblock \bibinfo{title}{A wildland fire model with data assimilation},
\newblock \bibinfo{journal}{Mathematics and Computers in Simulation}
  \bibinfo{volume}{79} (\bibinfo{year}{2008}) \bibinfo{pages}{584--606}.
%Type = Article
\bibitem[{Mandel et~al.(2011)Mandel, Beezley, and Kochanski}]{Mandel2011}
\bibinfo{author}{J.~Mandel}, \bibinfo{author}{J.~D. Beezley},
  \bibinfo{author}{A.~K. Kochanski},
\newblock \bibinfo{title}{Coupled atmosphere-wildland fire modeling with
  {WRF}-{Fire}},
\newblock \bibinfo{journal}{Geoscientific Model Development}
  \bibinfo{volume}{4} (\bibinfo{year}{2011}) \bibinfo{pages}{591--610}.
  \bibinfo{note}{ArXiv: 1102.1343}.
%Type = Article
\bibitem[{Rochoux et~al.(2018)Rochoux, Collin, Zhang, Trouv{\'e}, Lucor, and
  Moireau}]{Rochoux2018}
\bibinfo{author}{M.~Rochoux}, \bibinfo{author}{A.~Collin},
  \bibinfo{author}{C.~Zhang}, \bibinfo{author}{A.~Trouv{\'e}},
  \bibinfo{author}{D.~Lucor}, \bibinfo{author}{P.~Moireau},
\newblock \bibinfo{title}{Front {Shape} {Similarity} {Measure} for
  {Shape}-{Oriented} {Sensitivity} {Analysis} and {Data} {Assimilation} for
  {Eikonal} {Equation}},
\newblock \bibinfo{journal}{ESAIM: Proceedings and Surveys}
  \bibinfo{volume}{63} (\bibinfo{year}{2018}) \bibinfo{pages}{258--279}.
%Type = Article
\bibitem[{Zhang et~al.(2019)Zhang, Collin, Moireau, Trouv{\'e}, and
  Rochoux}]{Zhang2019}
\bibinfo{author}{C.~Zhang}, \bibinfo{author}{A.~Collin},
  \bibinfo{author}{P.~Moireau}, \bibinfo{author}{A.~Trouv{\'e}},
  \bibinfo{author}{M.~Rochoux},
\newblock \bibinfo{title}{Front shape similarity measure for data-driven
  simulations of wildland fire spread based on state estimation: {Application}
  to the {RxCADRE} field-scale experiment},
\newblock \bibinfo{journal}{Proceedings of the Combustion Institute}
  \bibinfo{volume}{37} (\bibinfo{year}{2019}) \bibinfo{pages}{4201--4209}.
%Type = Article
\bibitem[{Li et~al.(2017)Li, Le~Dimet, Ma, and Vidard}]{Li2017}
\bibinfo{author}{L.~Li}, \bibinfo{author}{F.-X. Le~Dimet},
  \bibinfo{author}{J.~Ma}, \bibinfo{author}{A.~Vidard},
\newblock \bibinfo{title}{A {Level}-{Set}-{Based} {Image} {Assimilation}
  {Method}: {Potential} {Applications} for {Predicting} the {Movement} of {Oil}
  {Spills}},
\newblock \bibinfo{journal}{IEEE Transactions on Geoscience and Remote Sensing}
  \bibinfo{volume}{55} (\bibinfo{year}{2017}) \bibinfo{pages}{6330--6343}.
%Type = Article
\bibitem[{Li et~al.(2019)Li, Vidard, Le~Dimet, and Ma}]{Li2019}
\bibinfo{author}{L.~Li}, \bibinfo{author}{A.~Vidard}, \bibinfo{author}{F.-X.
  Le~Dimet}, \bibinfo{author}{J.~Ma},
\newblock \bibinfo{title}{Topological data assimilation using {Wasserstein}
  distance},
\newblock \bibinfo{journal}{Inverse Problems} \bibinfo{volume}{35}
  (\bibinfo{year}{2019}) \bibinfo{pages}{015006}.
%Type = Article
\bibitem[{Peng et~al.(1999)Peng, Merriman, Osher, Zhao, and Kang}]{Peng1999}
\bibinfo{author}{D.~Peng}, \bibinfo{author}{B.~Merriman},
  \bibinfo{author}{S.~Osher}, \bibinfo{author}{H.~Zhao},
  \bibinfo{author}{M.~Kang},
\newblock \bibinfo{title}{A {PDE}-{Based} {Fast} {Local} {Level} {Set}
  {Method}},
\newblock \bibinfo{journal}{Journal of Computational Physics}
  \bibinfo{volume}{155} (\bibinfo{year}{1999}) \bibinfo{pages}{410--438}.
%Type = Article
\bibitem[{Sethian(1996)}]{Sethian1996}
\bibinfo{author}{J.~A. Sethian},
\newblock \bibinfo{title}{A fast marching level set method for monotonically
  advancing fronts},
\newblock \bibinfo{journal}{Proceedings of the National Academy of Sciences of
  the United States of America} \bibinfo{volume}{93} (\bibinfo{year}{1996})
  \bibinfo{pages}{1591--1595}.
%Type = Inproceedings
\bibitem[{Moreno and Aanonsen(2007)}]{Moreno2007}
\bibinfo{author}{D.~L. Moreno}, \bibinfo{author}{S.~I. Aanonsen},
\newblock \bibinfo{title}{Stochastic {Facies} {Modelling} {Using} the {Level}
  {Set} {Method}},
\newblock in: \bibinfo{booktitle}{Petroleum {Geostatistics}},
  \bibinfo{address}{Cascais, Portugal}, \bibinfo{year}{2007}. \URLprefix
  \url{http://www.earthdoc.org/publication/publicationdetails/?publication=7826}.
  \DOIprefix\doi{10.3997/2214-4609.201403056}.
%Type = Article
\bibitem[{Rochoux et~al.(2013)Rochoux, Cuenot, Ricci, Trouv{\'e}, Delmotte,
  Massart, Paoli, and Paugam}]{Rochoux2013}
\bibinfo{author}{M.~C. Rochoux}, \bibinfo{author}{B.~Cuenot},
  \bibinfo{author}{S.~Ricci}, \bibinfo{author}{A.~Trouv{\'e}},
  \bibinfo{author}{B.~Delmotte}, \bibinfo{author}{S.~Massart},
  \bibinfo{author}{R.~Paoli}, \bibinfo{author}{R.~Paugam},
\newblock \bibinfo{title}{Data assimilation applied to combustion},
\newblock \bibinfo{journal}{Comptes Rendus M{\'e}canique} \bibinfo{volume}{341}
  (\bibinfo{year}{2013}) \bibinfo{pages}{266--276}.
%Type = Article
\bibitem[{Gao et~al.(2017)Gao, Wang, Overton, Zupanski, and Tu}]{Gao2017}
\bibinfo{author}{X.~Gao}, \bibinfo{author}{Y.~Wang},
  \bibinfo{author}{N.~Overton}, \bibinfo{author}{M.~Zupanski},
  \bibinfo{author}{X.~Tu},
\newblock \bibinfo{title}{Data-assimilated computational fluid dynamics
  modeling of convection-diffusion-reaction problems},
\newblock \bibinfo{journal}{Journal of Computational Science}
  \bibinfo{volume}{21} (\bibinfo{year}{2017}) \bibinfo{pages}{38--59}.
%Type = Book
\bibitem[{Boyd and Vandenberghe(2004)}]{Boyd2004}
\bibinfo{author}{S.~P. Boyd}, \bibinfo{author}{L.~Vandenberghe},
  \bibinfo{title}{Convex {Optimization}}, \bibinfo{publisher}{Cambridge
  University Press}, \bibinfo{year}{2004}.
%Type = Book
\bibitem[{Jaynes(2003)}]{Jaynes2003}
\bibinfo{author}{E.~T. Jaynes}, \bibinfo{title}{Probability {Theory}: {The}
  {Logic} of {Science}}, \bibinfo{publisher}{Cambridge University Press},
  \bibinfo{address}{Cambridge}, \bibinfo{year}{2003}. \URLprefix
  \url{http://ebooks.cambridge.org/ref/id/CBO9780511790423}.
  \DOIprefix\doi{10.1017/CBO9780511790423}.
%Type = Book
\bibitem[{Sarkka(2013)}]{Sarkka2013}
\bibinfo{author}{S.~Sarkka}, \bibinfo{title}{Bayesian {Filtering} and
  {Smoothing}}, \bibinfo{publisher}{Cambridge University Press},
  \bibinfo{address}{Cambridge}, \bibinfo{year}{2013}. \URLprefix
  \url{http://ebooks.cambridge.org/ref/id/CBO9781139344203}.
  \DOIprefix\doi{10.1017/CBO9781139344203}.
%Type = Book
\bibitem[{Bellman(2003)}]{Bellman2003}
\bibinfo{author}{R.~Bellman}, \bibinfo{title}{Dynamic programming},
  \bibinfo{publisher}{Dover Publications}, \bibinfo{year}{2003}.
  \bibinfo{note}{OCLC: 834140800}.
%Type = Article
\bibitem[{Higham(2001)}]{Higham2001}
\bibinfo{author}{D.~J. Higham},
\newblock \bibinfo{title}{An {Algorithmic} {Introduction} to {Numerical}
  {Simulation} of {Stochastic} {Differential} {Equations}},
\newblock \bibinfo{journal}{SIAM Review} \bibinfo{volume}{43}
  (\bibinfo{year}{2001}) \bibinfo{pages}{525--546}.
%Type = Article
\bibitem[{Kalman(1960)}]{Kalman1960}
\bibinfo{author}{R.~E. Kalman},
\newblock \bibinfo{title}{A {New} {Approach} to {Linear} {Filtering} and
  {Prediction} {Problems}},
\newblock \bibinfo{journal}{Journal of Basic Engineering} \bibinfo{volume}{82}
  (\bibinfo{year}{1960}) \bibinfo{pages}{35}.
%Type = Article
\bibitem[{Kalman and Bucy(1961)}]{Kalman1961}
\bibinfo{author}{R.~E. Kalman}, \bibinfo{author}{R.~S. Bucy},
\newblock \bibinfo{title}{New {Results} in {Linear} {Filtering} and
  {Prediction} {Theory}},
\newblock \bibinfo{journal}{Journal of Basic Engineering} \bibinfo{volume}{83}
  (\bibinfo{year}{1961}) \bibinfo{pages}{95}.
%Type = Article
\bibitem[{Whitaker and Hamill(2002)}]{Whitaker2002}
\bibinfo{author}{J.~S. Whitaker}, \bibinfo{author}{T.~M. Hamill},
\newblock \bibinfo{title}{Ensemble {Data} {Assimilation} without {Perturbed}
  {Observations}},
\newblock \bibinfo{journal}{Monthly Weather Review} \bibinfo{volume}{130}
  (\bibinfo{year}{2002}) \bibinfo{pages}{1913--1924}.
%Type = Article
\bibitem[{Tippett et~al.(2003)Tippett, Anderson, Bishop, Hamill, and
  Whitaker}]{Tippett2003}
\bibinfo{author}{M.~K. Tippett}, \bibinfo{author}{J.~L. Anderson},
  \bibinfo{author}{C.~H. Bishop}, \bibinfo{author}{T.~M. Hamill},
  \bibinfo{author}{J.~S. Whitaker},
\newblock \bibinfo{title}{Ensemble {Square} {Root} {Filters}},
\newblock \bibinfo{journal}{Monthly Weather Review} \bibinfo{volume}{131}
  (\bibinfo{year}{2003}) \bibinfo{pages}{1485--1490}.
%Type = Article
\bibitem[{Livings et~al.(2008)Livings, Dance, and Nichols}]{Livings2008}
\bibinfo{author}{D.~M. Livings}, \bibinfo{author}{S.~L. Dance},
  \bibinfo{author}{N.~K. Nichols},
\newblock \bibinfo{title}{Unbiased ensemble square root filters},
\newblock \bibinfo{journal}{Physica D: Nonlinear Phenomena}
  \bibinfo{volume}{237} (\bibinfo{year}{2008}) \bibinfo{pages}{1021--1028}.
%Type = Book
\bibitem[{Arnold(1989)}]{Arnold1989}
\bibinfo{author}{V.~I. Arnold}, \bibinfo{title}{Mathematical {Methods} of
  {Classical} {Mechanics}}, volume~\bibinfo{volume}{60} of
  \textit{\bibinfo{series}{Graduate {Texts} in {Mathematics}}},
  \bibinfo{edition}{second} ed., \bibinfo{publisher}{Springer New York},
  \bibinfo{year}{1989}. \URLprefix
  \url{http://link.springer.com/10.1007/978-1-4757-2063-1}.
  \DOIprefix\doi{10.1007/978-1-4757-2063-1}.
%Type = Article
\bibitem[{Sethian(1999)}]{Sethian1999}
\bibinfo{author}{J.~A. Sethian},
\newblock \bibinfo{title}{Fast {Marching} {Methods}},
\newblock \bibinfo{journal}{SIAM Review} \bibinfo{volume}{41}
  (\bibinfo{year}{1999}) \bibinfo{pages}{199--235}.
%Type = Article
\bibitem[{Adalsteinsson and Sethian(1995)}]{Adalsteinsson1995}
\bibinfo{author}{D.~Adalsteinsson}, \bibinfo{author}{J.~A. Sethian},
\newblock \bibinfo{title}{A {Fast} {Level} {Set} {Method} for {Propagating}
  {Interfaces}},
\newblock \bibinfo{journal}{Journal of Computational Physics}
  \bibinfo{volume}{118} (\bibinfo{year}{1995}) \bibinfo{pages}{269--277}.
%Type = Article
\bibitem[{Sussman et~al.(1994)Sussman, Smereka, and Osher}]{Sussman1994}
\bibinfo{author}{M.~Sussman}, \bibinfo{author}{P.~Smereka},
  \bibinfo{author}{S.~Osher},
\newblock \bibinfo{title}{A {Level} {Set} {Approach} for {Computing}
  {Solutions} to {Incompressible} {Two}-{Phase} {Flow}},
\newblock \bibinfo{journal}{Journal of Computational Physics}
  \bibinfo{volume}{114} (\bibinfo{year}{1994}) \bibinfo{pages}{146--159}.
%Type = Book
\bibitem[{Ascher et~al.(1995)Ascher, Mattheij, and Russell}]{Ascher1995}
\bibinfo{author}{U.~M. Ascher}, \bibinfo{author}{R.~M.~M. Mattheij},
  \bibinfo{author}{R.~D. Russell}, \bibinfo{title}{Numerical {Solution} of
  {Boundary} {Value} {Problems} for {Ordinary} {Differential} {Equations}},
  \bibinfo{publisher}{Society for Industrial and Applied Mathematics},
  \bibinfo{year}{1995}. \URLprefix
  \url{http://epubs.siam.org/doi/book/10.1137/1.9781611971231}.
%Type = Book
\bibitem[{MacKay(2003)}]{MacKay2003}
\bibinfo{author}{D.~J.~C. MacKay}, \bibinfo{title}{Information {Theory},
  {Inference}, and {Learning} {Algorithms}}, \bibinfo{publisher}{Cambridge
  University Press}, \bibinfo{year}{2003}.
%Type = Book
\bibitem[{Szego(1939)}]{Szegoe1939}
\bibinfo{author}{G.~Szego}, \bibinfo{title}{Orthogonal {Polynomials}},
  volume~\bibinfo{volume}{23} of \textit{\bibinfo{series}{Colloquium
  {Publications}}}, \bibinfo{publisher}{American Mathematical Society},
  \bibinfo{address}{Providence, Rhode Island}, \bibinfo{year}{1939}. \URLprefix
  \url{http://www.ams.org/coll/023}. \DOIprefix\doi{10.1090/coll/023}.
%Type = Book
\bibitem[{Lieuwen and Yang(2005)}]{Lieuwen2005}
\bibinfo{editor}{T.~C. Lieuwen}, \bibinfo{editor}{V.~Yang} (Eds.),
  \bibinfo{title}{Combustion {Instabilities} in {Gas} {Turbine} {Engines}:
  {Operational} {Experience}, {Fundamental} {Mechanisms} and {Modeling}},
  Progress in astronautics and aeronautics, \bibinfo{publisher}{American
  Institute of Aeronautics and Astronautics}, \bibinfo{year}{2005}.
  \bibinfo{note}{OCLC: ocm62221288}.
%Type = Techreport
\bibitem[{Culick(2006)}]{Culick2006}
\bibinfo{author}{F.~E.~C. Culick}, \bibinfo{title}{Unsteady motions in
  combustion chambers for propulsion systems}, \bibinfo{type}{Technical
  Report}, AGARD, \bibinfo{year}{2006}.
%Type = Article
\bibitem[{Orchini et~al.(2016)Orchini, Rigas, and Juniper}]{Orchini2016a}
\bibinfo{author}{A.~Orchini}, \bibinfo{author}{G.~Rigas},
  \bibinfo{author}{M.~P. Juniper},
\newblock \bibinfo{title}{Weakly nonlinear analysis of thermoacoustic
  bifurcations in the {Rijke} tube},
\newblock \bibinfo{journal}{Journal of Fluid Mechanics} \bibinfo{volume}{805}
  (\bibinfo{year}{2016}) \bibinfo{pages}{523--550}.
%Type = Article
\bibitem[{Semlitsch et~al.(2017)Semlitsch, Orchini, Dowling, and
  Juniper}]{Semlitsch2017}
\bibinfo{author}{B.~Semlitsch}, \bibinfo{author}{A.~Orchini},
  \bibinfo{author}{A.~P. Dowling}, \bibinfo{author}{M.~P. Juniper},
\newblock \bibinfo{title}{G-equation modelling of thermoacoustic oscillations
  of partially premixed flames},
\newblock \bibinfo{journal}{International Journal of Spray and Combustion
  Dynamics} \bibinfo{volume}{9} (\bibinfo{year}{2017})
  \bibinfo{pages}{260--276}.
%Type = Book
\bibitem[{Reich and Cotter(2015)}]{Reich2015}
\bibinfo{author}{S.~Reich}, \bibinfo{author}{C.~Cotter},
  \bibinfo{title}{Probabilistic {Forecasting} and {Bayesian} {Data}
  {Assimilation}}, \bibinfo{publisher}{Cambridge University Press},
  \bibinfo{address}{Cambridge}, \bibinfo{year}{2015}. \URLprefix
  \url{http://ebooks.cambridge.org/ref/id/CBO9781107706804}.
  \DOIprefix\doi{10.1017/CBO9781107706804}.
%Type = Techreport
\bibitem[{Yu et~al.(2018)Yu, Jaravel, Ihme, Juniper, and Magri}]{Yu2018}
\bibinfo{author}{H.~Yu}, \bibinfo{author}{T.~Jaravel},
  \bibinfo{author}{M.~Ihme}, \bibinfo{author}{M.~P. Juniper},
  \bibinfo{author}{L.~Magri}, \bibinfo{title}{Physics-informed data-driven
  prediction of premixed flame dynamics with data assimilation},
  \bibinfo{type}{Technical Report}, Center for Turbulence Research,
  \bibinfo{year}{2018}. \URLprefix
  \url{https://ctr.stanford.edu/proceedings-2018-summer-program}.
%Type = Article
\bibitem[{Yu et~al.(2019)Yu, Jaravel, Juniper, Ihme, and Magri}]{Yu2019_asme}
\bibinfo{author}{H.~Yu}, \bibinfo{author}{T.~Jaravel},
  \bibinfo{author}{M.~Juniper}, \bibinfo{author}{M.~Ihme},
  \bibinfo{author}{L.~Magri},
\newblock \bibinfo{title}{{Data assimilation and optimal calibration in
  nonlinear models of flame dynamics}},
\newblock \bibinfo{journal}{Journal of Engineering for Gas Turbines and Power,
  doi:10.17863/CAM.41434}  (\bibinfo{year}{2019}).
%Type = Article
\bibitem[{Kashinath et~al.(2013)Kashinath, Hemchandra, and
  Juniper}]{Kashinath2013}
\bibinfo{author}{K.~Kashinath}, \bibinfo{author}{S.~Hemchandra},
  \bibinfo{author}{M.~P. Juniper},
\newblock \bibinfo{title}{Nonlinear thermoacoustics of ducted premixed flames:
  {The} influence of perturbation convection speed},
\newblock \bibinfo{journal}{Combustion and Flame} \bibinfo{volume}{160}
  (\bibinfo{year}{2013}) \bibinfo{pages}{2856--2865}.
%Type = Article
\bibitem[{Liu et~al.(1994)Liu, Osher, and Chan}]{Liu1994}
\bibinfo{author}{X.-D. Liu}, \bibinfo{author}{S.~Osher},
  \bibinfo{author}{T.~Chan},
\newblock \bibinfo{title}{Weighted {Essentially} {Non}-oscillatory {Schemes}},
\newblock \bibinfo{journal}{Journal of Computational Physics}
  \bibinfo{volume}{115} (\bibinfo{year}{1994}) \bibinfo{pages}{200--212}.
%Type = Article
\bibitem[{Jiang and Shu(1996)}]{Jiang1996}
\bibinfo{author}{G.-S. Jiang}, \bibinfo{author}{C.-W. Shu},
\newblock \bibinfo{title}{Efficient {Implementation} of {Weighted} {ENO}
  {Schemes}},
\newblock \bibinfo{journal}{Journal of Computational Physics}
  \bibinfo{volume}{126} (\bibinfo{year}{1996}) \bibinfo{pages}{202--228}.
%Type = Article
\bibitem[{Shu and Osher(1988)}]{Shu1988}
\bibinfo{author}{C.-W. Shu}, \bibinfo{author}{S.~Osher},
\newblock \bibinfo{title}{Efficient implementation of essentially
  non-oscillatory shock-capturing schemes},
\newblock \bibinfo{journal}{Journal of Computational Physics}
  \bibinfo{volume}{77} (\bibinfo{year}{1988}) \bibinfo{pages}{439--471}.
%Type = Phdthesis
\bibitem[{Waugh(2013)}]{Waugh2013}
\bibinfo{author}{I.~Waugh}, \bibinfo{title}{Methods for {Analysis} of
  {Nonlinear} {Thermoacoustic} {Systems}}, Ph.D. thesis, University of
  Cambridge, \bibinfo{year}{2013}. \URLprefix
  \url{https://doi.org/10.17863/CAM.36094}.
%Type = Article
\bibitem[{{Preetham} et~al.(2008){Preetham}, Santosh, and
  Lieuwen}]{Preetham2008}
\bibinfo{author}{{Preetham}}, \bibinfo{author}{H.~Santosh},
  \bibinfo{author}{T.~Lieuwen},
\newblock \bibinfo{title}{Dynamics of {Laminar} {Premixed} {Flames} {Forced} by
  {Harmonic} {Velocity} {Disturbances}},
\newblock \bibinfo{journal}{Journal of Propulsion and Power}
  \bibinfo{volume}{24} (\bibinfo{year}{2008}) \bibinfo{pages}{1390--1402}.
%Type = Article
\bibitem[{Orchini and Juniper(2016)}]{Orchini2016}
\bibinfo{author}{A.~Orchini}, \bibinfo{author}{M.~P. Juniper},
\newblock \bibinfo{title}{Linear stability and adjoint sensitivity analysis of
  thermoacoustic networks with premixed flames},
\newblock \bibinfo{journal}{Combustion and Flame} \bibinfo{volume}{165}
  (\bibinfo{year}{2016}) \bibinfo{pages}{97--108}.
%Type = Book
\bibitem[{Rasmussen and Williams(2006)}]{Rasmussen2006}
\bibinfo{author}{C.~E. Rasmussen}, \bibinfo{author}{C.~K.~I. Williams},
  \bibinfo{title}{Gaussian {Processes} for {Machine} {Learning}}, Adaptive
  {Computation} and {Machine} {Learning}, \bibinfo{publisher}{MIT Press},
  \bibinfo{year}{2006}. \bibinfo{note}{OCLC: ocm61285753}.
%Type = Article
\bibitem[{Juniper and Sujith(2018)}]{Juniper2018}
\bibinfo{author}{M.~P. Juniper}, \bibinfo{author}{R.~Sujith},
\newblock \bibinfo{title}{Sensitivity and {Nonlinearity} of {Thermoacoustic}
  {Oscillations}},
\newblock \bibinfo{journal}{Annual Review of Fluid Mechanics}
  \bibinfo{volume}{50} (\bibinfo{year}{2018}) \bibinfo{pages}{661--689}.
%Type = Article
\bibitem[{Magri et~al.(2016{\natexlab{a}})Magri, Bauerheim, and
  Juniper}]{Magri2016}
\bibinfo{author}{L.~Magri}, \bibinfo{author}{M.~Bauerheim},
  \bibinfo{author}{M.~P. Juniper},
\newblock \bibinfo{title}{Stability analysis of thermo-acoustic nonlinear
  eigenproblems in annular combustors. {Part} {I}. {Sensitivity}},
\newblock \bibinfo{journal}{Journal of Computational Physics}
  \bibinfo{volume}{325} (\bibinfo{year}{2016}{\natexlab{a}})
  \bibinfo{pages}{395--410}.
%Type = Article
\bibitem[{Magri et~al.(2016{\natexlab{b}})Magri, Bauerheim, Nicoud, and
  Juniper}]{Magri2016a}
\bibinfo{author}{L.~Magri}, \bibinfo{author}{M.~Bauerheim},
  \bibinfo{author}{F.~Nicoud}, \bibinfo{author}{M.~P. Juniper},
\newblock \bibinfo{title}{Stability analysis of thermo-acoustic nonlinear
  eigenproblems in annular combustors. {Part} {II}. {Uncertainty}
  quantification},
\newblock \bibinfo{journal}{Journal of Computational Physics}
  \bibinfo{volume}{325} (\bibinfo{year}{2016}{\natexlab{b}})
  \bibinfo{pages}{411--421}.
%Type = Article
\bibitem[{Magri(2019)}]{Magri2019}
\bibinfo{author}{L.~Magri},
\newblock \bibinfo{title}{Adjoint {Methods} as {Design} {Tools} in
  {Thermoacoustics}},
\newblock \bibinfo{journal}{Applied Mechanics Reviews} \bibinfo{volume}{71}
  (\bibinfo{year}{2019}) \bibinfo{pages}{020801}.
%Type = Article
\bibitem[{Juniper(2011)}]{Juniper2011}
\bibinfo{author}{M.~P. Juniper},
\newblock \bibinfo{title}{Triggering in the horizontal {Rijke} tube:
  non-normality, transient growth and bypass transition},
\newblock \bibinfo{journal}{Journal of Fluid Mechanics} \bibinfo{volume}{667}
  (\bibinfo{year}{2011}) \bibinfo{pages}{272--308}.
%Type = Article
\bibitem[{Kerswell(2018)}]{Kerswell2018}
\bibinfo{author}{R.~Kerswell},
\newblock \bibinfo{title}{Nonlinear {Nonmodal} {Stability} {Theory}},
\newblock \bibinfo{journal}{Annual Review of Fluid Mechanics}
  \bibinfo{volume}{50} (\bibinfo{year}{2018}) \bibinfo{pages}{319--345}.
%Type = Article
\bibitem[{Kiureghian and Ditlevsen(2009)}]{Kiureghian2009}
\bibinfo{author}{A.~D. Kiureghian}, \bibinfo{author}{O.~Ditlevsen},
\newblock \bibinfo{title}{Aleatory or epistemic? {Does} it matter?},
\newblock \bibinfo{journal}{Structural Safety} \bibinfo{volume}{31}
  (\bibinfo{year}{2009}) \bibinfo{pages}{105--112}.
%Type = Article
\bibitem[{Anderson and Anderson(1999)}]{Anderson1999}
\bibinfo{author}{J.~L. Anderson}, \bibinfo{author}{S.~L. Anderson},
\newblock \bibinfo{title}{A {Monte} {Carlo} {Implementation} of the {Nonlinear}
  {Filtering} {Problem} to {Produce} {Ensemble} {Assimilations} and
  {Forecasts}},
\newblock \bibinfo{journal}{Monthly Weather Review} \bibinfo{volume}{127}
  (\bibinfo{year}{1999}) \bibinfo{pages}{2741--2758}.
%Type = Article
\bibitem[{Ott et~al.(2004)Ott, Hunt, Szunyogh, Zimin, Kostelich, Corazza,
  Kalnay, Patil, and Yorke}]{Ott2004}
\bibinfo{author}{E.~Ott}, \bibinfo{author}{B.~R. Hunt},
  \bibinfo{author}{I.~Szunyogh}, \bibinfo{author}{A.~V. Zimin},
  \bibinfo{author}{E.~J. Kostelich}, \bibinfo{author}{M.~Corazza},
  \bibinfo{author}{E.~Kalnay}, \bibinfo{author}{D.~Patil},
  \bibinfo{author}{J.~A. Yorke},
\newblock \bibinfo{title}{A local ensemble {Kalman} filter for atmospheric data
  assimilation},
\newblock \bibinfo{journal}{Tellus A: Dynamic Meteorology and Oceanography}
  \bibinfo{volume}{56} (\bibinfo{year}{2004}) \bibinfo{pages}{415--428}.
%Type = Article
\bibitem[{Gaspari and Cohn(1999)}]{Gaspari1999}
\bibinfo{author}{G.~Gaspari}, \bibinfo{author}{S.~E. Cohn},
\newblock \bibinfo{title}{Construction of correlation functions in two and
  three dimensions},
\newblock \bibinfo{journal}{Quarterly Journal of the Royal Meteorological
  Society} \bibinfo{volume}{125} (\bibinfo{year}{1999})
  \bibinfo{pages}{723--757}.
%Type = Book
\bibitem[{Forster(2013)}]{Forster2013}
\bibinfo{author}{O.~Forster}, \bibinfo{title}{Analysis 1},
  \bibinfo{publisher}{Springer Fachmedien Wiesbaden}, \bibinfo{year}{2013}.
  \URLprefix \url{http://link.springer.com/10.1007/978-3-658-00317-3}.
  \DOIprefix\doi{10.1007/978-3-658-00317-3}.
%Type = Book
\bibitem[{Fischer(2014)}]{Fischer2014}
\bibinfo{author}{G.~Fischer}, \bibinfo{title}{Lineare {Algebra}},
  \bibinfo{publisher}{Springer Fachmedien Wiesbaden}, \bibinfo{year}{2014}.
  \URLprefix \url{http://link.springer.com/10.1007/978-3-658-03945-5}.
  \DOIprefix\doi{10.1007/978-3-658-03945-5}.

\end{thebibliography}

%------------------------------------------------------------------------------

\appendix

\section{Bayesian Inference}
\label{app:bayes}

\begin{theorem}[Marginalization]
For a joint probability distribution $p(X, Y)$, the marginal probability distributions $p(X)$ and $p(Y)$ are given by
\begin{equation}
p(X) = \int{p(X, Y)\,\ud{Y}} \quad ,
\end{equation}
\begin{equation}
p(Y) = \int{p(X, Y)\,\ud{X}} \quad .
\end{equation}
\end{theorem}

\begin{theorem}[Bayes' rule]
The conditional probability distributions $p(X \mid Y)$ and $p(Y \mid X)$ are related to the joint probability distribution via
\begin{equation}
p(X, Y) = p(X)p(Y \mid X) = p(Y)p(X \mid Y) \quad .
\end{equation}
The conditional probability distribution $p(X \mid Y)$ is sometimes referred to as the inverse of the conditional probability distribution $p(Y \mid X)$.
The inverse is given by Bayes' rule:
\begin{equation}
p(X \mid Y) = \frac{p(X, Y)}{p(Y)} = \frac{p(X)p(Y \mid X)}{p(Y)} \quad .
\end{equation}
If $p(X)$ and $p(Y \mid X)$ are given, $p(Y)$ can be computed through marginalization.
\end{theorem}

\section{Hamiltonian mechanics}
\label{app:hj}

\begin{lemma}[Lagrangian]
\label{thm:hj:lagrange}
\begin{equation}
\mathcal{L}(r(t), \dot{r}(t), t) = 0 \quad .
\label{eq:hj:lagrange}
\end{equation}
\end{lemma}

\begin{proof}
The generating function is formally the same as the action integral~\cite[Chapter~9]{Arnold1989}.
Thus, the Lagrangian~$\mathcal{L}(r, \dot{r}, t)$ is the total derivative of the generating function:
\begin{equation}
\mathcal{L}(r(t), \dot{r}(t), t)
= \tdif{t}G(r(t), t)
= \tdif{t}G(r(0), 0)
= 0 \quad .
\end{equation}
\end{proof}

\begin{lemma}[Hamiltonian]
\label{thm:hj:hamilton}
\begin{equation}
\mathcal{H}(r(t), n(t), t) = u(r(t)) \cdot n(t) \quad .
\label{eq:hj:hamilton}
\end{equation}
\end{lemma}

\begin{proof}
A Legendre transform of the Lagrangian~$\mathcal{L}(r, \dot{r}, t)$~(Eq.~\eqref{eq:hj:lagrange}) gives the Hamiltonian~$\mathcal{H}(r, n, t)$:
\begin{equation}
\mathcal{H}(r(t), n(t), t)
= \dot{r}(t) \cdot n(t) - \mathcal{L}(r(t), \dot{r}(t), t) = u(r(t)) \cdot n(t) \quad .
\end{equation}
It remains to be shown that~$\mathcal{H}(r, n, t)$ is indeed a Hamiltonian.
The constraints, Eq.~\eqref{eq:ls:hj:hamilton1} and~\eqref{eq:ls:hj:hamilton2}, are only functions in $r$ and $n$ respectively.
Hence, a suitable choice of Lagrange multipliers~$\lambda_r$ and~$\lambda_n$ gives
\begin{equation}
\pdiff{\mathcal{H}}{n}
= u + \lambda_n n
= u
= \tdiff{r}{t} \quad ,
\label{eq:hj:hamilton_r}
\end{equation}
\begin{equation}
\pdiff{\mathcal{H}}{r}
= n \cdot \nabla u + n \times (\nabla \times u) + \lambda_r\pdiff{G}{r}
= n \cdot \nabla u + n \times (\nabla \times u) - \left[n \cdot \left(n \cdot \nabla u\right)\right]n
= -\tdiff{n}{t} \quad .
\label{eq:hj:hamilton_n}
\end{equation}
Eq.~\eqref{eq:hj:hamilton_r} follows from the laws of motion~(Eq.~\eqref{eq:ls:hj:newton}).
Eq.~\eqref{eq:hj:hamilton_n} follows from tensor calculus:
\begin{equation}
\tdiff{n}{t}
= \pdiff{n}{t} + u \cdot \nabla n + \mu_n n \quad ,
\end{equation}
where~$\mu_n$ is the Lagrange multiplier for the constraint imposed by Eq.~\eqref{eq:ls:hj:hamilton2}.
The gradient of the total derivative of the generating function gives an expression for the partial derivative of the normal vector~(Lemma~\ref{thm:hj:lagrange}):
\begin{equation}
\nabla\left(\tdiff{G}{t}\right)
= \nabla\left(\pdiff{G}{t} + u\cdot\nabla G\right)
= \pdiff{n}{t} + \nabla\left(u \cdot n\right)
= 0 \quad .
\end{equation}
The dot product of the normal vector with the total derivative of the normal vector gives an expression for the Lagrange multiplier~$\mu_n$:
\begin{align}
n \cdot \tdiff{n}{t}
&= \frac{1}{2}\tdif{t}\left(n \cdot n\right)
 = 0 \quad , \\
n \cdot \tdiff{n}{t}
&= n \cdot \left(\pdiff{n}{t} + u \cdot \nabla n + \mu_n n\right) \\
&= n \cdot \pdiff{n}{t} + \frac{1}{2}\left(u \cdot \nabla\right)\left(n \cdot n\right) + \mu_n n \cdot n \\
& = -n \cdot \nabla\left(u \cdot n\right) + \mu_n \quad .
\end{align}
This gives:
\begin{align}
\tdiff{n}{t}
&= -\nabla\left(u \cdot n\right) + u \cdot \nabla n + \left[n \cdot \nabla\left(u \cdot n\right)\right]n \\
&= -u \cdot \nabla n - n \cdot \nabla u - u \times \left(\nabla \times n\right) - n \times \left(\nabla \times u\right) + u \cdot \nabla n + \left[\left(n \cdot \nabla u\right) \cdot n\right]n + \left[\left(n \cdot \nabla n\right) \cdot u\right]n \\
&= -n \cdot \nabla u - n \times \left(\nabla \times u\right) + n\left[n \cdot \left(n \cdot \nabla u\right)\right] \quad .
\end{align}
Note that $\nabla \times n = 0$ because $n = \nabla G$.
This concludes the proof that~$\mathcal{H}(r, n, t)$ is a Hamiltonian.
%Together, Eq.~\eqref{eq:hj:hamilton_r} and~\eqref{eq:hj:hamilton_n} give a complete description in phase space.
\end{proof}

\begin{theorem}[Hamilton-Jacobi equation]
\label{thm:hj:g}
\begin{equation}
\pdiff{G}{t} + u(r(t)) \cdot n(t) = 0 \quad .
\label{eq:hj:g}
\end{equation}
\end{theorem}

\begin{proof}
The Hamilton-Jacobi equation is given by~\cite[Chapter~9]{Arnold1989}
\begin{equation}
\pdiff{G}{t} + \mathcal{H}
= \pdiff{G}{t} + u \cdot n
= 0 \quad .
\end{equation}
\end{proof}

%------------------------------------------------------------------------------

\section{Level-set methods}
\label{app:ls}

\begin{lemma}
The covariance matrix of a gradient in a filtered generated function without reinitialization is given by
\begin{equation}
\mathbb{E}\left[D\psi^a\left(D\psi^a\right)^T\right] = \mathbb{E}\left[D\psi^f\left(D\psi^f\right)^T\right] - D\left(MC_{\psi\psi}^f\right)^T\left[C_{\epsilon\epsilon}+MC_{\psi\psi}^fM^T\right]^{-1}MC_{\psi\psi}^fD^T \quad ,
\end{equation}
where $D$ denotes the gradient operator.
\end{lemma}

\begin{proof}
Eq.~\eqref{eq:da:kalman:mean} gives
\begin{align}
D\psi^a &= D\psi^f + D\left(MC_{\psi\psi}^f\right)^T\left[C_{\epsilon\epsilon}+MC_{\psi\psi}^fM^T\right]^{-1}\left(y_k-M\psi^f\right) \quad , \\
D\psi^a\left(D\psi^a\right)^T
&= D\psi^f\left(D\psi^f\right)^T \nonumber\\
&\quad + D\psi^f\left(y_k-M\psi^f\right)^T\left[C_{\epsilon\epsilon}+MC_{\psi\psi}^fM^T\right]^{-1}MC_{\psi\psi}^fD^T \nonumber\\
&\quad + D\left(MC_{\psi\psi}^f\right)^T\left[C_{\epsilon\epsilon}+MC_{\psi\psi}^fM^T\right]^{-1}\left(y_k-M\psi^f\right)\left(D\psi^f\right)^T \nonumber\\
&\quad + D\left(MC_{\psi\psi}^f\right)^T\left[C_{\epsilon\epsilon}+MC_{\psi\psi}^fM^T\right]^{-1}\left(y_k-M\psi^f\right)\left(y_k-M\psi^f\right)^T\left[C_{\epsilon\epsilon}+MC_{\psi\psi}^fM^T\right]^{-1}MC_{\psi\psi}^fD^T \\
&= D\psi^f\left(D\psi^f\right)^T \nonumber\\
&\quad + D\psi^f\left(y_k-Mx_k\right)^T\left[C_{\epsilon\epsilon}+MC_{\psi\psi}^fM^T\right]^{-1}MC_{\psi\psi}^fD^T \nonumber\\
&\quad + D\psi^f\left(x_k-\psi^f\right)^TM^T\left[C_{\epsilon\epsilon}+MC_{\psi\psi}^fM^T\right]^{-1}MC_{\psi\psi}^fD^T \nonumber\\
&\quad + D\left(MC_{\psi\psi}^f\right)^T\left[C_{\epsilon\epsilon}+MC_{\psi\psi}^fM^T\right]^{-1}\left(y_k-Mx_k\right)\left(D\psi^f\right)^T \nonumber\\
&\quad + D\left(MC_{\psi\psi}^f\right)^T\left[C_{\epsilon\epsilon}+MC_{\psi\psi}^fM^T\right]^{-1}M\left(x_k-\psi^f\right)\left(D\psi^f\right)^T \nonumber\\
&\quad + D\left(MC_{\psi\psi}^f\right)^T\left[C_{\epsilon\epsilon}+MC_{\psi\psi}^fM^T\right]^{-1}\left(y_k-Mx_k\right)\left(y_k-Mx_k\right)^T\left[C_{\epsilon\epsilon}+MC_{\psi\psi}^fM^T\right]^{-1}MC_{\psi\psi}^fD^T \nonumber\\
&\quad + D\left(MC_{\psi\psi}^f\right)^T\left[C_{\epsilon\epsilon}+MC_{\psi\psi}^fM^T\right]^{-1}M\left(x_k-\psi^f\right)\left(x_k-\psi^f\right)^TM^T\left[C_{\epsilon\epsilon}+MC_{\psi\psi}^fM^T\right]^{-1}MC_{\psi\psi}^fD^T \nonumber\\
&\quad + D\left(MC_{\psi\psi}^f\right)^T\left[C_{\epsilon\epsilon}+MC_{\psi\psi}^fM^T\right]^{-1}\left(y_k-Mx_k\right)\left(x_k-\psi^f\right)^TM^T\left[C_{\epsilon\epsilon}+MC_{\psi\psi}^fM^T\right]^{-1}MC_{\psi\psi}^fD^T \nonumber\\
&\quad + D\left(MC_{\psi\psi}^f\right)^T\left[C_{\epsilon\epsilon}+MC_{\psi\psi}^fM^T\right]^{-1}M\left(x_k-\psi^f\right)\left(y_k-Mx_k\right)^T\left[C_{\epsilon\epsilon}+MC_{\psi\psi}^fM^T\right]^{-1}MC_{\psi\psi}^fD^T \quad .
\label{eq:ls:lemma1}
\end{align}
The probabilistic state space model (Eq.~\eqref{eq:da:state_space1}, \eqref{eq:da:state_space2}) gives
\begin{equation*}
\mathbb{E}\left[\left(y_k-Mx_k\right)\left(y_k-Mx_k\right)^T\right] = C_{\epsilon\epsilon} \quad , \quad
\mathbb{E}\left[\left(\psi^f-x_k\right)\left(\psi^f-x_k\right)^T\right] = C_{\psi\psi}^f \quad ,
\end{equation*}
\begin{equation}
\mathbb{E}\left[\left(\psi^f-x_k\right)\left(y_k-Mx_k\right)^T\right] = 0 \quad .
\label{eq:ls:lemma2}
\end{equation}
Substituting Eq.~\ref{eq:ls:lemma2} into Eq.~\ref{eq:ls:lemma1} gives
\begin{align}
\mathbb{E}\left[D\psi^a\left(D\psi^a\right)^T\right]
&= \mathbb{E}\left[D\psi^f\left(D\psi^f\right)^T\right] \nonumber\\
&\quad - D\left(MC_{\psi\psi}^f\right)^T\left[C_{\epsilon\epsilon}+MC_{\psi\psi}^fM^T\right]^{-1}MC_{\psi\psi}^fD^T \nonumber\\
&\quad - DC_{\psi\psi}^fM^T\left[C_{\epsilon\epsilon}+MC_{\psi\psi}^fM^T\right]^{-1}MC_{\psi\psi}^fD^T \nonumber\\
&\quad + D\left(MC_{\psi\psi}^f\right)^T\left[C_{\epsilon\epsilon}+MC_{\psi\psi}^fM^T\right]^{-1}C_{\epsilon\epsilon}\left[C_{\epsilon\epsilon}+MC_{\psi\psi}^fM^T\right]^{-1}MC_{\psi\psi}^fD^T \nonumber\\
&\quad + D\left(MC_{\psi\psi}^f\right)^T\left[C_{\epsilon\epsilon}+MC_{\psi\psi}^fM^T\right]^{-1}MC_{\psi\psi}^fM^T\left[C_{\epsilon\epsilon}+MC_{\psi\psi}^fM^T\right]^{-1}MC_{\psi\psi}^fD^T \\
&= \mathbb{E}\left[D\psi^f\left(D\psi^f\right)^T\right] - D\left(MC_{\psi\psi}^f\right)^T\left[C_{\epsilon\epsilon}+MC_{\psi\psi}^fM^T\right]^{-1}MC_{\psi\psi}^fD^T \quad .
\label{eq:ls:lemma3}
\end{align}
\end{proof}

\begin{lemma}
The expected value of a slope in a filtered generated function without reinitialization satisfies
\begin{equation}
\mathbb{E}\left[\left|D\psi^a\right|^2\right] = 1 - \mathrm{tr}\left(D\left(MC_{\psi\psi}^f\right)^T\left[C_{\epsilon\epsilon}+MC_{\psi\psi}^fM^T\right]^{-1}MC_{\psi\psi}^fD^T\right) \quad .
\label{eq:ls:theorem}
\end{equation}
\label{thm:ls}
\end{lemma}

\begin{proof}
In general, the expected value of a slope is related to the covariance matrix of the corresponding gradient in the following way:
\begin{equation}
\mathbb{E}\left[\left|D\psi\right|^2\right]
= \mathbb{E}\left[\mathrm{tr}\left(D\psi\left(D\psi\right)^T\right)\right]
= \mathrm{tr}\left(\mathbb{E}\left[D\psi\left(D\psi\right)^T\right]\right) \quad .
\label{eq:ls:theorem1}
\end{equation}
Substituting Eq.~\eqref{eq:ls:lemma3} into Eq.~\eqref{eq:ls:theorem1} gives
\begin{align}
\mathbb{E}\left[\left|D\psi^a\right|^2\right]
&= \mathrm{tr}\left(\mathbb{E}\left[D\psi^a\left(D\psi^a\right)^T\right]\right) \\
&= \mathrm{tr}\left(\mathbb{E}\left[D\psi^f\left(D\psi^f\right)^T\right]\right) - \mathrm{tr}\left(D\left(MC_{\psi\psi}^f\right)^T\left[C_{\epsilon\epsilon}+MC_{\psi\psi}^fM^T\right]^{-1}MC_{\psi\psi}^fD^T\right) \\
&= \mathbb{E}\left[\left|D\psi^f\right|^2\right] - \mathrm{tr}\left(D\left(MC_{\psi\psi}^f\right)^T\left[C_{\epsilon\epsilon}+MC_{\psi\psi}^fM^T\right]^{-1}MC_{\psi\psi}^fD^T\right) \quad .
\end{align}
In the first term, the slope $\left|D\psi^f\right|$ is equal to 1 due to Huygens' principle.
This gives Lemma~\ref{thm:ls}.
\end{proof}

\begin{theorem}
\begin{equation}
\mathbb{E}\left[\left|D\psi^a\right|\right] < 1 \quad .
\end{equation}
\label{corr:ls}
\end{theorem}

\begin{proof}
Applying the following two observations to Eq.~\eqref{eq:ls:theorem} gives Theorem~\ref{corr:ls}:
\begin{enumerate}
\item $\mathbb{E}\left[\left|D\psi^a\right|\right]^2 \leq \mathbb{E}\left[\left|D\psi^a\right|^2\right]$ due to Jensen's inequality \cite{Forster2013}.
\item $\mathrm{tr}\left(D\left(MC_{\psi\psi}^f\right)^T\left[C_{\epsilon\epsilon}+MC_{\psi\psi}^fM^T\right]^{-1}MC_{\psi\psi}^fD^T\right) \geq 0$ because $D\left(MC_{\psi\psi}^f\right)^T\left[C_{\epsilon\epsilon}+MC_{\psi\psi}^fM^T\right]^{-1}MC_{\psi\psi}^fD^T$ is symmetric and positive semi-definite \cite{Fischer2014}.
\end{enumerate}
\end{proof}

\end{document}